\documentclass[a4paper,10pt,notitlepage]{article}
\usepackage{amsmath}
\usepackage{amssymb}
\usepackage{xstring}
\usepackage[english]{babel}
\usepackage{hyperref}
\usepackage[amsthm,thref,hyperref,thmmarks]{ntheorem}
\usepackage{caption}
\usepackage{subcaption}
\usepackage[boxed]{algorithm2e}
\usepackage{xcolor} 
\usepackage{graphicx}
\usepackage[a4paper,left=20mm,right=20mm,top=20mm,bottom=20mm]{geometry}
\newcommand{\refP}[1]{%
	\def\InputString{#1}%
	\IfBeginWith{\InputString}{Equation}{%
		(\ref{#1})}{%
	\IfBeginWith{\InputString}{Section}{%
		Section \ref{#1}}{%
	\IfBeginWith{\InputString}{Subsection}{%
		Subsection \ref{#1}}{%
	\IfBeginWith{\InputString}{Chapter}{%
		Chapter \ref{#1}}{%
	\IfBeginWith{\InputString}{Subsubsection}{%
		Subsubsection \ref{#1}}{%
	\IfBeginWith{\InputString}{Problem}{%
		(\ref{#1})}{%
	\IfBeginWith{\InputString}{Property}{%
		property (\ref{#1})}{%
	\IfBeginWith{\InputString}{Algorithm}{%
		Algorithm \ref{#1}}{%
	\IfBeginWith{\InputString}{Figure}{%
		Figure (\ref{#1})}{%
	\IfBeginWith{\InputString}{Question}{%
		Question (\ref{#1})}{%
	\IfBeginWith{\InputString}{Footnote}{%
		Footnote \ref{#1}}{%
		\ref{#1}}}}}}}}}}}}%
}

\definecolor{TodoRed}{RGB}{150,50,0}

\newcommand{\TextForAll}{\hspace{2pt} \text{ for all } \hspace{2pt}}
\newcommand{\TextSuchThat}{\hspace{2pt}\text{ such that }\hspace{2pt}}
\newcommand{\TextIf}{\hspace{2pt}\text{ if }\hspace{2pt}}
\newcommand{\TextWith}{\hspace{2pt}\text{ with }\hspace{2pt}}
\newcommand{\TextAnd}{\hspace{2pt}\text{ and }\hspace{2pt}}
\newcommand{\TextExists}{\hspace{2pt}\text{ there exists }\hspace{2pt}}
\newcommand{\TextAs}{\hspace{2pt}\text{ as }\hspace{2pt}}

\newcommand{\ZeroSet}{\left\{0\right\}}

\newcommand{\abs}[1]{\left|#1\right|}
\newcommand{\RoundUp}[1]{\left\lceil #1\right\rceil}

\newcommand{\norm}[1]{\left\|#1\right\|}
\newcommand{\scprod}[2]{\langle#1,#2\rangle}

\newcommand{\ProjToIndex}[2]{\left.#2\right|_{#1}}

\newcommand{\SetOf}[1]{\left[#1\right]}

\newcommand{\argmin}[1]{\underset{#1}{\textnormal{argmin}}\hspace{1pt}}

\newcommand{\Prob}[1]{\mathbb{P}\left[#1\right]}
\newcommand{\Expect}[1]{\mathbb{E}\left[#1\right]}

\newcommand{\GaussianRV}[2]{\mathcal{N}\left(#1,#2\right)}

\def\AddBasicFunction#1#2{
	\expandafter\def\csname #1\endcsname##1{
		\def\InputString{##1}
		\def\CheckString{}
		\ifx\InputString\CheckString 
			#2
		\else
			#2 \left(##1\right)
		\fi
	}
}
\AddBasicFunction{Exp}{\exp}
\AddBasicFunction{Cos}{\cos}
\AddBasicFunction{Sin}{\sin}
\AddBasicFunction{Tan}{\tan}
\AddBasicFunction{Ln}{\ln}
\AddBasicFunction{Log}{\log}
\AddBasicFunction{ArcCos}{\arccos}
\AddBasicFunction{ArcSin}{\arcsin}
\AddBasicFunction{ArcTan}{\arctan}
\AddBasicFunction{Cosh}{\cosh}
\AddBasicFunction{Sinh}{\sinh}
\AddBasicFunction{Tanh}{\tanh}
\AddBasicFunction{ArcCosh}{\textnormal{arccosh}}
\AddBasicFunction{ArcSinh}{\textnormal{arcsinh}}
\AddBasicFunction{ArcTanh}{\textnormal{arctanh}}

\AddBasicFunction{Kernel}{\ker}
\AddBasicFunction{Rank}{\textnormal{rank}}
\AddBasicFunction{Range}{\textnormal{Ran}}

\AddBasicFunction{supp}{\textnormal{supp}}
\newcommand{\conv}[1]{\textnormal{conv}\left(#1\right)}

\AddBasicFunction{sgn}{\textnormal{sgn}}

\newcommand{\ExpE}{\mathrm{e}}

\newcommand{\SetSize}[1]{\#\left(#1\right)}

\newcommand{\SubDiff}[2]{\partial\left.#1\hspace{1pt}\right|_{#2}}
\newcommand{\SecOrd}{\textnormal{SecOrd}}
\newcommand{\NormIndex}{ }
\newcommand{\DualNormSymbol}{\ast}
\newcommand{\normRHS}[1]{\norm{#1}\NormIndex}
\newcommand{\normDual}[1]{\norm{#1}_\DualNormSymbol}
\newcommand{\compr}[2]{d_1\left(#1,\Sigma_#2\right)}

\AddBasicFunction{Decoder}{Q}

\newcommand{\GW}[1]{\ell\left(#1\right)}
\newcommand{\UnitSphere}[2]{\mathbb{S}^{#1-1}_{#2}}

\newcommand{\GMC}[1]{E_#1}
\newcommand{\NSPCone}{T}

\AddBasicFunction{RightVertices}{\textnormal{Row}}

\newcommand{\hrlone}{rLASSO}
\newcommand{\hrloneParam}{\lambda}

\AddBasicFunction{fFunc}{f}
\AddBasicFunction{gFunc}{g}

\newcommand{\tempmin}[1]{\min_{#1}}
\renewcommand{\subset}{\subseteq}
\renewcommand{\supset}{\supseteq}

\newtheorem{CounterTheorem}{}[section]

\newtheorem{Definition}[CounterTheorem]{Definition}
\newtheorem{Theorem}[CounterTheorem]{Theorem}
\newtheorem{Proposition}[CounterTheorem]{Proposition}
\newtheorem{Lemma}[CounterTheorem]{Lemma}
\newtheorem{Corollary}[CounterTheorem]{Corollary}
\newtheorem{Remark}[CounterTheorem]{Remark}

\newtheorem{Experiment}{Experiment}

\newtheorem{AlgorithmEnvirorementForNTheorem}[CounterTheorem]{Algorithm}

\renewcommand{\Vec}[1]{\mathbf{#1}}

\newcommand{\eVec}{\Vec{e}}

\newcommand{\gVec}{\Vec{g}}

\newcommand{\vVec}{\Vec{v}}
\newcommand{\wVec}{\Vec{w}}
\newcommand{\xVec}{\Vec{x}}
\newcommand{\yVec}{\Vec{y}}
\newcommand{\zVec}{\Vec{z}}

\newcommand{\Mat}[1]{\mathbf{#1}}

\newcommand{\AMat}{\Mat{A}}
\newcommand{\BMat}{\Mat{B}}

\newcommand{\UMat}{\Mat{U}}
\newcommand{\VMat}{\Mat{V}}

\newcommand{\ID}{\textnormal{Id}}

\newcommand{\IDMat}{\ID_{mat}}

\title{Robust Instance-Optimal Recovery of Sparse Signals at Unknown Noise Levels}
\date{}
\author{
	Hendrik Bernd Petersen
	\footnote{
		Communications and Information Theory Group,
		Technische Universtität Berlin, Berlin,
		\href{mailto:petersen@tu-berlin.de}{petersen@tu-berlin.de}
	}
	\and
	Peter  Jung
	\footnote{
		Communications and Information Theory Group,
		Technische Universtität Berlin, Berlin,
		\href{mailto:peter.jung@tu-berlin.de}{peter.jung@tu-berlin.de}
	}
}

\begin{document}
	\maketitle
\begin{abstract}
	We consider the problem of sparse signal recovery from noisy measurements.
	Many of frequently used recovery methods rely on some sort of tuning depending on
	either noise or signal parameters.
	If no estimates for either of them are available,
	the noisy recovery problem is significantly harder.
	The square root LASSO and the least absolute deviation LASSO are known to be noise-blind,
	in the sense that the tuning parameter can be chosen independent on the noise and the signal.
	We generalize those recovery methods to the \hrlone{} and give a recovery guarantee once
	the tuning parameter is above a threshold. Moreover, we analyze the effect of
	mistuning on a theoretic level and prove
	the optimality of our recovery guarantee.
	Further, for Gaussian matrices we give a refined analysis
	of the threshold of the tuning parameter
	and proof a new relation of the tuning parameter
	on the dimensions. Indeed, for a certain amount of measurements the tuning parameter
	becomes independent on the sparsity.
	Finally, we verify that the least absolute deviation LASSO
	can be used with random walk matrices of uniformly at random
	chosen left regular biparitite graphs.
\end{abstract}
\section{Introduction}
We consider the problem of sparse signal recovery from noisy measurements.
Classical recovery methods require additional information about either the noise or the signal.
The basis pursuit denoising needs to be tuned in the order of the noise power \cite[Theorem~4.22]{IntroductionCS},
the $\ell_1$-norm constrained least residual needs to be tuned in the order of the $\ell_1$-norm
of the signal \cite[Theorem~11.1]{Tib_LASSO}, and the tuning parameter of the
$\ell_1$-norm penalized least squares (LASSO)
allegedly depends on the noise power \cite[Theorem~11.1]{Tib_LASSO}.
If no prior information about the signal and noise are available,
these methods fail or, in their sub-optimally tuned
versions, yield a sub-optimal recovery guarantee \cite{QuotientProperty}.
Thus, it is desirable to find other noise-blind recovery methods.
A commonly used approach is cross validation which is often computationally more
expensive and theoretical guarantees are not fully understood, see
exemplary \cite{Meinshausen2006} for further discussion.
If the signal is non-negative, the non-negative least squares \cite{NNLS_first}
and non-negative least absolute deviation \cite{NNLAD} are tuning free methods that achieve
almost as good robustness bounds as the optimally tuned basis pursuit denoising
and $\ell_1$-norm constrained least residual.
Without the non-negativity assumption this problem is harder.\\
The square-root LASSO, introduced in \cite{SQLIntroduction}, is an alteration of the LASSO
where the square of the $\ell_2$-norm is removed. The square-root LASSO is known
to be a noise-blind recovery method. Indeed, in \cite{SQLIntroduction}
it has been shown that the tuning parameter can be chosen independent on the noise power.
Further, the square root LASSO has been studied in \cite{SQL_01,SQL_02,SQL_03,SQL_04,SQL_05,SQL_alteration_01,SQL_alteration_02}.\\
The least absolute deviation LASSO is an alteration of the square root LASSO,
where the $\ell_2$-norm of the data fidelity term is replaced by an $\ell_1$-norm.
The least absolute deviation LASSO has also been studied frequently
\cite{LAD_01,LAD_02,LAD_05,LAD_sparse_noise_01,LAD_sparse_noise_02,
	LAD_sparse_noise_03,LAD_sparse_noise_04,LAD_sparse_noise_05,LAD_sparse_noise_06,
	LAD_alteration_01,LAD_alteration_02,LAD_2level_01,LAD_2level_02}.
Under the assumption that the measurement matrix extended by the identity of the measurement domain
has a null space property, it was proven that the least absolute deviation LASSO
can recover sparse signals exactly even
in the presence of sparse noise \cite{LAD_sparse_noise_01},
see also
\cite{LAD_sparse_noise_02,LAD_sparse_noise_03,LAD_sparse_noise_04,
	LAD_sparse_noise_05,LAD_sparse_noise_06,LAD_2level_01,LAD_2level_02}.
\subsection*{Our Contribution}
We will introduce the notion of a stable and robust decoder and
generalize the square root LASSO and the least absolute deviation LASSO to the
``$p$th-root LASSO `` (\hrlone{}). Under the assumption of a robust null space property we generalize the recovery guarantee of the square root LASSO.
In particular, we prove
that robust recovery is possible if the tuning parameter is larger than a threshold.
Further, this threshold is a smooth function of the parameters of the robust null space property.
We will then discuss the effect of the tuning parameter on the recovery guarantee
in a larger theoretical detail.
In particular, we prove that the error bound does not degenerate when the tuning parameter
is chosen too large.
On the other hand we prove that, if the tuning parameter
is chosen smaller than the threshold of our recovery guarantee,
recovery has to fail for at least one sparse signal.
This yields that the recovery guarantee is optimal in a certain sense and can not be improved.
In the second part of our work we focus on the estimation of the tuning parameter.
For Gaussian matrices we use Gordon's escape through the mesh \cite{gordon}
to estimate the tuning parameter from the phase transition
and refine the dependence of the tuning parameter on the dimensions.
This dependence coincides with the general established rule, to choose
the tuning parameter in the order of the square root of the sparsity,
only if sufficient measurements are present,
but if the number of measurements is close to the optimal number of measurements,
a different rule for the tuning parameter is better suitable.
Further, we will establish that \hrlone{} can be used with
random walk matrices of uniformly at random chosen left regular bipartite graphs
and will prove that the tuning parameter can be chosen in the order of a
constant, independent on all dimensions.
Lastly, we will verify our theoretical results by short numerical tests.
\section{Preliminaries}\label{Section:Prelimninaries}
Given a set $C\subset\mathbb{R}^N$ and a function $\fFunc{}:C\rightarrow\mathbb{R}$
we denote the set of minimizers of $\fFunc{}$ on $C$ by
\begin{align*}
	\argmin{\zVec\in C}\fFunc{\zVec}:=\left\{\zVec\in C:\fFunc{\zVec}=\inf_{\zVec'\in C}\fFunc{\zVec'}\right\}.
\end{align*}
For $q\in\left[1,\infty\right)$ and $\xVec\in\mathbb{R}^N$ we denote the $\ell_q$-norm by
$\norm{\xVec}_q:=\left(\sum_{n=1}^N \abs{\xVec_n}^q\right)^\frac{1}{q}$
and the $\ell_\infty$-norm by
$\norm{\xVec}_\infty:=\sup_{n\in\SetOf{N}}\abs{\xVec_n}$.
By
$\UnitSphere{N}{\ell_q}:=\left\{\zVec\in\mathbb{R}^N:\norm{\zVec}_q=1\right\}$
we denote the unit sphere in the $\ell_q$-norm.
If $q=\infty$, we use the notation $\frac{1}{q}:=q^{-1}:=0$.
For a number $N\in\mathbb{N}$ we denote $\SetOf{N}:=\left\{1,\dots,N\right\}$.
For a set $T\subset\SetOf{N}$ we denote the number of elements in $T$ by $\SetSize{T}$.
We also write $\SetSize{T}\leq S$ to mean $T\subset\SetOf{N}$ and $\SetSize{T}\leq S$.
For $T\subset\left[N\right]$ we denote the projection
onto the subspace $\left\{\zVec\in\mathbb{R}^N:\supp{\zVec}\subset T\right\}$
by $\ProjToIndex{T}{\cdot}$. For $\xVec\in\mathbb{R}^N$ it is given by
\begin{align*}
	\ProjToIndex{T}{\xVec}
	:=\begin{Bmatrix}
		\left(\ProjToIndex{T}{\xVec}\right)_n=\xVec_n & \TextIf n\in T
		\\\left(\ProjToIndex{T}{\xVec}\right)_n=0 & \TextIf n \notin T
	\end{Bmatrix}.
\end{align*}
We call $\AMat\in\mathbb{R}^{M\times N}$ a measurement matrix and any map
$\Decoder{}:\mathbb{R}^M\rightarrow\mathbb{R}^N$ a decoder.
A vector $\xVec\in\mathbb{R}^N$ is called signal,
and any $\yVec\in\mathbb{R}^M$ is called measurement.
Given all these we call $\eVec:=\yVec-\AMat\xVec$ the noise.
A signal $\xVec$ is called $S$-sparse if $\norm{\xVec}_0:=\SetSize{\supp{\xVec}}\leq S$.
The set of $S$-sparse signals is denoted by
$\Sigma_S:=\left\{\zVec\in\mathbb{R}^N:\norm{\zVec}_0\leq S\right\}$.
Given some $q\in\left[1,\infty\right]$ and $S\in\SetOf{N}$
the compressibility of a signal $\xVec$ is measured by
\begin{align*}
	d_q\left(\xVec,\Sigma_S\right)=\inf_{\zVec\in\Sigma_S}\norm{\xVec-\zVec}_q,
\end{align*}
and describes how close the signal is to being $S$-sparse.
Motivated by \cite{InstanceOptimality} we would like to get decoders
that are able to bound the estimation error by a linear function in the
uncertainties $\compr{\xVec}{S}$ and $\normRHS{\eVec}$,
for some norm $\normRHS{\cdot}$ on $\mathbb{R}^M$.
\begin{Definition}
	Let $S\in\SetOf{N}$, $q\in\left[1,\infty\right]$ and $\normRHS{\cdot}$ be any norm on $\mathbb{R}^M$.
	Let $\AMat\in\mathbb{R}^{M\times N}$ and $\Decoder{}:\mathbb{R}^M\rightarrow\mathbb{R}^N$.
	If there exist $C,D\in\left[0,\infty\right)$ such that
	\begin{align*}
		\norm{\Decoder{\yVec}-\xVec}_q
		\leq CS^{\frac{1}{q}-1}\compr{\xVec}{S}
			+D\normRHS{\yVec-\AMat\xVec}
		\TextForAll\xVec\in\mathbb{R}^N,\yVec\in\mathbb{R}^M
	\end{align*}
	holds true, then we say $\Decoder{}$ is an
	$\ell_q$-stable robust decoder of order $S$ with respect to $\normRHS{\cdot}$
	for $\AMat$ with constants $C$ and $D$.
	We shorten this to $\ell_q$-SRD of order $S$ wrt $\normRHS{\cdot}$ for $\AMat$
	with constants $C$ and $D$,
	and omit parts of it in case they are not of importance or clear from context.
\end{Definition}
To find such decoders the measurement matrix needs to obey certain properties.
We consider a robust null space property introduced in \cite[Definition~4.21]{IntroductionCS}.
\begin{Definition}[ {\cite[Definition~4.21]{IntroductionCS}} ]\label{Definition:SRNSP}
	Let $S\in\SetOf{N}$, $q\in\left[1,\infty\right]$ and $\normRHS{\cdot}$ be any norm on $\mathbb{R}^M$ and
	$\AMat\in\mathbb{R}^{M\times N}$.
	If there exist $\rho\in\left[0,1\right)$ and $\tau\in\left[0,\infty\right)$ such that
	\begin{align*}
		\norm{\ProjToIndex{T}{\vVec}}_q
		\leq \rho S^{\frac{1}{q}-1}\norm{\ProjToIndex{T^c}{\vVec}}_1
			+\tau\normRHS{\AMat\vVec}
		\TextForAll \SetSize{T}\leq S \TextAnd \vVec\in\mathbb{R}^N
	\end{align*}
	holds true, then we say $\AMat$ has the
	$\ell_q$-robust null space property of order $S$ with
	respect to $\normRHS{\cdot}$ with constants $\rho$ and $\tau$.
	We shorten this to
	$\ell_q$-RNSP of order $S$ wrt $\normRHS{\cdot}$
	with constants $\rho$ and $\tau$,
	and omit parts of it in case they are not of importance or clear from context.
	$\rho$ is called stableness constant and $\tau$ is called robustness constant.
\end{Definition}
It is well known that certain decoders (basis pursuit denoising, $\ell_1$-norm constrained least residual, LASSO) obey robust recovery guarantees
if the measurement matrix obeys an RNSP
\cite[Theorem~4.22]{IntroductionCS}\cite[Theorem~11.1]{Tib_LASSO}\cite[Theorem~11.1]{Tib_LASSO}.
However, these fail to define an SRD for $\AMat$,
since they rely on some form of a priori knowledge to achieve these bounds.
Only under the additional assumption of a quotient property, the basis pursuit defines
an SRD for $\AMat$ \cite{QuotientProperty}.\\
We introduce the pth root LASSO (\hrlone{}) as follows:
Let $\normRHS{\cdot}$ be any norm on $\mathbb{R}^M$ and $\hrloneParam\in\left[0,\infty\right)$.
Then \hrlone{} with input $\yVec\in\mathbb{R}^M$ and $\AMat\in\mathbb{R}^{M\times N}$ is the optimization problem
\begin{align}\label{Problem:HRL1}
	\tag{$\hrlone_{\hrloneParam}$}
	\argmin{\zVec\in\mathbb{R}^n}
		\norm{\zVec}_1+\hrloneParam\normRHS{\yVec-\AMat\zVec}.
\end{align}
\hrlone{} has been studied in the case $\normRHS{\cdot}=\norm{\cdot}_2$ under the name
square root LASSO
\cite{SQLIntroduction,SQL_01,SQL_02,SQL_03,SQL_04,SQL_05,
	SQL_alteration_01,SQL_alteration_02}
and in the case $\normRHS{\cdot}=\norm{\cdot}_1$ under the name least absolute deviation LASSO
\cite{LAD_01,LAD_02,LAD_05,LAD_sparse_noise_01,LAD_sparse_noise_02,
	LAD_sparse_noise_03,LAD_sparse_noise_04,LAD_sparse_noise_05,
	LAD_sparse_noise_06,LAD_alteration_01,LAD_alteration_02,
	LAD_2level_01,LAD_2level_02}.
The first recovery guarantee of the square root LASSO has been presented in \cite{SQLIntroduction}.
To obtain recovery guarantees the authors assumed a compatibility condition (or sometimes restricted eigenvalue condition) and other minor conditions.
Note that by \thref{Proposition:$S$-CC<=>$S$-SRNSP} the compatibility condition is equivalent to the corresponding RNSP, and thus
our work is a generalization of \cite{SQLIntroduction}.
Relations of the compatibility condition to other conditions for LASSO can be found in \cite{CompatibilityCondition}.\\
One of the first recovery guarantee of the least absolute deviation LASSO was presented in
\cite{LAD_01}.
Closest to our work and results is \cite{LAD_2level_02}, although it considers
a certain structured sparsity model. In particular,
under the assumption of a $2$-level robust null space property (definition see
\cite{LAD_01}) they prove
that the least absolute deviation LASSO
can recover sparse signals exact, even in the presence of sparse noise.
There are numerous works that consider
exact recovery in the presence of sparse noise under such stronger conditions
\cite{LAD_sparse_noise_02,LAD_sparse_noise_03,LAD_sparse_noise_04,
	LAD_sparse_noise_05,LAD_sparse_noise_06,LAD_2level_01,LAD_2level_02}.
We will not consider such stronger requirements, but only use the weakest requirements possible,
i.e., the RNSP.
	
\section{Theoretic Results for \hrlone{}}\label{Section:Theory}
\subsection{Recovery Guarantee for \hrlone{}}
\label{Subsection:Theory_RecoveryGuarantees}
The main statement here is that robust recovery independent on the noise power is possible,
as long as $\hrloneParam$ is above a threshold.
This threshold is a smooth function in the constants
$\tau$ and $S$ of the RNSP.
\begin{Theorem}[Recovery with \hrlone{}]\label{Theorem:HRL1Recovery:Lambda>TauSq}
	Let $\AMat\in\mathbb{R}^{M\times N}$ have $\ell_q$-RNSP of order $S$ wrt $\normRHS{\cdot}$
	with constants $\rho$ and $\tau$. Let
	\begin{align}\label{Equation:thresh_hold:Theorem:HRL1Recovery:Lambda>TauSq}
		\hrloneParam>\tau S^{1-\frac{1}{q}}
		\text{ and set }
		\rho'
		=
		\begin{Bmatrix}
			\max\left\{\rho,\frac{1}{2\hrloneParam}\tau S^{1-\frac{1}{q}}\left(1+\sqrt{8\frac{\hrloneParam}{\tau}S^{\frac{1}{q}-1}+1}\right)-1\right\}
			& \TextIf & q\in\left(1,\infty\right]
			\\
			\max\left\{\rho,\frac{2}{\hrloneParam}\tau-1\right\}
			& \TextIf & q=1
		\end{Bmatrix}.
	\end{align}
	Then $\rho'\in\left[\rho,1\right)$ and for all $\xVec\in\mathbb{R}^N$ and $\yVec\in\mathbb{R}^M$
	any minimizer $\xVec^\#$ of
	\begin{align*}
		\tempmin{\zVec\in\mathbb{R}^n}
			\norm{\zVec}_1
			+\hrloneParam\normRHS{\yVec-\AMat\zVec}
	\end{align*}
	obeys
	\begin{align}\label{Equation:ErrorBound:Theorem:HRL1Recovery:Lambda>TauSq}
		\norm{\xVec^\#-\xVec}_q
		\leq\begin{Bmatrix}
			2\frac{\left(1+\rho'\right)^2}{1-\rho'}S^{\frac{1}{q}-1}\compr{\xVec}{S}
				+\left(\frac{3+\rho'}{1-\rho'}\tau
				+\frac{\left(1+\rho'\right)^2}{1-\rho'}S^{\frac{1}{q}-1}\hrloneParam\right)\normRHS{\yVec-\AMat\xVec}
			& \TextIf & q\in\left(1,\infty\right]
			\\
			2\frac{1+\rho'}{1-\rho'}\compr{\xVec}{S}
				+\left(\frac{2}{1-\rho'}\tau+\frac{1+\rho'}{1-\rho'}\hrloneParam\right)\normRHS{\yVec-\AMat\xVec}
			& \TextIf & q=1
		\end{Bmatrix}.
	\end{align}
	In particular, $\rho'=\rho$ happens if and only if
	\begin{align}\label{Equation:EQ1:Subsection:Theory_RecoveryGuarantees}
		\hrloneParam\geq\begin{Bmatrix}
			\frac{3+\rho}{\left(1+\rho\right)^2}\tau S^{1-\frac{1}{q}} & \TextIf & q\in\left(1,\infty\right]
			\\\frac{2}{1+\rho}\tau & \TextIf & q=1
		\end{Bmatrix}.
\end{align}
\end{Theorem}
	The result is proven in \refP{Subsection:Theory_RecoveryGuarantees:Proofs}.
Note that $S$ from the threshold is not $\norm{\xVec}_0$
but the order of the RNSP of the matrix $\AMat$,
and exact recovery in the absence of noise is possible for all $\xVec$ with $\norm{\xVec}_0\leq S$.
The case $q=1$ is obviously interesting since the sparsity disappears from the condition for $\hrloneParam$.
If $\hrloneParam{}$ obeys \refP{Equation:EQ1:Subsection:Theory_RecoveryGuarantees} with
equality, then the bound \refP{Equation:ErrorBound:Theorem:HRL1Recovery:Lambda>TauSq}
is the same error bound as the so far best known result for the optimally tuned basis pursuit denoising
and $\ell_1$-norm constrained least residual \cite[Theorem~4.22]{IntroductionCS}.\\
At this point we have three open problems to address.
The first problem is: The error bound of the recovery guarantee scales with $\hrloneParam$.
What happens when $\hrloneParam$ converges to infinity?
This problem will be studied
in \refP{Subsection:Theory_ConvergenceToBP}.
The second problem is: Is there a recovery guarantee for $\hrloneParam\leq\tau S^{1-\frac{1}{q}}$?
This problem is related to the optimality of the recovery guarantee and
will be answered in \refP{Subsection:Theory_RecoveryEquivalence}.
The third problem is: Given a choice on $\hrloneParam{}$, how do we determine whether or not
the threshold $\hrloneParam>\tau S^{1-\frac{1}{q}}$ is fulfilled?
This problem will be answered in \refP{Section:NotNPHard}.
In particular, we will study the threshold explicitly for Gaussian matrices in
\refP{Subsection:NotNPHard_GaussianMatrices}.
Using the threshold 
\refP{Equation:thresh_hold:Theorem:HRL1Recovery:Lambda>TauSq} or
\refP{Equation:EQ1:Subsection:Theory_RecoveryGuarantees} with $q=2$
suggests to choose $\hrloneParam\asymp \sqrt{S}$.
A similar argument has been used in \cite[Section~5.1.4]{LAD_2level_02}.
However, this ignores the dependence of $\tau$ on the dimensions $M,N,S$.
We will give a more detailed analysis by estimating
$\tau$ from the phase transition inequality with Gordon's escape through the mesh \cite{gordon}.
From this it will follow that $\hrloneParam\asymp \sqrt{S}$
is only valid if the number of measurements is suboptimal
and close to the optimal number of measurements the tuning parameter scales differently.
For the exact results we refer to \thref{Theorem:Gaussian=>NSP}
and the discussion afterwards.
Before we proceed with these problems, we formulate one result.
\begin{Corollary}\label{Corollary:S-SRNSP=>S-SRD}
	Let $\AMat\in\mathbb{R}^{M\times N}$ have $\ell_q$-RNSP of order $S$ wrt $\normRHS{\cdot}$
	with constants $\rho$ and $\tau$.
	Then with \hrlone{} there exists an $\ell_q$-SRD of order $S$ wrt $\normRHS{\cdot}$ for $\AMat$
	with constants $C=2\frac{\left(1+\rho\right)^2}{1-\rho}$ and $D=2\frac{3+\rho}{1-\rho}\tau$.
	In particular, if $q=1$, we get the improved constants $C=2\frac{1+\rho}{1-\rho}$ and $D=\frac{4}{1-\rho}\tau$.
\end{Corollary}
	The result is proven in \refP{Subsection:Theory_RecoveryGuarantees:Proofs}.
%
%
%
\subsection{Asymptotic Analysis of \hrlone{} for $\hrloneParam\rightarrow\infty$}\label{Subsection:Theory_ConvergenceToBP}
Heuristically, if $\hrloneParam$ goes to infinity,
the second summand of \refP{Problem:HRL1} becomes more dominant and we expect that
the minimizer needs to be closer to a minimizer of $\tempmin{\zVec\in\mathbb{R}^N}\normRHS{\yVec-\AMat\zVec}$.
\refP{Problem:HRL1} then basically only minimizes $\norm{\zVec}_1$ under the restriction that $\normRHS{\yVec-\AMat\zVec}$
is almost minimal.
In this section we will prove that, indeed, if $\hrloneParam$ goes to infinity, the minimizers of \refP{Problem:HRL1}
move closer to the minimizers of
\begin{align}\label{Problem:BPImp}\tag{$BPImp$}
	\tempmin{\zVec\in\argmin{\zVec'\in\mathbb{R}^N}\normRHS{\AMat\zVec'-\yVec}}\norm{\zVec}_1.
\end{align}
Note that, if $\yVec\in\Range{\AMat}$,
this problem is the basis pursuit and \refP{Problem:HRL1} can be used to approximate a minimizer of basis pursuit.
Further, we get a verifiable condition if a minimizer of \refP{Problem:HRL1} is also an optimizer of \refP{Problem:BPImp}.
\begin{Theorem}\label{Theorem:ConvergenceHRL1Param}
	Let $\normRHS{\cdot}$ be any norm on $\mathbb{R}^M$, $\AMat\in\mathbb{R}^{M\times N}$
	and $\yVec\in\mathbb{R}^M$.
	For every $\hrloneParam\in\left[0,\infty\right)$ let $\xVec^\hrloneParam$ be any minimizer of 
	\begin{align*}
		\tempmin{\zVec\in\mathbb{R}^N}
			\norm{\zVec}_1
			+\hrloneParam\normRHS{\yVec-\AMat\zVec}.
	\end{align*}
	Then we have the following results:
	\begin{itemize}
		\item[(1)] We have two stopping criteria:\\
			$\xVec^\hrloneParam$ is a minimizer of \refP{Problem:BPImp}
			if and only if
			$
				\norm{\xVec^\hrloneParam}_1=\inf_{\zVec\in\argmin{\zVec'\in\mathbb{R}^N}\normRHS{\AMat\zVec'-\yVec}}\norm{\zVec}_1
			$, and\\
			$\xVec^\hrloneParam$ is a minimizer of \refP{Problem:BPImp}
			if and only if
			$
				\normRHS{\AMat\xVec^\hrloneParam-\yVec}=\inf_{\zVec'\in\mathbb{R}^N}\normRHS{\AMat\zVec'-\yVec}
			$.
		\item[(2)] We have the convergence of the stopping criteria
			\begin{align}
				\label{Equation:Convergence:StoppingCriteria:Theorem:ConvergenceHRL1Param}
				\norm{\xVec^\hrloneParam}_1&\nearrow \inf_{\zVec\in\argmin{\zVec'\in\mathbb{R}^N}\normRHS{\AMat\zVec'-\yVec}}\norm{\zVec}_1
				\TextAs
				\hrloneParam\nearrow\infty \TextAnd
				\\\label{Equation:Convergence:StoppingCriteriaResidual:Theorem:ConvergenceHRL1Param}
				\normRHS{\AMat\xVec^\hrloneParam-\yVec}&\searrow \inf_{\zVec'\in\mathbb{R}^N}\normRHS{\AMat\zVec'-\yVec}
				\TextAs
				\hrloneParam\nearrow\infty
				\text{ with distance bounded by }
				\hrloneParam^{-1}\left(\inf_{\zVec\in\argmin{\zVec'\in\mathbb{R}^N}\normRHS{\AMat\zVec'-\yVec}}\norm{\zVec}_1\right).
			\end{align}
		\item[(3)]
			The sequence $\xVec^\hrloneParam$ converges to the set of minimizers
			of \refP{Problem:BPImp}, meaning that
			\begin{align*}
				\lim_{\hrloneParam\rightarrow\infty}\hspace{5pt}
				\inf_{ \underset{ \text{\refP{Problem:BPImp}} }{ \text{$\zVec$ minimizer of} } }
				\norm{\xVec^\hrloneParam-\zVec}_2
				= 0.
			\end{align*}
		\item[(4)]
			Let $\lim_{n\rightarrow\infty}\hrloneParam_n=\infty$ and
			consider the sequence $\left(\xVec^{\hrloneParam_n}\right)_{n\in\mathbb{N}}$.
			If this sequence converges or the minimizer of \refP{Problem:BPImp} is unique,
			the sequence converges to a minimizer of \refP{Problem:BPImp}.
			In particular, the sequence always has a subsequence that converges to a minimizer of \refP{Problem:BPImp}.
		\item[(5)]
			If $\yVec\in\Range{\AMat}$, we have the convergence of optimal values
			\begin{align*}
				\lim_{\hrloneParam\rightarrow\infty}
				\norm{\xVec^\hrloneParam}_1+\hrloneParam\normRHS{\AMat\xVec^\hrloneParam-\yVec}
				=\inf_{\zVec:\AMat\zVec=\yVec}\norm{\zVec}_1.
			\end{align*}
	\end{itemize}
\end{Theorem}
	The result is proven in \refP{Subsection:Theory_ConvergenceToBP:Proofs}.
Surprisingly, the convergence for $\hrloneParam\rightarrow\infty$ occurs at a finite value
whenever $\AMat$ is surjective, since then the operator $\AMat^T$ is bounded below.
\begin{Proposition}[Finite Convergence]\label{Proposition:FiniteConvergence}
	Let $\norm{\cdot}$ be a norm on $\mathbb{R}^M$ with dual norm
	$\normDual{\cdot}:=\sup_{\normRHS{\wVec}\leq 1}\abs{\scprod{\wVec}{\cdot}}$.
	Let $\AMat\in\mathbb{R}^{M\times N}$ be surjective
	and $\hrloneParam>\hrloneParam_\infty:=\sup_{0\neq\wVec\in\mathbb{R}^M}\frac{\normDual{\wVec}}{\norm{\AMat^T\wVec}_\infty}$.
	Then $\hrloneParam_\infty<\infty$.
	Further, for $\yVec\in\mathbb{R}^M$ any minimizer $\xVec^\hrloneParam$ of 
	$
		\tempmin{\zVec\in\mathbb{R}^N}
			\norm{\zVec}_1
			+\hrloneParam\normRHS{\yVec-\AMat\zVec}
	$
	is also a minimizer of
	$
		\tempmin{\zVec\in\mathbb{R}^N:\AMat\zVec=\yVec}\norm{\zVec}_1
	$.
\end{Proposition}
	The result is proven in \refP{Subsection:Theory_ConvergenceToBP:Proofs}.
If $\normRHS{\cdot}=\norm{\cdot}_p$, the dual norm is $\normDual{\cdot}=\norm{\cdot}_{\left(1-\frac{1}{p}\right)^{-1}}$
and we can calculate an upper bound for $\hrloneParam_\infty$ in polynomial time.
Let $\AMat\in\mathbb{R}^{M\times N}$ be surjective and have the singular value decomposition $\AMat=\UMat\Sigma\VMat^T$ with
$\UMat\in\mathbb{R}^{M\times M}$, $\Sigma\in\mathbb{R}^{M\times M}$ and $\VMat\in\mathbb{R}^{N\times M}$.
Then $\Sigma$ is invertible and ${\AMat^T}^\dagger:=\UMat\Sigma^{-1}\VMat^T\in\mathbb{R}^{M\times N}$ is the Moore-Penrose inverse of $\AMat^T$
and obeys ${\AMat^T}^\dagger\AMat^T=\IDMat\in\mathbb{R}^{M\times M}$.
It follows that
\begin{align*}
	\hrloneParam_\infty
	=&\sup_{0\neq\wVec\in\mathbb{R}^M}\frac{\norm{\wVec}_{\left(1-\frac{1}{p}\right)^{-1}}}{\norm{\AMat^T\wVec}_\infty}
	=\sup_{0\neq\wVec\in\mathbb{R}^M}\frac{\norm{{\AMat^T}^\dagger\AMat^T\wVec}_{\left(1-\frac{1}{p}\right)^{-1}}}{\norm{\AMat^T\wVec}_\infty}
	=\sup_{0\neq\vVec\in\Range{\AMat^T}}\frac{\norm{{\AMat^T}^\dagger\vVec}_{\left(1-\frac{1}{p}\right)^{-1}}}{\norm{\vVec}_\infty}
	\\&\leq M^{1-\left(1-\frac{1}{p}\right)}\sup_{0\neq\vVec\in\Range{\AMat^T}}\frac{\norm{{\AMat^T}^\dagger\vVec}_\infty}{\norm{\vVec}_\infty}
	\leq M^{1-\left(1-\frac{1}{p}\right)}\sup_{0\neq\wVec\in\mathbb{R}^N}\frac{\norm{{\AMat^T}^\dagger\wVec}_\infty}{\norm{\wVec}_\infty}
	=M^{\frac{1}{p}}\norm{{\AMat^T}^\dagger}_{\infty\rightarrow\infty}.
\end{align*}
This value is computable in polynomial time since the norm is the maximum absolute row sum of ${\AMat^T}^\dagger$.
Other possible bounds involve the smallest, non-zero singular value.
\subsubsection*{Large Tuning Parameters for \hrlone{} and the Quotient Property}\label{SubSubsection:Theory_QuotientProperty}
By statement (3) of \thref{Theorem:ConvergenceHRL1Param} the minimizers of \hrlone{} converge to the minimizers of \refP{Problem:BPImp}
for $\hrloneParam\rightarrow\infty$.
This suggests that the error bound of \thref{Theorem:HRL1Recovery:Lambda>TauSq} is not tight for large
$\hrloneParam$ and that it should be possible to replace $\hrloneParam$ in the error bound
by a constant.
Indeed, this is true if $\AMat$ obeys a quotient property.
\begin{Definition}\label{Definition:QuotientProperty}
	Let $q\in\left[1,\infty\right]$ and $\normRHS{\cdot}$ be any norm on $\mathbb{R}^M$ and $\AMat\in\mathbb{R}^{M\times N}$.
	If there exists $d\in\left[0,\infty\right)$ such that
	\begin{align*}
		\TextForAll \eVec\in\mathbb{R}^M \TextExists \vVec\in\mathbb{R}^N \TextSuchThat
		\AMat\vVec=\eVec\TextAnd
		\norm{\vVec}_q\leq d\normRHS{\eVec}
	\end{align*}
	holds true, then we say $\AMat$ has $\ell_q$-quotient property with constant $d$
	relative to $\normRHS{\cdot}$.
\end{Definition}
In \cite{QuotientProperty} it was shown that the additional assumption of the quotient property
yields robust recovery guarantees for the basis pursuit.
Further, it was shown that Gaussian matrices obey a good quotient property with high
probability. In \cite[Chapter 11.2]{IntroductionCS} the techniques were adapted to
account for the RNSP instead of the restricted isometry property.
In particular, from \cite[Lemma~11.15 and 11.16]{IntroductionCS} we can deduce directly the following result.
\begin{Proposition}[ {\cite[Chapter 11.2]{IntroductionCS}} ]\label{Proposition:HRL1Recovery:QP}
	Under the additional assumption that $\AMat\in\mathbb{R}^{M\times N}$ has
	$\ell_q$-quotient property with constant $d$ relative to $\normRHS{\cdot}$
	and that \refP{Equation:EQ1:Subsection:Theory_RecoveryGuarantees} holds true,
	the error bound \refP{Equation:ErrorBound:Theorem:HRL1Recovery:Lambda>TauSq}
	of \thref{Theorem:HRL1Recovery:Lambda>TauSq} can be replaced by
	\begin{align*}
		\norm{\xVec^\#-\xVec}_q
		\leq\begin{Bmatrix}
			2\frac{\left(1+\rho\right)^2}{1-\rho}S^{\frac{1}{q}-1}\compr{\xVec}{S}
			+\left(\frac{\left(1+\rho\right)\left(3+\rho\right)}{1-\rho}S^{\frac{1}{q}-1}d+\tau\right)\normRHS{\yVec-\AMat\xVec}
			& \TextIf & q\in\left(1,\infty\right]
			\\
			2\frac{1+\rho}{1-\rho}\compr{\xVec}{S}
				+\frac{3+\rho}{1-\rho}d\normRHS{\yVec-\AMat\xVec}_1
			& \TextIf & q=1
		\end{Bmatrix},
	\end{align*}
	which is independent on the possibly large $\hrloneParam$.
\end{Proposition}
A sketch for the proof can be found in
\refP{Subsection:Theory_ConvergenceToBP:Proofs}.
The strength of \hrlone{} is not that it achieves the error bound from \thref{Proposition:HRL1Recovery:QP} whenever $\AMat$ suffices a quotient property, but the strength of \hrlone{} is that it achieves stable and robust recovery even if $\AMat$ has
a bad quotient property constant.
Especially if the number of measurements $M$ increases,
the quotient property is harder to fulfill and the constant $d$ gets worse.
Thus, the error bound of
\thref{Proposition:HRL1Recovery:QP} gets worse just as the error bound of the minimizer of the basis pursuit \cite[Figure~1]{QP_plots}.
Opposed to that, the stableness constant $\rho$ and robustness constant $\tau$ of the RNSP get better when the number of measurements increases.
If we get a good estimate on these parameters, \thref{Theorem:HRL1Recovery:Lambda>TauSq}
with
\begin{align*}
	\hrloneParam =\begin{Bmatrix}
		\frac{3+\rho}{\left(1+\rho\right)^2}\tau S^{1-\frac{1}{q}} & \TextIf & q\in\left(1,\infty\right]
		\\\frac{2}{1+\rho}\tau & \TextIf & q=1
	\end{Bmatrix}
\end{align*}
gives an error bound that gets better with increasing number of measurements $M$!
This effect will also be verified numerically in \refP{Subsection:Numerics_GaussianMatrices}.
\subsection{Equivalent Conditions for Successful Recovery with \hrlone{}}\label{Subsection:Theory_RecoveryEquivalence}
In this subsection we consider the second problem from
\refP{Subsection:Theory_RecoveryGuarantees},
i.e. what happens when $\hrloneParam$ goes to the threshold $\tau S^{1-\frac{1}{q}}$
from \thref{Theorem:HRL1Recovery:Lambda>TauSq}.
For this we need to introduce some new null space properties.
In \thref{Definition:SRNSP} we have introduced a robust null space property,
but, before this property was introduced in the way we have it here, there have been different notions of
null space properties. We will give a brief history of null space properties, but
for a general overview about null space properties we refer the reader to \cite[Section~4]{IntroductionCS}.
\begin{Definition}\label{Definition:RNSP}
	Let $S\in\SetOf{N}$, $q\in\left[1,\infty\right]$ and $\normRHS{\cdot}$ be any norm on $\mathbb{R}^M$ and
	$\AMat\in\mathbb{R}^{M\times N}$.
	\begin{itemize}
		\item[(1)]
		If
	\begin{align*}
		\norm{\ProjToIndex{T}{\vVec}}_q
		<S^{\frac{1}{q}-1}\norm{\ProjToIndex{T^c}{\vVec}}_1
		\TextForAll \SetSize{T}\leq S \TextAnd \vVec\in\Kernel{\AMat}\setminus\ZeroSet
	\end{align*}
	holds true, then we say $\AMat$ has
	$\ell_q$-null space property of order $S$.
	We shorten this to $\ell_q$-NSP of order $S$.	
		\item[(2)] 
	If there exists a $\tau\in\left[0,\infty\right)$ such that
	\begin{align*}
		\norm{\ProjToIndex{T}{\vVec}}_q
		< S^{\frac{1}{q}-1}\norm{\ProjToIndex{T^c}{\vVec}}_1
			+\tau\normRHS{\AMat\vVec}
		\TextForAll \SetSize{T}\leq S \TextAnd \vVec\in\mathbb{R}^N\setminus\ZeroSet
	\end{align*}
	holds true, then we say $\AMat$ has
	$\ell_q$-only robust null space property of order $S$ with respect to $\normRHS{\cdot}$
	with constant $\tau$.
	We shorten this to $\ell_q$-ORNSP of order $S$ wrt $\normRHS{\cdot}$
	with constant $\tau$.
	\end{itemize}
	We omit parts of these definitions in case they are not of importance or clear
	from context.
\end{Definition}
The first null space property used was the NSP for noiseless recovery of sparse 
signals with the basis pursuit.
One of the first uses was \cite{NSP}, although the term null space property was not used.
To account for compressibility one considered the stableness constant.
In \cite{InstanceOptimality} the stable null space property was used for the first time
and also the term null space property appeared for the first time.
Lastly, the robustness constant was added in addition to the stableness constant to account for additive noise.
The result is the RNSP defined in \thref{Definition:SRNSP}. See for instance \cite[Section~4.3]{IntroductionCS}.
However, if we only add the robustness to the NSP, we obtain the ORNSP, which only accounts
for noise in the measurements but not for compressibility.
To the best of the knowledge of the authors this property has not been used before.
By \thref{Lemma:NSP<=>RNSP} and \thref{Lemma:RNSP<=>SRNSP} all null space properties introduced are
equivalent to each other. More interesting is that
the $\ell_1$-ORNSP characterizes whether or not
recovery with \hrlone{} is successful.
\begin{Theorem}[Equivalent Condition for Stable and Robust Decodability]\label{Theorem:HRL1SRdecEquivalence}
	Let $S\in\SetOf{N}$ and $\normRHS{\cdot}$ be any norm on $\mathbb{R}^M$.
	Let $\AMat\in\mathbb{R}^{M\times N}$ and $\hrloneParam\in\left[0,\infty\right)$.
	Then the following are equivalent:
	\begin{itemize}
		\item[(1)] $\AMat$ has $\ell_1$-ORNSP of order $S$ wrt $\normRHS{\cdot}$ with constant $\hrloneParam$.
		\item[(2)]
			Any decoder $\Decoder{}:\mathbb{R}^M\rightarrow\mathbb{R}^N$ such that
			$\Decoder{\yVec}\in\argmin{\zVec\in\mathbb{R}^N}\norm{\zVec}_1
					+\hrloneParam\normRHS{\AMat\zVec-\yVec}$ for all $\yVec\in\mathbb{R}^M$
				is an $\ell_1$-SRD of order $S$ wrt $\normRHS{\cdot}$ for $\AMat$.
		\item[(3)]
			For all $\xVec\in\Sigma_S$ we have $\left\{\xVec\right\}
				=\argmin{\zVec\in\mathbb{R}^N}\norm{\zVec}_1
					+\hrloneParam\normRHS{\AMat\zVec-\AMat\xVec}$.
	\end{itemize}
\end{Theorem}
	The result is proven in \refP{Subsection:Theory_RecoveryEquivalence:Proofs}.
None of the statements is equivalent to $\AMat$ having $\ell_1$-RNSP with robustness
constant $\hrloneParam$ since for such $A$ we can only guarantee the $\ell_1$-ORNSP
with constant $\tau>\hrloneParam$, see \thref{Lemma:RNSP<=>SRNSP}.
Further, if only one decoder that maps to solutions of \refP{Problem:HRL1} is an
SRD of order $S$, we can not prove any equivalence, but we can prove the equivalence
if all possible decoders are SRD of order $S$.
To find the exact parameters $\hrloneParam$ which give a recovery guarantee,
it thus remains to characterize all constants of the ORNSP. Thus, we introduce the
NSP shape constant.
\begin{Definition}\label{Definition:NSPShape}
	Let $S\in\SetOf{N}$, $q\in\left[1,\infty\right]$ and $\normRHS{\cdot}$ be a norm on $\mathbb{R}^M$.
	Let $\AMat\in\mathbb{R}^{M\times N}$ have $\ell_q$-NSP of order $S$.
	The constant
	\begin{align*}
		\tau_q^0
			:=&\sup_{\SetSize{T}\leq S}\sup_{\vVec\in\mathbb{R}^N\setminus\Kernel{\AMat}}
				\frac{\norm{\ProjToIndex{T}{\vVec}}_q-S^{\frac{1}{q}-1}\norm{\ProjToIndex{T^c}{\vVec}}_1}{\normRHS{\AMat\vVec}}
	\end{align*}
	is called NSP shape constant and the function
	\begin{align*}
		\rho_q\left(\tau\right)
			:=&\max\left\{0,\sup_{\SetSize{T}\leq S}\sup_{\underset{\ProjToIndex{T^c}{\vVec}\neq 0}{\vVec\in\mathbb{R}^N:}}
					\frac{\norm{\ProjToIndex{T}{\vVec}}_q-\tau\normRHS{\AMat\vVec}}{
					S^{\frac{1}{q}-1}\norm{\ProjToIndex{T^c}{\vVec}}_1}\right\}
			\TextForAll \tau\in\left[\tau_q^0,\infty\right).
	\end{align*}
	is called NSP shape function.
\end{Definition}
By \thref{Lemma:NSP<=>RNSP} the NSP shape constant and by
\thref{Lemma:RNSP<=>SRNSP} the NSP shape function
are well defined and obey
$\tau_q^0\in\left(0,\infty\right)$
and $\rho_q\left(\tau\right)\in\left[0,1\right)$ for all $\tau\in\left(\tau_q^0,\infty\right)$.
Further, the outer maximum of $\rho_q\left(\tau\right)$
is only required when the inner suprema all have no feasible points,
i.e. when $\AMat$ is injective.
We can prove the following result.
\begin{Corollary}\label{Corollary:ShapeHRL1Param}
	Let $S\in\SetOf{N}$, $q\in\left[1,\infty\right]$ and $\normRHS{\cdot}$ be a norm on $\mathbb{R}^M$ and
	$\AMat\in\mathbb{R}^{M\times N}$.
	Then we have the set equality
	\begin{align}
		\nonumber
		&\bigg\{\hrloneParam\in\left[0,\infty\right):
			\text{any decoder $\Decoder{}:\mathbb{R}^M\rightarrow\mathbb{R}^N$ such that
				$\Decoder{\yVec}\in\argmin{\zVec\in\mathbb{R}^N}\norm{\zVec}_1
					+\hrloneParam\normRHS{\AMat\zVec-\yVec}$ for all $\yVec\in\mathbb{R}^M$
			}
		\\\nonumber&\hspace{60pt}
		\text{is an $\ell_q$-SRD of order $S$ wrt $\normRHS{\cdot}$ for $\AMat$}
		\bigg\}
		\\=&
		\left\{\hrloneParam\in\left[0,\infty\right):
			\text{for all $\xVec\in\Sigma_S$ we have $\left\{\xVec\right\}
				=\argmin{\zVec\in\mathbb{R}^N}\norm{\zVec}_1
					+\hrloneParam\normRHS{\AMat\zVec-\AMat\xVec}$}\right\}.
		\label{Equation:Sets:Corollary:ShapeHRL1Param}
	\end{align}
	If $\AMat$ has $\ell_q$-NSP of order $S$, then
	$\AMat$ has $\ell_1$-NSP of order $S$ and $\tau_1^0\leq \tau_q^0S^{1-\frac{1}{q}}$
	and the sets in \refP{Equation:Sets:Corollary:ShapeHRL1Param}
	are equal to the set $\left(\tau_1^0,\infty\right)$.
\end{Corollary}
	The result is proven in \refP{Subsection:Theory_RecoveryEquivalence:Proofs}.
Let us emphasize that under the $\ell_q$-NSP assumption both
sets are equal to $\left(\tau_1^0,\infty\right)$ and not $\left(\tau_q^0,\infty\right)$.
Note that the robustness of the recovery guarantee for \hrlone{} comes for free.
Now that we know that the set of parameters which yield a recovery guarantee is the open interval
$\left(\tau_1^0,\infty\right)$, we come back to the original problem, namely what happens
if $\hrloneParam\rightarrow \tau S^{1-\frac{1}{q}}$?

Note that the two inequalities $\tau_q^0 S^{1-\frac{1}{q}} \geq \tau_1^0$
and $\tau\geq\tau_q^0$ do not hold with equality in general
and hence $\tau S^{1-\frac{1}{q}}\geq\tau_1^0$ does not hold
with equality in general.
If $\tau S^{1-\frac{1}{q}}=\tau_1^0$,
then \thref{Theorem:HRL1Recovery:Lambda>TauSq} yields
a robust recovery guarantee for all $\hrloneParam>\tau_1^0$.
According to \thref{Corollary:ShapeHRL1Param} the parameter
$\hrloneParam=\tau_1^0$ does not have a robust recovery guarantee,
so that the bound $\hrloneParam>\tau S^{1-\frac{1}{q}}$
from \thref{Theorem:HRL1Recovery:Lambda>TauSq} is optimal.
On the other hand if $\tau S^{1-\frac{1}{q}}>\tau_1^0$,
then the tuning parameter $\hrloneParam=\tau S^{1-\frac{1}{q}}$ obeys
$\hrloneParam>\tau_1^0$ and thus has some robust recovery guarantee.
This gives the impression that \thref{Theorem:HRL1Recovery:Lambda>TauSq}
is not optimal.
Even though, $\tau S^{1-\frac{1}{q}}=\tau_1^0$ might fail to hold, we can still
deduce from this result that \thref{Theorem:HRL1Recovery:Lambda>TauSq} yields
a recovery guarantee for all parameters $\hrloneParam$ that have a recovery guarantee.
\begin{Remark}
	\thref{Theorem:HRL1Recovery:Lambda>TauSq} is optimal in the sense that the inclusion
	$\supset$ in \refP{Equation:Sets:Corollary:ShapeHRL1Param} is proven by applying
	\thref{Theorem:HRL1Recovery:Lambda>TauSq},
	meaning that, whenever the tuning parameter $\hrloneParam$ yields a uniform recovery guarantee,
	the requirements of \thref{Theorem:HRL1Recovery:Lambda>TauSq} are fulfilled.
\end{Remark}
Note that this remark only states the optimality of the threshold
but not optimality of the bounds of \thref{Theorem:HRL1Recovery:Lambda>TauSq}.
\section{Estimating the Tuning Parameter}\label{Section:NotNPHard}
In this section we consider the remaining problem from
\refP{Subsection:Theory_RecoveryGuarantees}. Namely, how to determine
whether or not the threshold $\hrloneParam>\tau S^{1-\frac{1}{q}}$ from
\thref{Theorem:HRL1Recovery:Lambda>TauSq} holds true.
By \thref{Corollary:ShapeHRL1Param} it would be sufficient to calculate the NSP shape constant
$\tau_q^0$. However, in \cite{SRNSPIsNotCalculatable} it was proven that, given the order $S$ of a
stable null space property, calculating the smallest stableness constant
of the stable null space property is NP-hard in general.
We have to accept that calculating the NSP shape constant $\tau_q^0$ might also be NP-hard in general.
If we want to use \hrlone{} to actually recover a signal,
we should prove that there are methods to calculate a $\hrloneParam$ above the threshold in a polynomial time.
In \thref{Proposition:FiniteConvergence} we have proven that, if $\AMat$ is surjective,
there exists a tuning parameter such that the minimizers of \hrlone{} are the minimizers
of the basis pursuit. Thus, this tuning parameter has a recovery guarantee
and is above the threshold.
Furthermore, it is computable in polynomial time. Hence, given the order $S$ of a null space property, there are methods
that calculate an upper bound on the NSP shape constant in polynomial time. We thus hope that there are methods
that calculate better bounds in polynomial time, although we have no such method yet.
The problem becomes easier, when we consider random matrices, since we can use the following idea:
For certain random matrices the phase transition inequality
is an intrinsic function in the variable $\tau$ and possibly $\rho$.
By solving for $\tau$ we get a function that maps the dimensions
$M,N,S$ and the constant $\rho$ to a tuning parameter $\hrloneParam$ that obeys
\refP{Equation:EQ1:Subsection:Theory_RecoveryGuarantees}.
\subsection{Gaussian Measurements}\label{Subsection:NotNPHard_GaussianMatrices}
We say that a random variable is an $\GaussianRV{\mu}{\sigma^2}$ random variable
if it is normal distributed with expectation $\mu$ and variance $\sigma^2$.
In this section we calculate a threshold for $\hrloneParam$ to ensure recovery guarantees
if the entries of $\AMat$ are i.i.d. $\GaussianRV{0}{M^{-1}}$ random variables.
In \cite[Theorem~11]{NSPCone} it was shown how the stableness and robustness constant
affect the phase transition inequality. In particular, the phase transition inequality is an
intrinsic function in the variable $\tau$. Thus, $\tau$ can be estimated from
the other constants. In view of the threshold
\refP{Equation:thresh_hold:Theorem:HRL1Recovery:Lambda>TauSq}
of \thref{Theorem:HRL1Recovery:Lambda>TauSq} we want to estimate
the smallest possible robustness constant, i.e., we want to estimate the NSP shape constant $\tau_2^0$.
We adapt the proof of \cite[Theorem~11]{NSPCone} suitably to be able to estimate the best possible
$\tau$. Indeed, the threshold we will calculate depends on the following constant.
\begin{Definition}
	Let $M\in\mathbb{N}$ and the entries of $\gVec\in\mathbb{R}^M$ be independent $\GaussianRV{0}{M^{-1}}$
	random variables.
	Then $\GMC{M}:=\Expect{\norm{\gVec}_2}$	is called Gordon's constant.
\end{Definition}
The constant originates from Gordon's escape through the mesh theorem
\cite[Theorem~9.21]{IntroductionCS}, which we will use in the proof of the main result.
By {\cite[Proposition~8.1(b)]{IntroductionCS}} Gordan's constant obeys
$\sqrt{\frac{M}{M+1}}\leq \GMC{M}\leq 1$.
For high $M$ it is thus feasible to estimate $\GMC{M}\approx 1$,
and $\GMC{M}$ should be considered as a constant.
The following result is basically \cite[Theorem~11]{NSPCone}, however their result uses estimates
that calculate a suboptimal $\tau$. Since we want to estimate $\tau_2^0$, we need
to optimize with respect to $\tau$.
\begin{Theorem}\label{Theorem:Gaussian=>NSP}
	Let $\rho\in\left(0,1\right)$, $\eta\in\left(0,1\right)$
	and the entries of $\AMat\in\mathbb{R}^{M\times N}$
	be independent $\GaussianRV{0}{M^{-1}}$ random variables.
	If
	\begin{align}\label{Equation:PhaseTransition:Theorem:Gaussian=>NSP}
		\tau\geq\left(\GMC{M}-\sqrt{1+\left(1+\rho^{-1}\right)^2}
			\left(\sqrt{2\frac{S}{M}\Ln{\ExpE\frac{N}{S}}}+\sqrt{\frac{S}{M}}\right)-\sqrt{\frac{2}{M}\Ln{\eta^{-1}}}\right)^{-1}>0,
	\end{align}
	then, with probability of at least $1-\eta$,
	$\AMat$ has $\ell_2$-RNSP of order $S$ wrt $\norm{\cdot}_2$
	with constants $\rho$ and $\tau$.
	In this case, setting $\hrloneParam:=\frac{3+\rho}{\left(1+\rho\right)^2}\tau\sqrt{S}$ yields
	that for all $\xVec\in\mathbb{R}^N$ and $\yVec\in\mathbb{R}^M$
	any minimizer $\xVec^\#$ of
	\begin{align*}
		\tempmin{\zVec\in\mathbb{R}^n}
			\norm{\zVec}_1
			+\hrloneParam\norm{\yVec-\AMat\zVec}_2
	\end{align*}
	obeys
	\begin{align*}
		\norm{\xVec^\#-\xVec}_2
		\leq
			2\frac{\left(1+\rho\right)^2}{1-\rho}S^{-\frac{1}{2}}\compr{\xVec}{S}
				+2\frac{3+\rho}{1-\rho}\tau\norm{\yVec-\AMat\xVec}_2.
	\end{align*}
\end{Theorem}
	The result is proven in \refP{Subsection:NotNPHard_GaussianMatrices:Proofs}.
Combining \refP{Equation:PhaseTransition:Theorem:Gaussian=>NSP} and $\hrloneParam=\frac{3+\rho}{\left(1+\rho\right)^2}\tau\sqrt{S}$ yields
\begin{align}\label{Equation:PhaseTransition:tuning}
	\hrloneParam\geq\frac{3+\rho}{\left(1+\rho\right)^2}\left(\GMC{M}-\sqrt{1+\left(1+\rho^{-1}\right)^2}
		\left(\sqrt{2\frac{S}{M}\Ln{\ExpE\frac{N}{S}}}+\sqrt{\frac{S}{M}}\right)
			-\sqrt{\frac{2}{M}\Ln{\eta^{-1}}}\right)^{-1}\sqrt{S}.
\end{align}
This allows us to directly estimate $\hrloneParam$ from $M,N,S$
and any choice of $\eta,\rho$ such that the right hand side of
\refP{Equation:PhaseTransition:tuning} is positive.
In practice such an estimation is often infeasible due to suboptimally chosen bounds
in some inequalities of the proof.
Thus, the value of this inequality lies not in the direct relation of the
tuning parameter to the dimensions $M,N,S$, but in the order of this relation.\\
In particular, \refP{Equation:PhaseTransition:Theorem:Gaussian=>NSP} yields
that the robustness constant depends on the order of the term
\begin{align*}
	\SecOrd:=\GMC{M}\sqrt{M}-\sqrt{1+\left(1+\rho^{-1}\right)^2}
			\left(\sqrt{2S\Ln{\ExpE\frac{N}{S}}}+\sqrt{S}\right)-\sqrt{2\Ln{\eta^{-1}}}.
\end{align*}
If for instance for some $\alpha>0$ one of the three equalities
\begin{align*}
	\SecOrd
	=\begin{Bmatrix}
		\alpha
		\\\alpha S^\frac{1}{2}
		\\\alpha M^\frac{1}{2}
	\end{Bmatrix}
\end{align*}
holds true, then, with probability of at least $1-\eta$,
$\AMat$ has $\ell_2$-RNSP of order $S$ wrt $\norm{\cdot}_2$ with constants $\rho$ and
\begin{align*}
		\tau=\frac{\sqrt{M}}{\SecOrd}
		=\begin{Bmatrix}
			\alpha^{-1}M^\frac{1}{2}
			\\\alpha^{-1}\left(\frac{M}{S}\right)^{\frac{1}{2}}
			\\\alpha^{-1}
		\end{Bmatrix}
\end{align*}
respectively. In this cases, we should choose the tuning parameter
\begin{align*}
	\hrloneParam=\frac{3+\rho}{\left(1+\rho\right)^2}\tau S^{1-\frac{1}{2}}
	=\frac{3+\rho}{\left(1+\rho\right)^2}\alpha^{-1}
	\begin{Bmatrix}
		\sqrt{MS}
		\\\sqrt{M}
		\\\sqrt{S}
	\end{Bmatrix}
\end{align*}
respectively. Only in the third case this coincides with the simple rule to choose
$\hrloneParam\asymp \sqrt{S}$.
However, this case is also the case with the most measurements.
This suggests that $\hrloneParam\asymp\sqrt{MS}$ is a better choice
and requires less measurements. In between we have that
$\hrloneParam\asymp\sqrt{M}$ is also a viable choice
which has the advantage that it does not require knowledge about the possible
unknown $S$.
In \refP{Subsection:Numerics_GaussianMatrices}
we will verify in a short numeric experiment that a choice of tuning parameter independent on $S$ is indeed possible.\\
In any case it should be noted that the choice of the tuning parameter is
more complicated than $\hrloneParam\asymp\sqrt{S}$ and rather given by
the relation in \refP{Equation:PhaseTransition:tuning}.
It is rather impressiv that, if the measurements are optimal in the sense that
$\GMC{M}\sqrt{M}$ is in the same order as
$\sqrt{1+\left(1+\rho^{-1}\right)^2}
			\left(\sqrt{2S\Ln{\ExpE\frac{N}{S}}}+\sqrt{S}\right)+\sqrt{2\Ln{\eta^{-1}}}$
with the same leading constant,
the tuning parameter depends heavily on the second order term, i.e., the difference
\begin{align*}
	\SecOrd=\GMC{M}\sqrt{M}-\sqrt{1+\left(1+\rho^{-1}\right)^2}
			\left(\sqrt{2S\Ln{\ExpE\frac{N}{S}}}+\sqrt{S}\right)-\sqrt{2\Ln{\eta^{-1}}}.
\end{align*}
It is remarkable that contrary to other compressed sensing results the constant $C$ from the
phase transition inequality $M\geq C S\Ln{\ExpE\frac{N}{S}}$ is of importance.
Not only this, also the second order term $\sqrt{M}-\sqrt{C} \sqrt{S\Ln{\ExpE\frac{N}{S}}}$ is of importance as well.\\
In view of \thref{Corollary:ShapeHRL1Param} it is desirable
to estimate $\tau_2^0$ and thus $\tau_1^0\leq \tau_2^0\sqrt{S}$
as good as possible to determine the exact set of tuning parameters
that yield an $\ell_2$-SRD for $\AMat$.
If the entries of $\AMat\in\mathbb{R}^{M\times N}$ are independent $\GaussianRV{0}{M^{-1}}$ random variables, then $\AMat$ can have $\ell_2$-NSP of order $S$. Restricted to this event the NSP shape
constant $\tau_2^0$ is well defined and a random variable. The following result states two things.
At first it bounds the probability that $\tau_2^0$ is bounded above by a constant.
Secondly, given a sufficiently large tuning parameter, it bounds the probability that the
sufficient condition for recovery with \hrlone{} is fulfilled.
\begin{Proposition}\label{Proposition:Gaussian=>HRL1}
	Let the entries of $\AMat\in\mathbb{R}^{M\times N}$
	be independent $\GaussianRV{0}{M^{-1}}$ random variables.
	If
	\begin{align}\label{Equation:PhaseTransition:Theorem:Gaussian=>NSP:HRL1}
		\hrloneParam>\sqrt{S}\left(\GMC{M}-\sqrt{5}\left(\sqrt{2\frac{S}{M}\Ln{\ExpE\frac{N}{S}}}+\sqrt{\frac{S}{M}}\right)\right)^{-1}>0,
	\end{align}
	then, with probability of at least
	\begin{align*}
			1-\Exp{-\frac{1}{2}\left(\GMC{M}-\sqrt{5}\left(\sqrt{2\frac{S}{M}\Ln{\ExpE\frac{N}{S}}}+\sqrt{\frac{S}{M}}\right)-\sqrt{S}\hrloneParam^{-1}\right)^2M}\in\left(0,1\right),
	\end{align*}
	$\AMat$ has $\ell_2$-NSP of order $S$ and the NSP shape constant obeys
	$\hrloneParam>\tau_2^0 \sqrt{S}$.
	In this case, any decoder $\Decoder{}:\mathbb{R}^M\rightarrow\mathbb{R}^N$
	with $\Decoder{\yVec}\in\argmin{\zVec\in\mathbb{R}^N}\norm{\zVec}_1+\hrloneParam\norm{\AMat\zVec-\yVec}_2$ for all $\yVec\in\mathbb{R}^M$ is an $\ell_2$-SRD of order $S$ wrt $\norm{\cdot}_2$ for $\AMat$.
\end{Proposition}
	The proof can be found in \refP{Subsection:NotNPHard_GaussianMatrices:Proofs}
\subsection{Random Walk Matrices of Uniformly Distributed $D$-Left Regular Bipartite Graphs}\label{Subsection:NotNPHard_TuningForLeftRegularBipartiteGraphs}
\label{Subsection:NotNPHard_LRBGMatrices}
In view of \thref{Theorem:HRL1Recovery:Lambda>TauSq} 
it is desirable to find matrices which obey the $\ell_1$-RNSP since $S$ does not appear in the threshold for $\hrloneParam$. To generate such matrices we introduce left regular bipartite graphs. Although we will not present a detailed analysis as in
\refP{Subsection:NotNPHard_GaussianMatrices}, we can still deduce a result for
the least absolute deviation LASSO.
\begin{Definition}
	Let $\AMat\in\left\{0,1\right\}^{M\times N}$ and $D\in\SetOf{M}$. For $T\subset N$ we set
	$
		\RightVertices{T}:=\bigcup_{n\in T}\left\{m\in\SetOf{M}:\AMat_{m,n}=1\right\}
	$. If
	$
		\SetSize{\RightVertices{\left\{n\right\}}}=D \TextForAll n\in \SetOf{N}
	$,
	then $D^{-1}\AMat$ is called a random walk matrix of a $D$-left regular bipartite graph.
\end{Definition}
Uniformly at random chosen $D$-left regular bipartite graphs are similar to Gaussian matrices.
In particular, if $M\geq C S \Log{\ExpE\frac{N}{S}}$ they obey with high probability
the lossless expansion property,
which is the $\ell_1$ counterpart to the restricted isometry property and yields an $\ell_1$-RNSP.
For more details see \refP{Subsection:NotNPHard_LRBGMatrices:Proofs}.
\begin{Theorem}[Left Regular Bipartite Graph $\Rightarrow \hrloneParam=3$]\label{Theorem:RecoveryViaExpander}
	Let $\theta\in\left(0,\frac{1}{6}\right)$ and
	$D:=\RoundUp{\frac{2}{\theta}\Ln{\frac{\ExpE N}{2S}}}$.
	Let $\AMat\in\left\{0,D^{-1}\right\}^{M\times N}$
	be a uniformly at random chosen $D$-left regular bipartite graph.
	If
	\begin{align*}
		M\geq \frac{4}{\theta}\Exp{\frac{2}{\theta}}S\Ln{\frac{\ExpE N}{2S}},
	\end{align*}
	then, with probability of at least
	$
		1-\frac{2S}{\ExpE N}
	$,
	the matrix $\AMat$ has $\ell_1$-RNSP of order $S$ wrt $\norm{\cdot}_1$ with
	constants $\rho=\frac{2\theta}{1-4\theta}$ and $\tau=\frac{1}{1-4\theta}$.
	In this case, setting
	$\hrloneParam:=\frac{2}{1+\rho}\tau=\frac{2}{1-2\theta}\in\left(2,3\right)$ yields that
	for all $\xVec\in\mathbb{R}^N$ and $\yVec\in\mathbb{R}^M$ any minimizer $\xVec^\#$ of
	\begin{align*}
		\tempmin{\zVec\in\mathbb{R}^n}
			\norm{\zVec}_1
			+\hrloneParam\norm{\yVec-\AMat\zVec}_1
	\end{align*}
	obeys
	\begin{align*}
		\norm{\xVec^\#-\xVec}_1
		\leq 2\frac{1-2\theta}{1-6\theta}\compr{\xVec}{S}
			+\frac{4}{1-6\theta}\norm{\yVec-\AMat\xVec}_1.
	\end{align*}
\end{Theorem}
	The result is proven in \refP{Subsection:NotNPHard_LRBGMatrices:Proofs}.
In \refP{Subsection:Numerics_LRBGMatrices} we will verify that the upper bound
$\hrloneParam:=3$ is a viable choice
for \hrlone{} with uniformly at random chosen $D$-left regular bipartite graph.
Note that \hrlone{} defines an $\ell_1$-SRD of order $S$ wrt $\norm{\cdot}_1$ in this case
and can be used without prior information about the noise $\eVec$, the signal $\xVec$
or the order $S$ of the RNSP.
	
\section{Numerical Experiments}\label{Section:Numerics}
Given $q\in\left\{1,2\right\}$ we will compare the estimation error of
the minimizers of
\begin{align*}
	\tempmin{\zVec\in\mathbb{R}^N}\norm{\zVec}_1+\hrloneParam\norm{\AMat\zVec-\yVec}_q
\end{align*}
against the estimation errors of the minimizers of the following well known
optimization problems 
\begin{align*}
	\tag{BP}
	\tempmin{\zVec\in\mathbb{R}^N:\AMat\zVec=\yVec}\norm{\zVec}_1,
	&
	\\\tag{BPDN}
	\tempmin{\zVec\in\mathbb{R}^N:\norm{\AMat\zVec-\yVec}_q\leq \epsilon}\norm{\zVec}_1
	&
	\TextWith \epsilon:=\norm{\eVec}_q \TextAnd
	\\\tag{CLR}
	\tempmin{\zVec\in\mathbb{R}^N:\norm{\zVec}_1\leq \tau}\norm{\AMat\zVec-\yVec}_q
	&
	\TextWith \tau:=\norm{\xVec}_1.
\end{align*}
We will estimate some minimizer of these problems with the CVX toolbox of Matlab \cite{CVX1}\cite{CVX2}.
For the BPDN and the CLR we use the optimal tuning $\epsilon=\norm{\eVec}_q$
and $\tau=\norm{\xVec}_1$ respectively. These two represent a best case benchmark
that is using the unknown prior information $\norm{\eVec}_q$ and $\norm{\xVec}_1$.
As a worst case benchmark we use the BP since it requires no prior
information about the noise or the signal.
Tuning the BPDN too high often leads to worse estimation errors
than tuning it too low \cite[Figure~1]{BPDN_misstuning}.
Thus, it is reasonable to choose the BP as the best worst case benchmark.
Given $M,N,S,D,SNR\in\SetOf{N}$
the following experiment will be repeated multiple times:
\begin{Experiment}\label{Experiment:Experiment1}
	For each $k\in\SetOf{100}$ and given $SNR$ do the following:
	\begin{enumerate}
		\item
			If $q=2$, draw each component of $\AMat_k\in\mathbb{R}^{M\times N}$ independent
			as $\GaussianRV{0}{M^{-1}}$ random variable.
			If $q=1$, draw $\AMat_k\in\left\{0,D^{-1}\right\}^{M\times N}$ as
			a uniformly at random drawn $D$-left regular bipartite graph.
		\item
			Draw the signal $\xVec_k$ uniformly at random from $\Sigma_S\cap\UnitSphere{N}{\ell_q}$.
		\item
			Draw the noise $\eVec_k$ uniformly at random from
			$\frac{\norm{\AMat_k\xVec_k}_q}{SNR}\UnitSphere{M}{\ell_q}$.
		\item
			Define the observation $\yVec_k:=\AMat_k\xVec_k+\eVec_k$.
		\item
			For each optimization problem estimate a minimizer by $\xVec^\#_k$
			and collect the relative estimation errors\\
			$\frac{\norm{\xVec_k-\xVec^\#_k}_q}{\norm{\xVec_k}_q}$
			and the noise powers $\norm{\eVec_k}_q$.
	\end{enumerate}
	We calculate the mean normalized $\ell_q$-estimation error
	and the mean normalized $\ell_q$-estimation error per noise power, i.e.,
	we set
	\begin{align*}
		\frac{\norm{\xVec-\xVec^\#}_q}{\norm{\xVec}_q}
		:=\frac{1}{100}\sum_{k\in\SetOf{100}}\frac{\norm{\xVec_k-\xVec^\#_k}_q}{\norm{\xVec_k}_q}
		\TextAnd
		\frac{\norm{\xVec-\xVec^\#}_q}{\norm{\xVec}_q\norm{\eVec}_q}
		:=\frac{1}{100}\sum_{k\in\SetOf{100}}\frac{\norm{\xVec_k-\xVec^\#_k}_q}{\norm{\xVec_k}_q\norm{\eVec_k}_q},
	\end{align*}
	where the left hand sides are understood as an assigned symbol.
\end{Experiment}
Note that in this experiment the relative and absolute $\ell_q$ estimation errors coincide, i.e.,
$\frac{\norm{\xVec_k-\xVec^\#_k}_q}{\norm{\xVec_k}_q}=\norm{\xVec_k-\xVec^\#_k}_q$
since $\xVec_k$ is normalized.
\subsection{Tuning Parameter}\label{Subsection:Numerics_Misstuning}
\begin{figure}[ht]
    \begin{subfigure}[t]{0.5\textwidth}
		\includegraphics[scale=0.5]{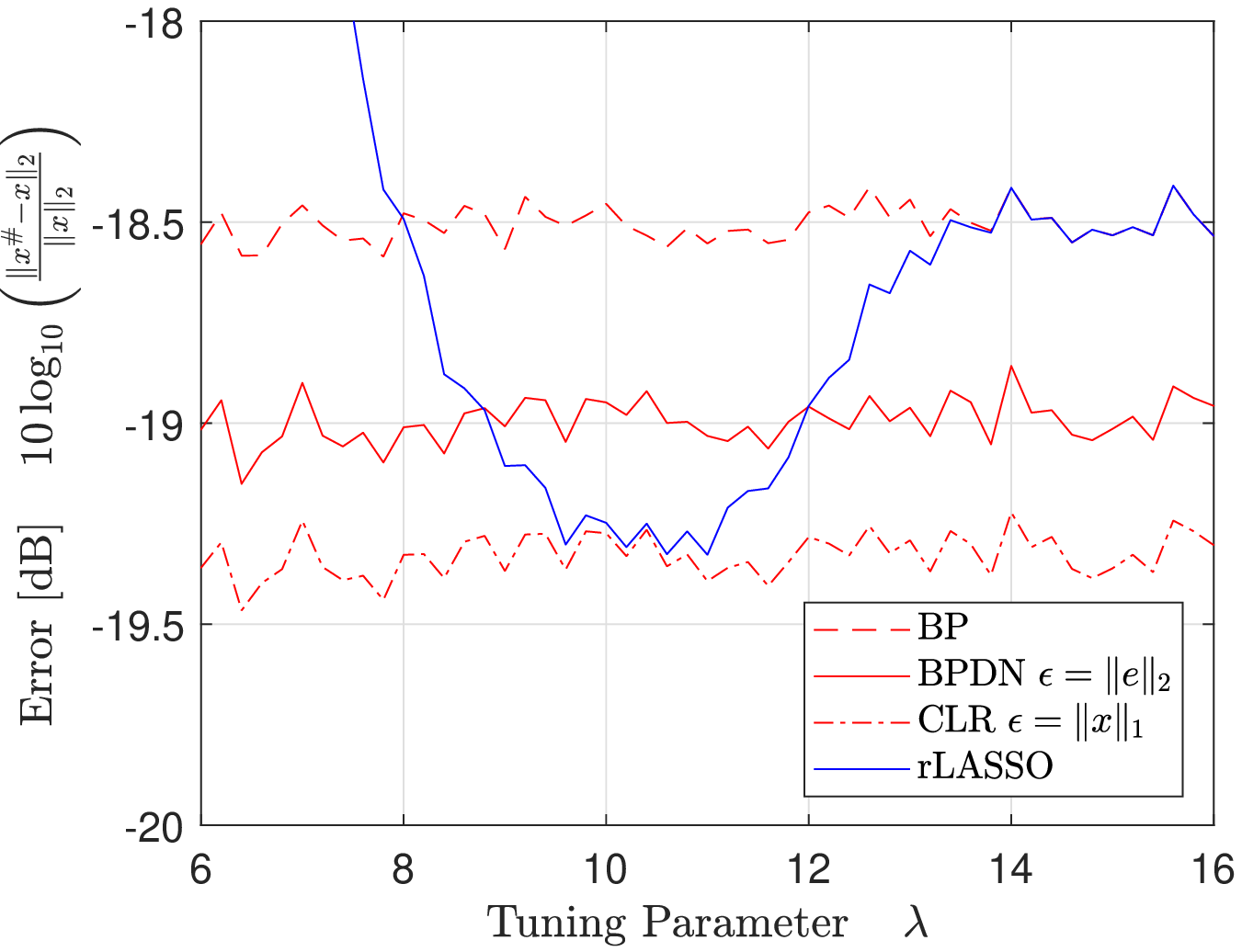}
		\subcaption{\label{Figure:Numerics:Misstuning1}
			Performance of \hrlone{} as a function of $\hrloneParam$.}
    \end{subfigure}
    \begin{subfigure}[t]{0.5\textwidth}
		\includegraphics[scale=0.5]{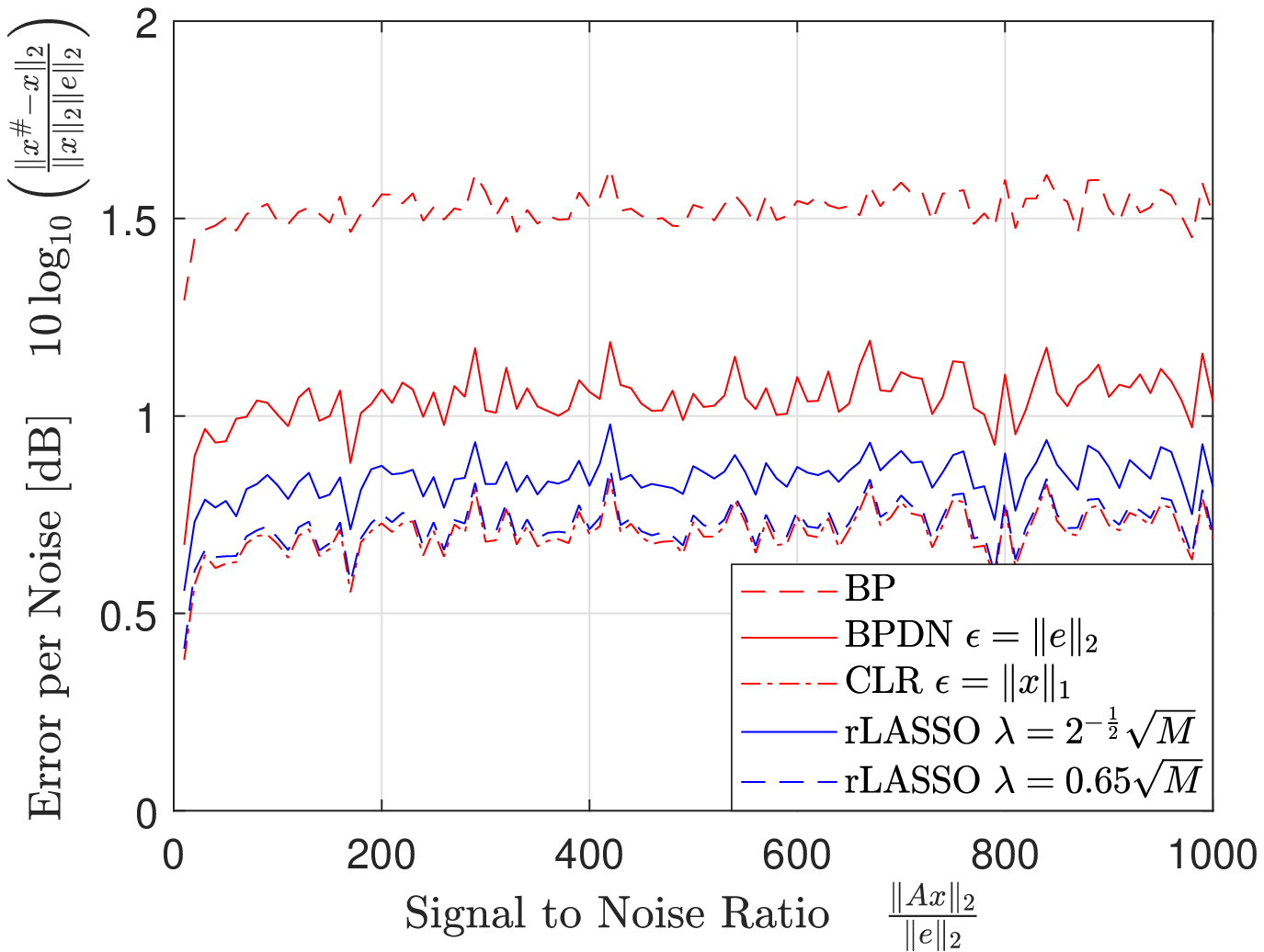}
		\subcaption{\label{Figure:Numerics:NoiseBlind1}
			Performance of \hrlone{} as a function of $SNR=\frac{\norm{\AMat\xVec}_2}{\norm{\eVec}_2}$.}
    \end{subfigure}
    \caption{Impact of the tuning parameter and noise-blindness of \hrlone{}.}
\end{figure}
In order to analyze the threshold from \thref{Theorem:HRL1Recovery:Lambda>TauSq}
and the results of \thref{Theorem:ConvergenceHRL1Param} and \thref{Proposition:HRL1Recovery:QP}
we set $q=2$, i.e., we consider the Gaussian case
since Gaussian matrices suffice a good quotient property with high probability
\cite[Theorem~11.19]{IntroductionCS}.
We fix the parameters $N=1024$, $M=256$, $S=32$, $SNR=100$ and vary the tuning parameter
$\hrloneParam\in \frac{2}{10}\SetOf{100}$ of \hrlone{} in \thref{Experiment:Experiment1}.
The results are plotted in \refP{Figure:Numerics:Misstuning1}.
For $\hrloneParam\leq 8$ the recovery with \hrlone{} seems to fail.
For $\hrloneParam\geq 14$ the recovery with \hrlone{} succeeds, but
the estimation errors of BP and \hrlone{} are the same,
which suggests that they return a similar minimizer
as proposed in \thref{Theorem:ConvergenceHRL1Param}.
For $\hrloneParam\rightarrow\infty$ the estimation error does not
diverge to infinity and is capped by the estimation error from the basis pursuit,
which is a consequence of the quotient property.
This coincides with the results of \thref{Proposition:HRL1Recovery:QP}.
The recovery succeeds roughly for $\hrloneParam\geq 8$,
which suggests that $\tau_1^0=8$ and thus $\tau_2^0\geq\frac{\tau_1^0}{\sqrt{S}}=\sqrt{2}$.
The optimal estimation error is achieved at
$\hrloneParam=10.5$.
\subsection{Noise Blindness}\label{Subsection:Numerics_NoiseBlind}
We will verify that the choice of $\hrloneParam$ is indeed independent on the
noise power. For this we consider the Gaussian case $q=2$.
We fix $N=1024$, $M=256$, $S=32$ and vary the signal to noise ratio $SNR\in 10\SetOf{100}$ in \thref{Experiment:Experiment1}.
The resulting errors are plotted in \refP{Figure:Numerics:NoiseBlind1}.
Note that due to a quotient property and the optimal tuning
all decoders achieve robust recovery guarantees.
In particular, since $\norm{\xVec_k}_q=1$, \thref{Theorem:Gaussian=>NSP}
and corresponding results yield that for each decoder the quantity
\begin{align*}
	\frac{\norm{\xVec-\xVec^\#}_q}{\norm{\xVec}_q\norm{\eVec}_q}
		=\frac{1}{100}\sum_{k\in\SetOf{100}}\frac{\norm{\xVec_k-\xVec^\#_k}_q}{\norm{\xVec_k}_q\norm{\eVec_k}_q}
		=\frac{1}{100}\sum_{k\in\SetOf{100}}\frac{\norm{\xVec_k-\xVec^\#_k}_q}{\norm{\eVec_k}_q}
\end{align*}
should be bounded by some constant $D$ independent on the signal to noise ratio $SNR$.
Indeed, we can see that this quantity stays constant and even the relative proportions
between all decoders stay constant. Remarkably, \hrlone{} with the tuning $\hrloneParam=0.65\sqrt{M}$
achieves the same estimation errors as the optimally tuned CLR and even better estimation errors
than the optimally tuned BPDN. However, \hrlone{} does not require any prior information
about the noise or the signal. We deduce that \hrlone{} is noise-blind.
\subsection{Gaussian Matrices}\label{Subsection:Numerics_GaussianMatrices}
\begin{figure}[ht]
    \begin{subfigure}[t]{0.5\textwidth}
		\includegraphics[scale=0.5]{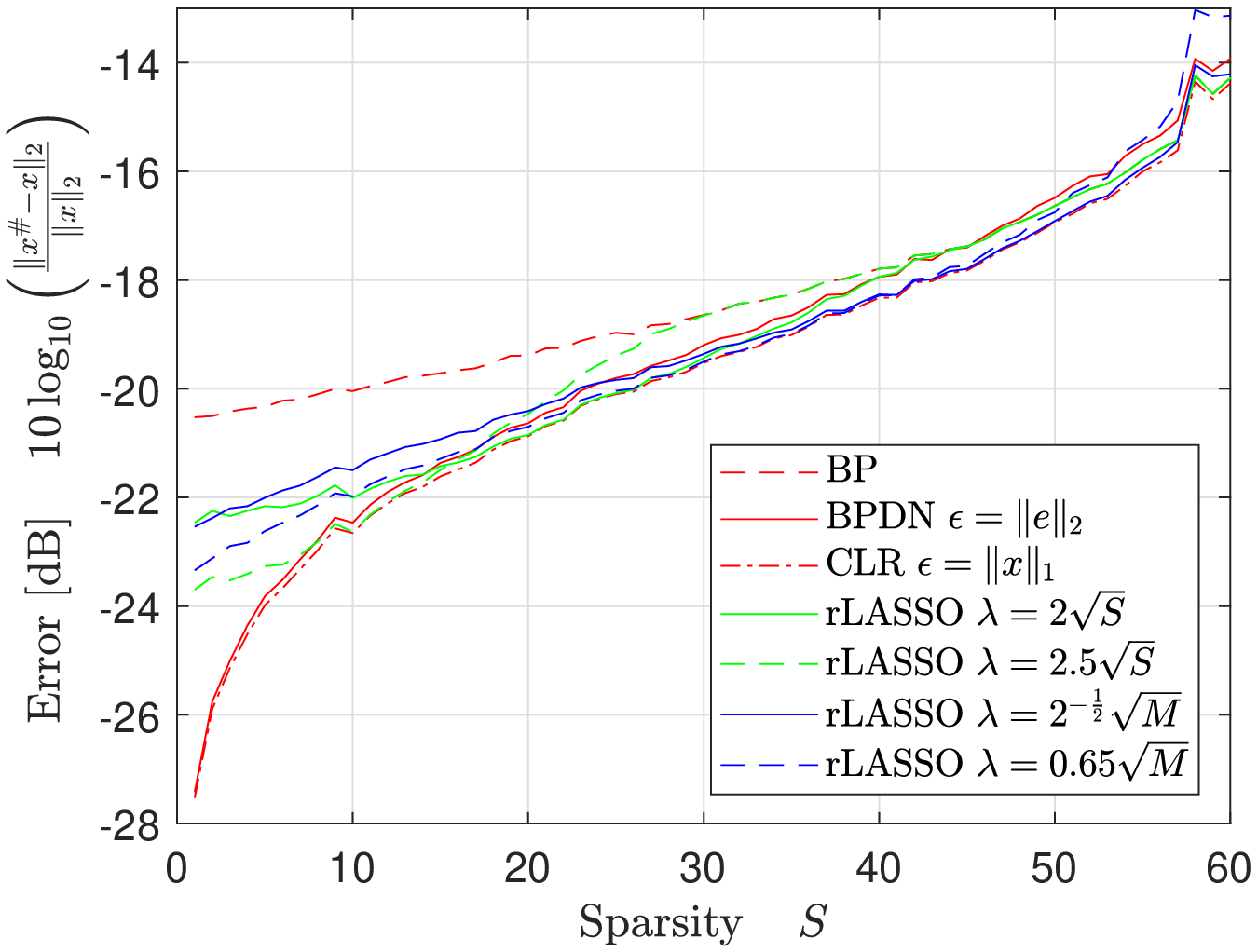}
		\subcaption{\label{Figure:Numerics:Gaussian_Sparsity}
			Performance of \hrlone{} as a function of the sparsity $S=\norm{\xVec}_0$.}
    \end{subfigure}
    \begin{subfigure}[t]{0.5\textwidth}
		\includegraphics[scale=0.5]{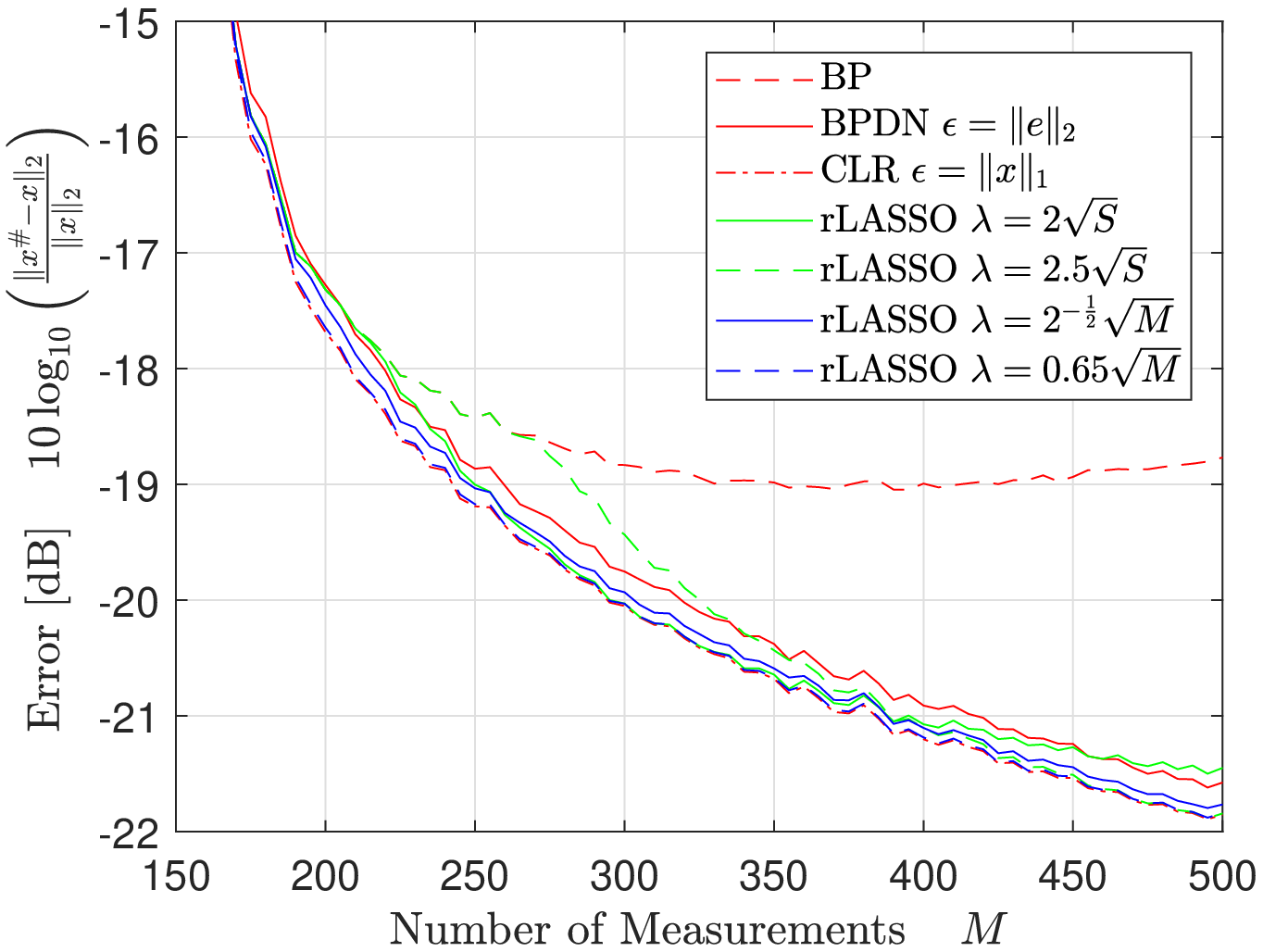}
		\subcaption{\label{Figure:Numerics:Gaussian_Measurements}
			Performance of \hrlone{} as a function of the number of measurements $M$.}
    \end{subfigure}
    \caption{Performance of different tuning methods for Gaussian matrices.}
\end{figure}
We investigate the dependence of the tuning parameter $\hrloneParam$
on the dimensions $M,N,S$ for Gaussian matrices as proposed
in \thref{Theorem:Gaussian=>NSP} and in
the argumentation around \refP{Equation:PhaseTransition:Theorem:Gaussian=>NSP}.
Thus, we set $q=2$.
We fix $N=1024$, $M=256$, $SNR=100$ and vary $S\in\SetOf{128}$ in \thref{Experiment:Experiment1}.
The results are plotted in \refP{Figure:Numerics:Gaussian_Sparsity}.
For $S\geq 60$ all estimation errors grow rapidly and we expect that the recovery fails.
Thus, we discarded plots for $S>60$.
As expected the CLR performs as a best case benchmark and the BP performs
like a worst case benchmark.
Against our expectation the BPDN only performs as a best case benchmark for $S\leq 15$.
The authors have no explanation for that yet.
For $S\geq 30$ the tuning $\hrloneParam=2.5\sqrt{S}$ and for $S\geq 45$ the tuning
$\hrloneParam=2\sqrt{S}$ perform exactly like the basis pursuit.
The tuning methods $\hrloneParam=2^{-\frac{1}{2}}\sqrt{M}$ and $\hrloneParam=0.65\sqrt{M}$
do not suffer from this problem.
In particular, they perform nearly as good as the best case benchmark CLR
for $S\geq 20$ and $S\geq 15$ respectively.
However, the tuning $\hrloneParam=0.65\sqrt{M}$ performs slighly better
than the tuning $\hrloneParam=2^{-\frac{1}{2}}\sqrt{M}$ for $S\leq 40$ while
the opposite is true for $S\geq 40$.
We deduce that the tuning parameter might be chosen independent of $S=\norm{\xVec}_0$
and still achieve fine error bounds and might even outperform tuning with
$\hrloneParam\asymp\sqrt{S}$.
Since for $S\leq 10$ all tuning methods of \hrlone{} perform sub-optimal,
we deduce that the optimal tuning parameter depends non-trivially on $S$.
As an alternative experiment we fix $N=1024,S=32,SNR=100$ and vary $M=5\SetOf{100}$ in
\thref{Experiment:Experiment1}.
The results are plotted in \refP{Figure:Numerics:Gaussian_Measurements}.
For $M\leq 270$ the tuning $\hrloneParam=2.5\sqrt{S}$ performs exactly like the BPDN
and the tuning $\hrloneParam=2\sqrt{S}$ has a similar problem for $M\leq 220$.
If $M\geq 400$, the tuning $\hrloneParam=2\sqrt{S}$ gets suboptimal.
The tunings $\hrloneParam=2^{-\frac{1}{2}}\sqrt{M}$ and $\hrloneParam=0.65\sqrt{M}$ do not
share these same problems. In particular, the tuning parameter $\hrloneParam=0.65\sqrt{M}$
seems to be indistinguishable from the optimal benchmark.
We deduce that $\hrloneParam\asymp\sqrt{M}$ might
reflect the behavior of the tuning parameter better than $\hrloneParam\asymp\sqrt{S}$.
If the number of measurements increases, the gap between \hrlone{} and BP increases.
This is due to the fact that the quotient property gets harder to fulfill
and thus the constant of the quotient property gets worse.
Thus, whenever the number of measurements is not optimal,
\hrlone{} works better than the basis pursuit with a quotient property
as it was proposed in the argumentation following \thref{Proposition:HRL1Recovery:QP}.
\subsection{Random Walk Matrices of Uniformly Distributed $D$-Left Regular Bipartite Graphs}\label{Subsection:Numerics_LRBGMatrices}
We will verify that the choice of tuning parameters from \thref{Theorem:RecoveryViaExpander}
is viable, and thus set $q=1$.
We fix $N=1024,M=256,D=10,SNR=100$ and vary $S\in\SetOf{128}$ in \thref{Experiment:Experiment1}.
The results are plotted in \refP{Figure:Numerics:LRBG1}.
\begin{figure}[ht]
	\centering
	\includegraphics[scale=0.5]{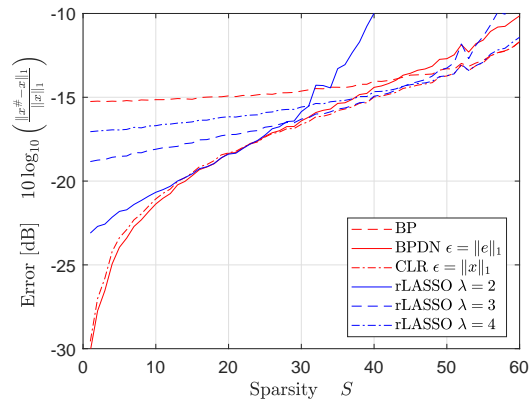}
	\caption{\label{Figure:Numerics:LRBG1}Performance of \hrlone{} as a function of the sparsity $S=\norm{\xVec}_0$
			for uniformly at random drawn $D$-left regular bipartite graph.}
\end{figure}
Similar to Gaussian matrices the estimation errors grow rapidly for $S\geq 60$
and we expect that the recovery fails.
For $S\geq35$ the tuning $\hrloneParam=2$ fails, and
for $S\geq55$ the tuning $\hrloneParam=3$ fails.
For smaller $S$, smaller tuning parameters perform better.
This suggests that the optimal tuning parameter is depending on at least the dimensions $S$
even though the requirement $\hrloneParam>\tau S^{1-\frac{1}{q}}=\tau$ suggests
otherwise. We deduce that recovery with \hrlone{} with a constant
tuning parameter is possible, but the optimal tuning is depending at least on some parameters.
A more detailed analysis, as it was done for Gaussian matrices in
\thref{Theorem:Gaussian=>NSP}, is required to understand the optimal tuning parameter better.
	
\section{Proofs of Section \ref{Section:Prelimninaries} Preliminaries}\label{Section:proofs_of_preliminaries}
The RNSP and the compatibility condition are equivalent.
\begin{Proposition}\label{Proposition:$S$-CC<=>$S$-SRNSP}
	Let $S\in\SetOf{N}$, $\AMat\in\mathbb{R}^{M\times N}$ and $\SetSize{T}\leq S$.
	Then the following statements are equivalent
	\begin{itemize}
		\item[(1)] There exists $\rho\in\left(0,1\right)$ and $\tau\in\left(0,\infty\right)$ such that
		\begin{align*}
			\norm{\ProjToIndex{T}{\vVec}}_1
			\leq \rho \norm{\ProjToIndex{T^c}{\vVec}}_1
				+\tau\norm{\AMat\vVec}_2
			\TextForAll \vVec\in\mathbb{R}^N
		\end{align*}
		holds true.			
		\item[(2)] There exists $L\in\left(1,\infty\right)$ such that for the set
		\begin{align*}
			\Delta_{L,T}
			:=\left\{\vVec\in\mathbb{R}^N:\norm{\ProjToIndex{T^c}{\vVec}}_1\leq L\norm{\ProjToIndex{T}{\vVec}}_1
				\TextAnd \norm{\ProjToIndex{T}{\vVec}}_1\neq 0\right\}
		\end{align*}
		the condition
		$
			\inf_{\vVec\in\Delta_{L,T}}\frac{S\norm{\AMat\vVec}_2^2}{\norm{\ProjToIndex{T}{\vVec}}_1^2}>0
		$
		holds true.
	\end{itemize}
	The constants may change and this change may depend on the dimensions.
\end{Proposition}
\begin{proof}
	(1)$\Rightarrow$(2):
	Pick any $L\in\left(1,\rho^{-1}\right)$ which is a non empty interval since $\rho<1$.
	Let $\vVec\in\Delta_{L,T}$ be arbitrary.
	The assumption yields
	\begin{align*}
		\norm{\ProjToIndex{T}{\vVec}}_1
		\leq \rho \norm{\ProjToIndex{T^c}{\vVec}}_1
			+\tau\norm{\AMat\vVec}_2
		\leq\rho L \norm{\ProjToIndex{T}{\vVec}}_1
			+\tau\norm{\AMat\vVec}_2
	\end{align*}
	and by algebraic manipulation
	$
		\frac{1-\rho L}{\tau}\leq\frac{\norm{\AMat\vVec}_2}{\norm{\ProjToIndex{T}{\vVec}}_1}
	$.
	By the choice of $L$ we get
	\begin{align}
		\inf_{\vVec\in\Delta_{L,T}}\frac{S\norm{\AMat\vVec}_2^2}{\norm{\ProjToIndex{T}{\vVec}}_1^2}\geq S\left(\frac{1-\rho L}{\tau}\right)^2>0.
	\end{align}
	(2)$\Rightarrow$(1):
	Set $\rho:=L^{-1}$ and
	$\tau:=S^\frac{1}{2}\left(\inf_{\vVec\in\Delta_{L,T}}\frac{S\norm{\AMat\vVec}_2^2}{\norm{\ProjToIndex{T}{\vVec}}_1^2}\right)^{-\frac{1}{2}} $.
	Then $\rho\in\left(0,1\right)$ and $\tau>0$.
	At first let $\vVec\notin\Delta_{L,T}$.
	If $\norm{\ProjToIndex{T}{\vVec}}_1=0$, the bound to prove holds trivially.
	So let $\norm{\ProjToIndex{T}{\vVec}}_1\neq 0$. Hence, we have for any $\tau'>0$
	\begin{align*}
		\norm{\ProjToIndex{T}{\vVec}}_1< L^{-1} \norm{\ProjToIndex{T^c}{\vVec}}_1
		\leq \rho \norm{\ProjToIndex{T^c}{\vVec}}_1 + \tau'\norm{\AMat\vVec}_2.
	\end{align*}
	Now on the other hand assume that $\vVec\in\Delta_{L,T}$. Then we have for any $\rho'>0$
	\begin{align*}
		\norm{\ProjToIndex{T}{\vVec}}_1
		=\left(\frac{\norm{\ProjToIndex{T}{\vVec}}_1^2}{S\norm{\AMat\vVec}_2^2}\right)^\frac{1}{2} S^\frac{1}{2}\norm{\AMat\vVec}_2
		\leq\left(\inf_{\vVec\in\Delta_{L,T}}\frac{S\norm{\AMat\vVec}_2^2}{\norm{\ProjToIndex{T}{\vVec}}_1^2}\right)^{-\frac{1}{2}}S^\frac{1}{2}\norm{\AMat\vVec}_2
		=\tau\norm{\AMat\vVec}_2
		\leq\rho' \norm{\ProjToIndex{T^c}{\vVec}}_1 + \tau\norm{\AMat\vVec}_2.
	\end{align*}
	This finishes the proof
	if $\rho'$ and $\tau'$ are the particular $\rho$ and $\tau$ above.
\end{proof}
\section{Proofs of Section \ref{Section:Theory} Theoretic Results for \hrlone{}}\label{Section:Theory:Proofs}
\subsection{Proofs of Subsection \ref{Subsection:Theory_RecoveryGuarantees} Recovery Guarantees for \hrlone{}}\label{Subsection:Theory_RecoveryGuarantees:Proofs}
The following statement is well known for $p,q\in\left[1,\infty\right)$ as Stechkin bound, see for instance
\cite[Proposition~2.3]{IntroductionCS}.
We note that it holds even for $p,q\in\left[1,\infty\right]$.
\begin{Lemma}[ {\cite[Proposition~2.3]{IntroductionCS}} ]\label{Lemma:STermApproximationBound}
	Let $p,q\in\left[1,\infty\right]$, $S\in\SetOf{N}$ and $\vVec\in\mathbb{R}^N$ with $q\geq p$. Then,
	\begin{align*}
		d_q\left(\vVec,\Sigma_S\right)
		\leq S^{\frac{1}{q}-\frac{1}{p}}\norm{\vVec}_p.
	\end{align*}
\end{Lemma}
\begin{proof}
	For $q,p\in\left[1,\infty\right)$ the statement is \cite[Proposition~2.3]{IntroductionCS}.
	Now let $q,p\in\left[1,\infty\right]$ with $q\geq p$ and $\vVec\in\mathbb{R}^N$.
	Let $T$ is the set of the $S$ indices with largest absolute value
	of $\vVec$.
	Then $\inf_{\zVec\in\Sigma_S}\norm{\vVec-\zVec}_q=\norm{\vVec-\ProjToIndex{T}{\vVec}}_q$.
	The proof now follows from a limit argument and
	the fact that $\lim_{r\rightarrow\infty} \norm{\vVec'}_r=\norm{\vVec'}_\infty$ holds true for all $\vVec'\in\mathbb{R}^N$,
	which is a consequence of
	$\norm{\vVec}_\infty\leq\norm{\vVec}_r\leq N^\frac{1}{r}\norm{\vVec}_\infty$.
\end{proof}
To prove recovery guarantees we require
a lemma which is proven in \cite[Theorem~4.25]{IntroductionCS} and \cite[Theorem~4.20]{IntroductionCS}.
With the improved Stechkin bound the proof of \cite[Theorem~4.25]{IntroductionCS} also
holds for $q=\infty$.
\begin{Lemma}[ {\cite[Theorem~4.25]{IntroductionCS}} \& {\cite[Theorem~4.20]{IntroductionCS}} ]\label{Lemma:SRNSPConsequence}
	Let $\AMat\in\mathbb{R}^{M\times N}$ have 
	$\ell_q$-RNSP of order $S$ wrt $\normRHS{\cdot}$
	with constants $\rho$ and $\tau$. Then, for all $\xVec',\zVec\in\mathbb{R}^N$ it holds that
	\begin{align*}
		\norm{\zVec-\xVec'}_q
		\leq
		\begin{Bmatrix}
			\frac{\left(1+\rho\right)^2}{1-\rho}S^{\frac{1}{q}-1}
				\left(\norm{\zVec}_1-\norm{\xVec'}_1+2\compr{\xVec'}{S}\right)
				+\frac{3+\rho}{1-\rho}\tau\normRHS{\AMat\xVec'-\AMat\zVec}
			& \TextIf & q\in\left(1,\infty\right]
			\\
			\frac{1+\rho}{1-\rho}
				\left(\norm{\zVec}_1-\norm{\xVec'}_1+2\compr{\xVec'}{S}\right)
				+\frac{2}{1-\rho}\tau\normRHS{\AMat\xVec'-\AMat\zVec}
			& \TextIf & q=1
		\end{Bmatrix}.
	\end{align*}
\end{Lemma}
We prove an auxiliary statement that looks similar to
\thref{Theorem:HRL1Recovery:Lambda>TauSq}.
\begin{Theorem}[Weak bound on $\hrloneParam$]\label{Theorem:HRL1Recovery}
	Let $\AMat\in\mathbb{R}^{M\times N}$ have 
	$\ell_q$-RNSP of order $S$ wrt $\normRHS{\cdot}$
	with constants $\rho$ and $\tau$. Let
	\begin{align*}
		\hrloneParam\geq\begin{Bmatrix}
			\frac{3+\rho}{\left(1+\rho\right)^2}\tau S^{1-\frac{1}{q}} & \TextIf & q\in\left(1,\infty\right]
			\\\frac{2}{1+\rho}\tau & \TextIf & q=1
		\end{Bmatrix}.
	\end{align*}
	Then for all $\xVec\in\mathbb{R}^N$ and $\yVec\in\mathbb{R}^M$
	any minimizer $\xVec^\#$ of 
	\begin{align*}
		\tempmin{\zVec\in\mathbb{R}^n}
			\norm{\zVec}_1
			+\hrloneParam\normRHS{\yVec-\AMat\zVec}
	\end{align*}
	obeys
	\begin{align*}
		\norm{\xVec^\#-\xVec}_q
		\leq\begin{Bmatrix}
			2\frac{\left(1+\rho\right)^2}{1-\rho}S^{\frac{1}{q}-1}\compr{\xVec}{S}
				+\left(\frac{3+\rho}{1-\rho}\tau+\frac{\left(1+\rho\right)^2}{1-\rho}S^{\frac{1}{q}-1}\hrloneParam\right)
				\normRHS{\yVec-\AMat\xVec}
			& \TextIf & q\in\left(1,\infty\right]
			\\
			2\frac{1+\rho}{1-\rho}\compr{\xVec}{S}
				+\left(\frac{2}{1-\rho}\tau+\frac{1+\rho}{1-\rho}\hrloneParam\right)\normRHS{\yVec-\AMat\xVec}
			& \TextIf & q=1
		\end{Bmatrix}.
	\end{align*}
\end{Theorem}
\begin{proof}
	We apply \thref{Lemma:SRNSPConsequence} for $q\in\left(1,\infty\right]$
	with $\xVec':=\xVec$ and $\zVec:=\xVec^\#$ and obtain
	\begin{align*}
		\norm{\xVec^\#-\xVec}_q
		\leq& \frac{\left(1+\rho\right)^2}{1-\rho}S^{\frac{1}{q}-1}
				\left(\norm{\xVec^\#}_1-\norm{\xVec}_1+2\compr{\xVec}{S}\right)
			+\frac{3+\rho}{1-\rho}\tau\normRHS{\AMat\xVec-\AMat\xVec^\#}
		\\\leq& \frac{\left(1+\rho\right)^2}{1-\rho}S^{\frac{1}{q}-1}
				\left(\norm{\xVec^\#}_1-\norm{\xVec}_1+2\compr{\xVec}{S}\right)
			+\frac{3+\rho}{1-\rho}\tau\normRHS{\yVec-\AMat\xVec^\#}
			+\frac{3+\rho}{1-\rho}\tau\normRHS{\yVec-\AMat\xVec}
		\\=& 2\frac{\left(1+\rho\right)^2}{1-\rho}S^{\frac{1}{q}-1}\compr{\xVec}{S}
			-\frac{\left(1+\rho\right)^2}{1-\rho}S^{\frac{1}{q}-1}\norm{\xVec}_1
			+\frac{3+\rho}{1-\rho}\tau\normRHS{\yVec-\AMat\xVec}
		\\&+\frac{\left(1+\rho\right)^2}{1-\rho}S^{\frac{1}{q}-1}\left(
				\norm{\xVec^\#}_1
				+\frac{\left(3+\rho\right)}{\left(1+\rho\right)^2}S^{1-\frac{1}{q}}\tau\normRHS{\yVec-\AMat\xVec^\#}
			\right)
		\\\leq& 2\frac{\left(1+\rho\right)^2}{1-\rho}S^{\frac{1}{q}-1}\compr{\xVec}{S}
			-\frac{\left(1+\rho\right)^2}{1-\rho}S^{\frac{1}{q}-1}\norm{\xVec}_1
			+\frac{3+\rho}{1-\rho}\tau\normRHS{\yVec-\AMat\xVec}
		\\&+\frac{\left(1+\rho\right)^2}{1-\rho}S^{\frac{1}{q}-1}\left(
				\norm{\xVec^\#}_1
				+\hrloneParam\normRHS{\yVec-\AMat\xVec^\#}
			\right).
	\end{align*}
	Since $\xVec$ is feasible and $\xVec^\#$ is a minimizer, we get
	\begin{align*}
		\norm{\xVec^\#-\xVec}_q
		\leq& 2\frac{\left(1+\rho\right)^2}{1-\rho}S^{\frac{1}{q}-1}\compr{\xVec}{S}
			-\frac{\left(1+\rho\right)^2}{1-\rho}S^{\frac{1}{q}-1}\norm{\xVec}_1
			+\frac{3+\rho}{1-\rho}\tau\normRHS{\yVec-\AMat\xVec}
		\\&+\frac{\left(1+\rho\right)^2}{1-\rho}S^{\frac{1}{q}-1}\left(
				\norm{\xVec}_1
				+\hrloneParam\normRHS{\yVec-\AMat\xVec}
			\right)
		\\=& 2\frac{\left(1+\rho\right)^2}{1-\rho}S^{\frac{1}{q}-1}\compr{\xVec}{S}
			+\left(\frac{3+\rho}{1-\rho}\tau+\frac{\left(1+\rho\right)^2}{1-\rho}S^{\frac{1}{q}-1}\hrloneParam\right)\normRHS{\yVec-\AMat\xVec}.
	\end{align*}
	This proves the case $q\in\left(1,\infty\right]$.
	For the other case we apply \thref{Lemma:SRNSPConsequence} for $q=1$
	with $\xVec':=\xVec$ and $\zVec:=\xVec^\#$ and obtain
	\begin{align*}
		\norm{\xVec^\#-\xVec}_1
		\leq& \frac{1+\rho}{1-\rho}
				\left(\norm{\xVec^\#}_1-\norm{\xVec}_1+2\compr{\xVec}{S}\right)
			+\frac{2}{1-\rho}\tau\normRHS{\AMat\xVec-\AMat\xVec^\#}
		\\\leq& \frac{1+\rho}{1-\rho}
				\left(\norm{\xVec^\#}_1-\norm{\xVec}_1+2\compr{\xVec}{S}\right)
			+\frac{2}{1-\rho}\tau\normRHS{\yVec-\AMat\xVec^\#}
			+\frac{2}{1-\rho}\tau\normRHS{\yVec-\AMat\xVec}
		\\=& 2\frac{1+\rho}{1-\rho}\compr{\xVec}{S}
			-\frac{1+\rho}{1-\rho}\norm{\xVec}_1
			+\frac{2}{1-\rho}\tau\normRHS{\yVec-\AMat\xVec}
		+\frac{1+\rho}{1-\rho}\left(
				\norm{\xVec^\#}_1
				+\frac{2}{1+\rho}\tau\normRHS{\yVec-\AMat\xVec^\#}
			\right)
		\\\leq& 2\frac{1+\rho}{1-\rho}\compr{\xVec}{S}
			-\frac{1+\rho}{1-\rho}\norm{\xVec}_1
			+\frac{2}{1-\rho}\tau\normRHS{\yVec-\AMat\xVec}
		+\frac{1+\rho}{1-\rho}\left(
				\norm{\xVec^\#}_1
				+\hrloneParam\normRHS{\yVec-\AMat\xVec^\#}
			\right).
	\end{align*}
	Since $\xVec$ is feasible and $\xVec^\#$ is a minimizer, we obtain
	\begin{align*}
		\norm{\xVec^\#-\xVec}_1
		\leq& 2\frac{1+\rho}{1-\rho}\compr{\xVec}{S}
			-\frac{1+\rho}{1-\rho}\norm{\xVec}_1
			+\frac{2}{1-\rho}\tau\normRHS{\yVec-\AMat\xVec}
		+\frac{1+\rho}{1-\rho}\left(
				\norm{\xVec}_1
				+\hrloneParam\normRHS{\yVec-\AMat\xVec}
			\right)
		\\=& 2\frac{1+\rho}{1-\rho}\compr{\xVec}{S}
			+\left(\frac{2}{1-\rho}\tau+\frac{1+\rho}{1-\rho}\hrloneParam\right)\normRHS{\yVec-\AMat\xVec}.
	\end{align*}
	This proves the case $q=1$.
\end{proof}
It turns out that the bounds on $\hrloneParam$ of \thref{Theorem:HRL1Recovery}
are not tight. We can even use smaller parameters in a trade-off with possibly worse bounds for the estimation error.
The reason for this is that the functions $\rho\mapsto\frac{3+\rho}{\left(1+\rho\right)^2}$ and $\rho\mapsto \frac{2}{1+\rho}$
are monotonically decreasing on $\left[0,1\right)$. Thus, if we artificially assume that $\AMat$ has a worse stableness constant $\rho'$,
the bound on $\hrloneParam$ will get loosened and we can deduce
\thref{Theorem:HRL1Recovery:Lambda>TauSq} from \thref{Theorem:HRL1Recovery}.
\begin{proof}[Proof of \thref{Theorem:HRL1Recovery:Lambda>TauSq}]
	Note that the RNSP is preserved
	under increases in the stableness constant $\rho$.
	Hence, $\rho'\in\left[\rho,1\right)$
	yields that
	$\AMat$ has the 
	$\ell_q$-RNSP of order $S$ wrt $\normRHS{\cdot}$
	with constants $\rho'$ and $\tau$, and then
	the error bound of \thref{Theorem:HRL1Recovery:Lambda>TauSq}
	follows from \thref{Theorem:HRL1Recovery}.
	It is thus sufficient to prove $\rho'\in\left[\rho,1\right)$
	and the ``in particular part''.
	At first let $q\in\left(1,\infty\right]$.
	Note that the function $\fFunc{}:\left[0,1\right]\rightarrow\left[1,3\right]$ that maps
	$t$ to $\fFunc{t}:=\frac{3+t}{\left(1+t\right)^2}$ is by differentiation
	strictly monotonically decreasing.
	Hence, it is invertible, its inverse is also strictly monotonically decreasing
	and its inverse is given by
	$\gFunc{}:\left[1,3\right]\rightarrow\left[0,1\right]$
	which maps $r$ to $\gFunc{r}:=\frac{1}{2r}\left(1+\sqrt{8r+1}\right)-1$.
	Now $\hrloneParam\geq \frac{3+\rho}{\left(1+\rho\right)^2}\tau S^{1-\frac{1}{q}}$ is equivalent
	to $\frac{\hrloneParam}{\tau}S^{\frac{1}{q}-1}\geq\fFunc{\rho}$ and by the strict monotonicity
	we get the logical statement
	\begin{align}\label{Equation:Equivalence:Theorem:HRL1Recovery:Lambda>TauSq}
		\hrloneParam\geq \frac{3+\rho}{\left(1+\rho\right)^2}\tau S^{1-\frac{1}{q}}
		\hspace{5pt}\Leftrightarrow\hspace{5pt}
		\gFunc{\frac{\hrloneParam}{\tau}S^{\frac{1}{q}-1}}
		\leq \rho.
	\end{align}
	The definition of $\rho'$ can be rewritten as $\rho'=\max\left\{\rho,\gFunc{\frac{\hrloneParam}{\tau}S^{\frac{1}{q}-1}}\right\}$.
	By \refP{Equation:Equivalence:Theorem:HRL1Recovery:Lambda>TauSq}
	we get the ``in particular part'', namely
	that $\hrloneParam\geq \frac{3+\rho}{\left(1+\rho\right)^2}\tau S^{1-\frac{1}{q}}$ is equivalent to $\rho'=\rho$.
	In order to prove $\rho'\in\left[\rho,1\right)$,
	we distinguish two cases.
	If $\hrloneParam\geq \frac{3+\rho}{\left(1+\rho\right)^2}\tau S^{1-\frac{1}{q}}$,
	then \refP{Equation:Equivalence:Theorem:HRL1Recovery:Lambda>TauSq} yields that
	$\rho'=\max\left\{\rho,\gFunc{\frac{\hrloneParam}{\tau}S^{\frac{1}{q}-1}}\right\}=\rho$
	and thus $\rho'\in\left[\rho,1\right)$.
	If $\hrloneParam<\frac{3+\rho}{\left(1+\rho\right)^2}\tau S^{1-\frac{1}{q}}$,
	then \refP{Equation:Equivalence:Theorem:HRL1Recovery:Lambda>TauSq} yields that
	$\rho'=\max\left\{\rho,\gFunc{\frac{\hrloneParam}{\tau}S^{\frac{1}{q}-1}}\right\}=\gFunc{\frac{\hrloneParam}{\tau}
		S^{\frac{1}{q}-1}}> \rho$.
	Now by assumption $\hrloneParam>\tau S^{1-\frac{1}{q}}$, which is equivalent to $\frac{\hrloneParam}{\tau}S^{\frac{1}{q}-1}>1$.
	By the strict monotonicity of $\gFunc{}$ we have
	$\rho'=\gFunc{\frac{\hrloneParam}{\tau}S^{\frac{1}{q}-1}}<\gFunc{1}=1$
	and thus $\rho'\in\left[\rho,1\right)$.
	The case $q=1$ works similarly by choosing
	$\fFunc{t}:=\frac{2}{1+t}$ since the inverse is $\gFunc{r}:=\frac{2}{r}-1$.
\end{proof}
From \thref{Theorem:HRL1Recovery} we can deduce that \hrlone{} defines an SRD.
\begin{proof}[Proof of \thref{Corollary:S-SRNSP=>S-SRD}]
	Let $q>1$ and set
	$\hrloneParam:=\frac{3+\rho}{\left(1+\rho\right)^2}\tau S^{1-\frac{1}{q}}$.
	For $\yVec$ we set $\Decoder{\yVec}$ as any minimizer of \refP{Problem:HRL1} with input $\yVec$ and $\AMat$.
	This defines a mapping $\Decoder{}:\mathbb{R}^M\rightarrow\mathbb{R}^N$. By \thref{Theorem:HRL1Recovery}
	with $q\in\left(1,\infty\right]$ we have
	for all $\xVec\in\mathbb{R}^N$ and $\yVec\in\mathbb{R}^M$ that
	\begin{align*}
		\norm{\Decoder{\yVec}-\xVec}_q
		\leq& 2\frac{\left(1+\rho\right)^2}{1-\rho}S^{\frac{1}{q}-1}\compr{\xVec}{S}
			+\left(\frac{3+\rho}{1-\rho}\tau+\frac{\left(1+\rho\right)^2}{1-\rho}S^{\frac{1}{q}-1}\hrloneParam\right)\normRHS{\yVec-\AMat\xVec}
		\\=& 2\frac{\left(1+\rho\right)^2}{1-\rho}S^{\frac{1}{q}-1}\compr{\xVec}{S}
			+2\frac{3+\rho}{1-\rho}\tau\normRHS{\yVec-\AMat\xVec}
		=CS^{\frac{1}{q}-1}\compr{\xVec}{S} + D\normRHS{\yVec-\AMat\xVec}.
	\end{align*}
	Thus, $\Decoder{}$ is an $\ell_q$-SRD
	of order $S$ wrt $\normRHS{\cdot}$ for $\AMat$ with constants $C,D$.
	For the ``in particular part'' we repeat the same steps with
	$\hrloneParam=\frac{2}{1+\rho}\tau$ and use
	the bound of \thref{Theorem:HRL1Recovery} with $q=1$ instead.
\end{proof}
\subsection{Proofs of Subsection \ref{Subsection:Theory_ConvergenceToBP} Asymptotic Analysis for \hrlone{}}\label{Subsection:Theory_ConvergenceToBP:Proofs}
To prove the claimed convergence, we prove an auxiliary statement.
\begin{Lemma}\label{Lemma:ConvergenceHRL1Param}
	Let $\normRHS{\cdot}$ be any norm on $\mathbb{R}^M$, $\AMat\in\mathbb{R}^{M\times N}$
	and $\yVec\in\mathbb{R}^M$.
	For every $\hrloneParam\in\left[0,\infty\right)$ let $\xVec^\hrloneParam$ be any minimizer of 
	\begin{align*}
		\tempmin{\zVec\in\mathbb{R}^N}
			\norm{\zVec}_1
			+\hrloneParam\normRHS{\yVec-\AMat\zVec}.
	\end{align*}
	Then \refP{Problem:BPImp} has a minimizer $\xVec^{BPImp}$ and the following statements hold true:
	\begin{itemize}
		\item[(a)] The function $\hrloneParam\mapsto \norm{\xVec^\hrloneParam}_1$ is monotonically increasing.
		\item[(b)] The function $\hrloneParam\mapsto \normRHS{\AMat\xVec^\hrloneParam-\yVec}$ is monotonically decreasing.
		\item[(c)] The estimators are bounded
			$\norm{\xVec^\hrloneParam}_1\leq\norm{\xVec^{BPImp}}_1$.
		\item[(d)] The residuals are bounded
			\begin{align}
				\normRHS{\AMat\xVec^\hrloneParam-\yVec}
				\leq \normRHS{\AMat\xVec^{BPImp}-\yVec}+\hrloneParam^{-1}\left(\norm{\xVec^{BPImp}}_1-\norm{\xVec^\hrloneParam}_1\right).
				\label{Equation:EQ1:Lemma:ConvergenceHRL1Param}
			\end{align}
		\item[(e)] If $\norm{\xVec^\hrloneParam}_1\geq\norm{\xVec^{BPImp}}_1$,
			then $\xVec^\hrloneParam$ is a minimizer of \refP{Problem:BPImp}.
		\item[(f)] If $\normRHS{\AMat\xVec^\hrloneParam-\yVec}\leq\normRHS{\AMat\xVec^{BPImp}-\yVec}$,
			then $\xVec^\hrloneParam$ is a minimizer of \refP{Problem:BPImp}.
	\end{itemize}
\end{Lemma}
\begin{proof}
	At first we will prove that \refP{Problem:BPImp} indeed has an optimizer.
	Let $\left(\zVec'_n\right)_{n\in\mathbb{N}}$ be a sequence such that
	$\lim_{n\rightarrow\infty}\normRHS{\AMat\zVec'_n-\yVec}
		=\inf_{\zVec'\in\mathbb{R}^N}\normRHS{\AMat\zVec'-\yVec}$.
	Let $\zVec''_n$ be the orthogonal projection of $\zVec'_n$ onto
	$\Kernel{\AMat}^\perp$.
	Since $\AMat$ is injective on the finite-dimensional space
	$\Kernel{\AMat}^\perp$, it is also bounded below on this,
	i.e., there exists some $C>0$ such that
	$\normRHS{\AMat\zVec''_n}\geq C\norm{\zVec''_n}_2 \TextForAll n\in\SetOf{N}$.
	Hence,
	\begin{align*}
		\norm{\zVec''_n}_2
		\leq C^{-1}\normRHS{\AMat\zVec''_n}
		\leq C^{-1}\left(\normRHS{\AMat\zVec''_n-\yVec}+\normRHS{\yVec}\right).
	\end{align*}
	This together with the convergence of $\normRHS{\AMat\zVec''_n-\yVec}$ yields
	that $\norm{\zVec''_n}_2$ is bounded and thus
	$\left(\zVec''_n\right)_{n\in\mathbb{N}}$ contains some subsequence
	$\left(\zVec'''_n\right)_{n\in\mathbb{N}}$
	that converges to some $\zVec'''_\infty$. It follows that
	\begin{align*}
		\inf_{\zVec'\in\mathbb{R}^N}\normRHS{\AMat\zVec'-\yVec}
		=\lim_{n\rightarrow\infty}\normRHS{\AMat\zVec'_n-\yVec}
		=\lim_{n\rightarrow\infty}\normRHS{\AMat\zVec''_n-\yVec}
		=\lim_{n\rightarrow\infty}\normRHS{\AMat\zVec'''_n-\yVec}
		=\normRHS{\AMat\zVec'''_\infty-\yVec}.
	\end{align*}
	Thus, the problem $\tempmin{\zVec'\in\mathbb{R}^N}\normRHS{\AMat\zVec'-\yVec}$
	has a minimizer. Since the objective function is continuous and
	the set of feasible vectors is closed,
	$\argmin{\zVec'\in\mathbb{R}^N}\normRHS{\AMat\zVec'-\yVec}$
	is closed and non-empty.
	Now let $\left(\zVec_n\right)_{n\in\mathbb{N}}$ be a sequence such that
	$\zVec_n\in\argmin{\zVec'\in\mathbb{R}^N}\normRHS{\AMat\zVec'-\yVec}$
	and $\lim_{n\rightarrow\infty}\norm{\zVec_n}_1
		=\inf_{\zVec\in\argmin{\zVec'\in\mathbb{R}^N}\normRHS{\AMat\zVec'-\yVec}}
			\norm{\zVec}_1$,
	which is finite since there are feasible points.
	The sequence $\left(\zVec_n\right)_{n\in\mathbb{N}}$ is thus bounded and contains
	a subsequence $\left(\zVec'_n\right)_{n\in\mathbb{N}}$
	that converges to some $\xVec^{BPImp}$, which lies in
	$\argmin{\zVec'\in\mathbb{R}^N}\normRHS{\AMat\zVec'-\yVec}$ due to closedness.
	Hence, $\xVec^{BPImp}$ is feasible for \refP{Problem:BPImp} and
	\begin{align*}
		\norm{\xVec^{BPImp}}_1
		=\lim_{n\rightarrow\infty}\norm{\zVec'_n}_1
		=\lim_{n\rightarrow\infty}\norm{\zVec_n}_1
		=\inf_{\zVec\in\argmin{\zVec'\in\mathbb{R}^N}\normRHS{\AMat\zVec'-\yVec}}
			\norm{\zVec}_1.
	\end{align*}
	Thus, \refP{Problem:BPImp} has some minimizer $\xVec^{BPImp}$.
	We now prove the remaining statements.\\
	(b): Let $\hrloneParam'>\hrloneParam$.
	We use the optimality of $\xVec^{\hrloneParam}$ and $\xVec^{\hrloneParam'}$ to obtain
	\begin{align*}
		\norm{\xVec^{\hrloneParam}}_1+{\hrloneParam}\normRHS{\AMat\xVec^{\hrloneParam}-\yVec}
		\leq& \norm{\xVec^{\hrloneParam'}}_1+{\hrloneParam}\normRHS{\AMat\xVec^{\hrloneParam'}-\yVec}
		= \norm{\xVec^{\hrloneParam'}}_1+{\hrloneParam'}\normRHS{\AMat\xVec^{\hrloneParam'}-\yVec}
			-\left(\hrloneParam'-\hrloneParam\right)\normRHS{\AMat\xVec^{\hrloneParam'}-\yVec}
		\\\leq& \norm{\xVec^{\hrloneParam}}_1+{\hrloneParam'}\normRHS{\AMat\xVec^{\hrloneParam}-\yVec}
			-\left(\hrloneParam'-\hrloneParam\right)\normRHS{\AMat\xVec^{\hrloneParam'}-\yVec}.
	\end{align*}
	Since $\hrloneParam'-\hrloneParam>0$, it follows that
	$\normRHS{\AMat\xVec^{\hrloneParam'}-\yVec}\leq\normRHS{\AMat\xVec^{\hrloneParam}-\yVec}$.\\
	(a): Let $\hrloneParam'>\hrloneParam$. We use the optimality of $\xVec^{\hrloneParam}$ 
	and statement (b) to get
	\begin{align*}
		\norm{\xVec^{\hrloneParam}}_1
		=\norm{\xVec^{\hrloneParam}}_1+{\hrloneParam}\normRHS{\AMat\xVec^{\hrloneParam}-\yVec}-{\hrloneParam}\normRHS{\AMat\xVec^{\hrloneParam}-\yVec}
		\leq \norm{\xVec^{\hrloneParam'}}_1+{\hrloneParam}\normRHS{\AMat\xVec^{\hrloneParam'}-\yVec}-{\hrloneParam}\normRHS{\AMat\xVec^{\hrloneParam}-\yVec}
		\overset{(2)}{\leq} \norm{\xVec^{\hrloneParam'}}_1.
	\end{align*}
	(c): We use the optimality of $\xVec^\hrloneParam$ and the feasibility of $\xVec^{BPImp}$ to obtain
	\begin{align*}
		\norm{\xVec^\hrloneParam}_1
		=&\norm{\xVec^\hrloneParam}_1+\hrloneParam\normRHS{\AMat\xVec^\hrloneParam-\yVec}-\hrloneParam\normRHS{\AMat\xVec^\hrloneParam-\yVec}
		\\\leq&\norm{\xVec^{BPImp}}_1+\hrloneParam\normRHS{\AMat\xVec^{BPImp}-\yVec}-\hrloneParam\normRHS{\AMat\xVec^\hrloneParam-\yVec}
		\leq\norm{\xVec^{BPImp}}_1.
	\end{align*}
	(d): We use the optimality of $\xVec^\hrloneParam$ to obtain
	\begin{align*}
		\normRHS{\AMat\xVec^\hrloneParam-\yVec}
		=&\hrloneParam^{-1}\left(\norm{\xVec^\hrloneParam}_1+\hrloneParam\normRHS{\AMat\xVec^\hrloneParam-\yVec}\right)
			-\hrloneParam^{-1}\norm{\xVec^\hrloneParam}_1
		\\\leq&\hrloneParam^{-1}\left(\norm{\xVec^{BPImp}}_1+\hrloneParam\normRHS{\AMat\xVec^{BPImp}-\yVec}\right)
			-\hrloneParam^{-1}\norm{\xVec^\hrloneParam}_1
		\\
		=&\normRHS{\AMat\xVec^{BPImp}-\yVec}+\hrloneParam^{-1}\left(\norm{\xVec^{BPImp}}_1-\norm{\xVec^\hrloneParam}_1\right).
	\end{align*}
	(f): Note that the assumption means that $\xVec^\hrloneParam$ is a feasible point of \refP{Problem:BPImp}
	and hence $\norm{\AMat\xVec^\hrloneParam-\yVec}=\norm{\AMat\xVec^{BPImp}-\yVec}$.
	This and the optimality of $\xVec^\hrloneParam$ yields
	\begin{align*}
		\norm{\xVec^\hrloneParam}_1
		=\norm{\xVec^\hrloneParam}_1+\hrloneParam\norm{\AMat\xVec^\hrloneParam-\yVec}-\hrloneParam\norm{\AMat\xVec^{BPImp}-\yVec}
		\leq\norm{\xVec^{BPImp}}_1
	\end{align*}
	and $\xVec^\hrloneParam$ is a minimizer of \refP{Problem:BPImp}.\\
	(e): We use statement (d) to get
	\begin{align*}
		\normRHS{\AMat\xVec^\hrloneParam-\yVec}
		\leq\normRHS{\AMat\xVec^{BPImp}-\yVec}+\hrloneParam^{-1}\left(\norm{\xVec^{BPImp}}_1-\norm{\xVec^\hrloneParam}_1\right)
		\leq\normRHS{\AMat\xVec^{BPImp}-\yVec}.		
	\end{align*}
	By statement (f) we obtain that $\xVec^\hrloneParam$ is a minimizer of \refP{Problem:BPImp}.
\end{proof}
This lemma allows us to deduce the results of \thref{Theorem:ConvergenceHRL1Param}.
\begin{proof}[Proof of \thref{Theorem:ConvergenceHRL1Param}]
	Let $\xVec^{BPImp}$ be any minimizer of \refP{Problem:BPImp} which exists by \thref{Lemma:ConvergenceHRL1Param}.\\
	(1): This statement follows from statements (e)  and (f)
	of \thref{Lemma:ConvergenceHRL1Param}.\\
	(2):
	By the feasibility of $\xVec^{BPImp}$ and statement (d) of \thref{Lemma:ConvergenceHRL1Param} we have
	\begin{align*}
		0
		\leq&\normRHS{\AMat\xVec^{\hrloneParam}-\yVec}-\normRHS{\AMat\xVec^{BPImp}-\yVec}
		\overset{\refP{Equation:EQ1:Lemma:ConvergenceHRL1Param}}{\leq}{\hrloneParam}^{-1}\left(\norm{\xVec^{BPImp}}_1-\norm{\xVec^{\hrloneParam}}_1\right)
		\\\leq&{\hrloneParam}^{-1}\norm{\xVec^{BPImp}}_1
		={\hrloneParam}^{-1}\left(\inf_{\zVec\in\argmin{\zVec'\in\mathbb{R}^N}\normRHS{\AMat\zVec'-\yVec}}\norm{\zVec}_1\right).
	\end{align*}
	This yields the convergence and the bound on the distance in
	\refP{Equation:Convergence:StoppingCriteriaResidual:Theorem:ConvergenceHRL1Param}.
	The monotonicity of the convergence follows from statement (b) of \thref{Lemma:ConvergenceHRL1Param}.
	Now let $\left(\hrloneParam_n\right)_{n\in\mathbb{N}}$ be a sequence with
	$\lim_{n\rightarrow\infty} \hrloneParam_n=\infty$.
	By statement (c) of \thref{Lemma:ConvergenceHRL1Param} the sequence
	$\left(\xVec^{\hrloneParam_n}\right)_{n\in\mathbb{N}}$ is bounded. Hence, there exists a subsequence
	$\left(\hrloneParam'_n\right)_{n\in\mathbb{N}}$ such that $\left(\xVec^{\hrloneParam'_n}\right)_{n\in\mathbb{N}}$
	converges to some $\xVec$.
	Equation
	\refP{Equation:Convergence:StoppingCriteriaResidual:Theorem:ConvergenceHRL1Param} yields
	that $\xVec$ is a feasible point for \refP{Problem:BPImp}.
	By the optimality of $\xVec^{BPImp}$ we get
	$
		\lim_{n\rightarrow\infty}\norm{\xVec^{\hrloneParam'_n}}_1
		=\norm{\xVec}_1
		\geq \norm{\xVec^{BPImp}}_1
	$,
	and by statement (c) of \thref{Lemma:ConvergenceHRL1Param}
	we also get
	$
		\norm{\xVec}_1
		=\lim_{n\rightarrow\infty}\norm{\xVec^{\hrloneParam'_n}}_1
		\leq \norm{\xVec^{BPImp}}_1
	$.
	Combining these two inequalities yields
	\begin{align*}
		\lim_{n\rightarrow\infty}\norm{\xVec^{\hrloneParam_n}}_1
		=\lim_{n\rightarrow\infty}\norm{\xVec^{\hrloneParam'_n}}_1
		=\norm{\xVec^{BPImp}}_1
		=\inf_{\zVec\in\argmin{\zVec'\in\mathbb{R}^N}\normRHS{\AMat\zVec'-\yVec}}\norm{\zVec}_1.
	\end{align*}
	Doing this for all possible sequences with $\lim_{n\rightarrow\infty} \hrloneParam_n=\infty$
	we obtain the convergence
	\begin{align*}
		\lim_{\hrloneParam\rightarrow\infty}\norm{\xVec^\hrloneParam}_1
		=\inf_{\zVec\in\argmin{\zVec'\in\mathbb{R}^N}\normRHS{\AMat\zVec'-\yVec}}\norm{\zVec}_1.
	\end{align*}
	The monotonicity of convergence follows from statement (a) of \thref{Lemma:ConvergenceHRL1Param}.\\
	(3): 
	Towards a contradiction assume that
	$\inf_{ \underset{ \text{\refP{Problem:BPImp}} }{ \text{$\zVec$ minimizer of} } }
		\norm{\xVec^\hrloneParam-\zVec}_2$ does not converge to zero.
	Then there exists a sequence $\left(\hrloneParam_n\right)_{n\in\mathbb{R}^N}$
	and $\epsilon>0$ with
	$\lim_{n\rightarrow\infty} \hrloneParam_n=\infty$ and
	\begin{align*}
		\inf_{ \underset{ \text{\refP{Problem:BPImp}} }{ \text{$\zVec$ minimizer of} } }
		\norm{\xVec^{\hrloneParam_n}-\zVec}_2>\epsilon
		\TextForAll n\in\mathbb{N}.
	\end{align*}
	By statement (c) of \thref{Lemma:ConvergenceHRL1Param} the sequence $\xVec^{\hrloneParam_n}$ is bounded.
	Thus, there exists a subsequence $\left(\hrloneParam'_n\right)_{n\in\mathbb{N}}$
	such that $\xVec^{\hrloneParam'_n}$ converges to some $\xVec$ and $\lim_{n\rightarrow\infty} \hrloneParam'_n=\infty$.
	Equation
	\refP{Equation:Convergence:StoppingCriteriaResidual:Theorem:ConvergenceHRL1Param}
	yields that $\xVec$ is feasible for \refP{Problem:BPImp},
	and
	\refP{Equation:Convergence:StoppingCriteria:Theorem:ConvergenceHRL1Param}
	gives that $\xVec$ is a minimizer of \refP{Problem:BPImp}.
	Hence, we have
	\begin{align*}
		\epsilon<
		\lim_{n\rightarrow\infty} 
		\inf_{ \underset{ \text{\refP{Problem:BPImp}} }{ \text{$\zVec$ minimizer of} } }
		\norm{\xVec^{\hrloneParam_n}-\zVec}_2
		\leq 
		\lim_{n\rightarrow\infty}\norm{\xVec^{\hrloneParam'_n}-\xVec}_2
		=0.
	\end{align*}
	This is a contradiction to the assumption and thus proves statement (3).\\
	(4):
	We note that the set of minimizer of \refP{Problem:BPImp} is closed since the set of feasible points is closed
	and the objective function continuous.
	Since the objective function is a norm in a finite-dimensional space,
	the set of minimizers of \refP{Problem:BPImp} is bounded and hence compact.
	By a continuity/compactness argument
	there exists some minimizer $\xVec^{BPImp}_n$ of \refP{Problem:BPImp}
	such that
	$\inf_{ \underset{ \text{\refP{Problem:BPImp}} }{ \text{$\zVec$ minimizer of} } }
				\norm{\xVec^{\hrloneParam_n}-\zVec}_2=\norm{\xVec^{\hrloneParam_n}-\xVec^{BPImp}_n}_2$.
	Statement (3) yields that
	$\lim_{n\rightarrow\infty} \norm{\xVec^{\hrloneParam_n}-\xVec^{BPImp}_n}_2=0$.
	If $\xVec=\lim_{n\rightarrow\infty}\xVec^{\hrloneParam_n}$,
	it follows that
	\begin{align*}
		\lim_{n\rightarrow\infty} \norm{\xVec-\xVec^{BPImp}_n}_2
		\leq\lim_{n\rightarrow\infty} \norm{\xVec-\xVec^{\hrloneParam_n}}_2+\norm{\xVec^{\hrloneParam_n}-\xVec^{BPImp}_n}_2
		=0+0.
	\end{align*}
	By the closedness of the set of minimizers $\xVec$ is also a minimizer of \refP{Problem:BPImp}.
	The existence of such a convergent sequence follows from the boundedness ensured by statement (c) of \thref{Lemma:ConvergenceHRL1Param}.\\
	(5):
	By statement (d) of \thref{Lemma:ConvergenceHRL1Param} we have
	\begin{align*}
		0
		\leq\hrloneParam\normRHS{\AMat\xVec^\hrloneParam-\yVec}
		\leq \hrloneParam\normRHS{\AMat\xVec^{BPImp}-\yVec}+\left(\norm{\xVec^{BPImp}}_1-\norm{\xVec^\hrloneParam}_1\right)
		\leq \norm{\xVec^{BPImp}}_1-\norm{\xVec^\hrloneParam}_1,
	\end{align*}
	where the last inequality holds since $\xVec^{BPImp}$ is feasible and $\yVec\in\Range{\AMat}$.
	Equation
	\refP{Equation:Convergence:StoppingCriteria:Theorem:ConvergenceHRL1Param}
	yields that the right hand side and thus $\hrloneParam\normRHS{\AMat\xVec^\hrloneParam-\yVec}$
	converges to zero. Using this and \refP{Equation:Convergence:StoppingCriteria:Theorem:ConvergenceHRL1Param} again gives
	\begin{align*}
		\lim_{\hrloneParam\rightarrow\infty}
		\norm{\xVec^\hrloneParam}_1+\hrloneParam\normRHS{\AMat\xVec^\hrloneParam-\yVec}
		=\inf_{\zVec\in\argmin{\zVec'\in\mathbb{R}^N}\normRHS{\AMat\zVec'-\yVec}}\norm{\zVec}_1+0.
	\end{align*}
	Since $\yVec\in\Range{\AMat}$, this yields the last statement.
\end{proof}
The convergence to \refP{Problem:BPImp} occurs at a finite value,
but to prove this we need to introduce subdifferentials.
\begin{Definition}
	Let $C\subset\mathbb{R}^M$ be a convex set and
	$\fFunc{}:C\rightarrow\mathbb{R}$ be a convex function and $\wVec\in C$. The set
	\begin{align*}
		\SubDiff{\fFunc{}}{{\wVec}}
		:=\left\{\gVec\in\mathbb{R}^M:\scprod{\gVec}{\wVec'-\wVec}\leq\fFunc{\wVec'}-\fFunc{\wVec}
			\TextForAll \wVec'\in C\right\}
	\end{align*}
	is called the subdifferential of $\fFunc{}$ at $\wVec$. Any vector $\gVec\in\SubDiff{\fFunc{}}{{\wVec}}$
	is called subgradient of $\fFunc{}$ at $\wVec$.
\end{Definition}
We only require simple statements about subdifferentials, namely:
$\xVec^\#$ is a minimizer of $\tempmin{\xVec\in C}\fFunc{}$ if and only if
$0\in\SubDiff{\fFunc{}}{\xVec^\#}$ \cite[Section~5.2]{polyak_book}.
If $\fFunc{}$ and $\gFunc{}$ are convex with common domain $C$, then
$\SubDiff{\left(\fFunc{}+\gFunc{}\right)}{\wVec}=\SubDiff{\fFunc{}}{\wVec}+\SubDiff{\gFunc{}}{\wVec}$ \cite[Section~5.1]{polyak_book}.
Lastly, the concatenation of a convex function with an affine transformation obeys
$\SubDiff{\fFunc{\AMat\cdot-\yVec}}{\xVec}=\AMat^T\SubDiff{\fFunc{}}{\AMat\xVec-\yVec}$ \cite[Section~5.1]{polyak_book}.
For more information about subdifferentials we refer to \cite{polyak_book}.
The subdifferential of an arbitrary norm is not always a unique vector
but it has a nice characterization by its dual norm.
\begin{Lemma}\label{Lemma:SubDiff_of_Norm}
	Let $\norm{\cdot}$ be a norm on $\mathbb{R}^M$ with dual norm
	$\normDual{\cdot}:=\sup_{\normRHS{\wVec}\leq 1}\scprod{\cdot}{\wVec}$.
	Then
	\begin{align*}
		\SubDiff{\normRHS{\cdot}}{\wVec}
		=\left\{\gVec\in\mathbb{R}^M:\scprod{\gVec}{\wVec}=\normRHS{\wVec}\TextAnd \normDual{\gVec}\leq 1\right\}
	\end{align*}
	and in particular, if $\wVec\neq 0$, we have
	\begin{align*}
		\SubDiff{\normRHS{\cdot}}{\wVec}
		=\left\{\gVec\in\mathbb{R}^M:\scprod{\gVec}{\wVec}=\normRHS{\wVec}\TextAnd \normDual{\gVec}= 1\right\}.
	\end{align*}
\end{Lemma}
\begin{proof}
	At first let
	$\gVec\in\left\{\gVec'\in\mathbb{R}^M:\scprod{\gVec'}{\wVec}=\normRHS{\wVec}\TextAnd \normDual{\gVec'}\leq 1\right\}$.
	Then for any $\wVec'\in\mathbb{R}^M$ we have
	\begin{align*}
		\scprod{\gVec}{\wVec'-\wVec}
		=\scprod{\gVec}{\wVec'}-\normRHS{\wVec}
		\leq\normDual{\gVec}\normRHS{\wVec'}-\normRHS{\wVec}
		\leq\normRHS{\wVec'}-\normRHS{\wVec}
	\end{align*}
	and thus $\gVec$ is a subgradient.
	Now let $\gVec\in\SubDiff{\normRHS{\cdot}}{\wVec}$.
	If we apply the definition of a subgradient once for $\wVec':=0\in\mathbb{R}^M$ and once for $\wVec':=2\wVec$, we get
	\begin{align*}
		\scprod{\gVec}{\wVec}=\scprod{\gVec}{2\wVec-\wVec}\leq 2\normRHS{\wVec}-\normRHS{\wVec}=\normRHS{\wVec}
		\\
		\TextAnd\scprod{\gVec}{-\wVec}=\scprod{\gVec}{0-\wVec}\leq\normRHS{0}-\normRHS{\wVec}=-\normRHS{\wVec}.
	\end{align*}
	Hence, $\scprod{\gVec}{\wVec}=\normRHS{\wVec}$.
	By a continuity/compactness argument there exists a $\wVec'\in\mathbb{R}^M$ such that
	$\normRHS{\wVec'}\leq 1$ and $\normDual{\gVec}=\scprod{\gVec}{\wVec'}$.
	If we apply the definition of a subgradient for $\wVec'$, it follows that
	\begin{align*}
		\normDual{\gVec}
		=\scprod{\gVec}{\wVec'-\wVec}+\scprod{\gVec}{\wVec}
		\leq \normRHS{\wVec'}-\normRHS{\wVec}+\scprod{\gVec}{\wVec}
		=\normRHS{\wVec'}
		\leq 1.
	\end{align*}
	Thus, we obtain the second inclusion.
	For the ``in particular part'' note that, if $\wVec\neq 0$, we additionally get
	$
		1=\scprod{\gVec}{\wVec}\normRHS{\wVec}^{-1}\leq \normDual{\gVec}
	$,
	which proves the last statement.
\end{proof}
Now we can prove the result about convergence to BPImp for a finite value $\hrloneParam$.
\begin{proof}[Proof of \thref{Proposition:FiniteConvergence}]
	Since $\AMat$ is surjective, $\AMat^T:\mathbb{R}^M\rightarrow\Range{\AMat^T}$ is bijective and
	there exists an inverse mapping $\BMat:\Range{\AMat^T}\rightarrow\mathbb{R}^M$.
	Since the spaces are finite-dimensional, the operator norm of $\BMat$ is finite and obeys
	\begin{align*}
		\infty
		>\norm{\BMat}_{\infty\rightarrow\DualNormSymbol}
		:=\sup_{0\neq\vVec\in\Range{\AMat^T}}\frac{\normDual{\BMat\vVec}}{\norm{\vVec}_\infty}
		=\sup_{0\neq\wVec\in\mathbb{R}^M}\frac{\normDual{\BMat\AMat^T\wVec}}{\norm{\AMat^T\wVec}_\infty}
		=\sup_{0\neq\wVec\in\mathbb{R}^M}\frac{\normDual{\wVec}}{\norm{\AMat^T\wVec}_\infty}
		=\hrloneParam_\infty.
	\end{align*}
	Towards a contradiction assume that $\AMat\xVec^\hrloneParam\neq \yVec$.
	By the affine transformation concatenation formula and \thref{Lemma:SubDiff_of_Norm} we have
	\begin{align*}
		\SubDiff{\left(\hrloneParam\normRHS{\AMat\cdot-\yVec}\right)}{\xVec^\hrloneParam}
		=\hrloneParam\AMat^T\SubDiff{\normRHS{\cdot}}{\AMat\xVec^\hrloneParam-\yVec}
		=\left\{\hrloneParam\AMat^T\gVec:\gVec\in\mathbb{R}^M\TextAnd
			\scprod{\gVec}{\AMat\xVec^\hrloneParam-\yVec}=\normRHS{\AMat\xVec^\hrloneParam-\yVec}
			\TextAnd \normDual{\gVec}= 1\right\}.
	\end{align*}
	Let $\gVec'=\hrloneParam\AMat^T\gVec\in\SubDiff{\left(\hrloneParam\normRHS{\AMat\cdot-\yVec}\right)}{\xVec^\hrloneParam}$
	be any subgradient.
	By the bound on $\hrloneParam$ it follows that
	\begin{align*}
		\norm{\gVec'}_\infty
		=\hrloneParam\norm{\AMat^T\gVec}_\infty
		\geq\normDual{\gVec}\hrloneParam\norm{\BMat}_{\infty\rightarrow\DualNormSymbol}^{-1}
		>\normDual{\gVec}
		=1,
	\end{align*}
	i.e., no vector in $\SubDiff{\left(\hrloneParam\normRHS{\AMat\cdot-\yVec}\right)}{\xVec^\hrloneParam}$ lies in the
	$\ell_\infty$ unit ball.
	By \thref{Lemma:SubDiff_of_Norm} on the other hand $\SubDiff{\norm{\cdot}_1}{\xVec^\hrloneParam}$
	is a subset of the $\ell_\infty$ unit ball
	and thus
	\begin{align}\label{Equation:EQ1:Proposition:FiniteConvergence}
		\SubDiff{\left(\hrloneParam\normRHS{\AMat\cdot-\yVec}\right)}{\xVec^\hrloneParam}
		\cap \left(-\SubDiff{\norm{\cdot}_1}{\xVec^\hrloneParam}\right)
		=\emptyset.
	\end{align}
	Since $\xVec^\hrloneParam$ is an optimizer,
	the optimality criterion of convex optimization yields that
	$0\in\SubDiff{\left(\norm{\cdot}_1+\hrloneParam\normRHS{\AMat\cdot-\yVec}\right)}{\xVec^\hrloneParam}$ 
	$=\SubDiff{\norm{\cdot}_1}{\xVec^\hrloneParam}+\SubDiff{\left(\hrloneParam\normRHS{\AMat\cdot-\yVec}\right)}{\xVec^\hrloneParam}$
	which is a contradiction to \refP{Equation:EQ1:Proposition:FiniteConvergence}.
	Hence, the original assumption was wrong and we have $\normRHS{\AMat\xVec^\hrloneParam-\yVec}=0$.
	By statement (1) of \thref{Theorem:ConvergenceHRL1Param}
	$\xVec^\hrloneParam$ is a minimizer of \refP{Problem:BPImp}.
\end{proof}
We only give a sketch for a proof of \thref{Proposition:HRL1Recovery:QP},
since we would need to cite too many results for a complete proof.
\begin{proof}[Sketch for a proof of \thref{Proposition:HRL1Recovery:QP}]
	The proof is a consequence of \cite[Lemma~11.15]{IntroductionCS} and \cite[Lemma~11.16]{IntroductionCS}.
	In order to use \cite[Lemma~11.15]{IntroductionCS}, we need to consider
	\cite[Definition~11.2]{IntroductionCS} and \cite[Definition~11.4]{IntroductionCS}.
	We set $c:=S$ and $s_\ast:=1$. If $q\neq 1$, we apply
	\cite[Lemma~11.16]{IntroductionCS} to obtain the simultaneous $\left(\ell_q,\ell_1\right)$-quotient property
	relative to $\normRHS{\cdot}$ with constant $D=\left(1+\rho\right)dS^{\frac{1}{q}-1}+\tau$ and
	$D'=d$.
	If $q=1$, the simultaneous $\left(\ell_q,\ell_1\right)$-quotient property
	relative to $\normRHS{\cdot}$ is directly fulfilled with constant $D=d$ and	$D'=d$.
	In both cases this yields the second requirement of \cite[Lemma~11.15]{IntroductionCS}.
	Let $\Delta\left(\yVec\right)\in \argmin{\zVec\in\mathbb{R}^N}\norm{\zVec}_1+\hrloneParam\normRHS{\AMat\zVec-\yVec}$
	for all $\yVec\in\mathbb{R}^N$.
	If \refP{Equation:EQ1:Subsection:Theory_RecoveryGuarantees} holds true, then $\rho'=\rho$ and
	\thref{Theorem:HRL1Recovery:Lambda>TauSq} yields that
	$\left(A,\Delta\right)$ is mixed $\left(\ell_q,\ell_1\right)$-instance optimal of order $S$ with constant
	$C=\begin{Bmatrix}
		2\frac{\left(1+\rho\right)^2}{1-\rho} & \TextIf & q\in\left(1,\infty\right]\\
		2\frac{1+\rho}{1-\rho} & \TextIf & q=1 \end{Bmatrix}$.
	Thus, the other requirement of \cite[Lemma~11.15]{IntroductionCS} is fulfilled which then yields the claim.
\end{proof}
\subsection{Proofs of Subsection \ref{Subsection:Theory_RecoveryEquivalence} Equivalent Conditions for Successfull Recovery with \hrlone{}}
\label{Subsection:Theory_RecoveryEquivalence:Proofs}
In order to proof \thref{Theorem:HRL1SRdecEquivalence} and \thref{Corollary:ShapeHRL1Param},
we need to prove that certain null space properties are equivalent.
This is in general straightforward but we also need to prove that certain constants
can be preserved and this is highly nontrivial.
The corresponding statements will also give a construction formula for missing constants
which we require for the proofs of \thref{Theorem:HRL1SRdecEquivalence} and \thref{Corollary:ShapeHRL1Param}.
However, these constructions
require calculating $\frac{N!}{S!\left(N-S\right)!}$ values.
Hence, the constants can not be calculated in polynomial time using these results. See also \cite[Section~IV]{SRNSPIsNotCalculatable}.
\begin{Lemma}[NSP and ORNSP]\label{Lemma:NSP<=>RNSP}
	Let $\AMat\in\mathbb{R}^{M\times N}$, $S\in\SetOf{N}$, $q\in\left[1,\infty\right]$ and $\normRHS{\cdot}$ be a norm on $\mathbb{R}^M$.
	Then we have the equivalence:
	\begin{itemize}
		\item[(1)]
			If $\AMat$ has $\ell_q$-ORNSP of order $S$ wrt $\normRHS{\cdot}$
			with constant $\tau$,
			then $\AMat$ has $\ell_q$-NSP of order $S$.
		\item[(2)]
			If $\AMat$ has $\ell_q$-NSP of order $S$,
			then $\tau_q^0\in\left(0,\infty\right)$ and for every $\tau'>\tau_q^0$
			$\AMat$ has $\ell_q$-ORNSP of order $S$ wrt $\normRHS{\cdot}$
			with constant $\tau'$.
	\end{itemize}
\end{Lemma}
\begin{proof}
	Statement (1):
	Let $\vVec\in\Kernel{\AMat}\setminus\ZeroSet$ and $\SetSize{T}\leq S$. By the ORNSP we have
	\begin{align*}
		\norm{\ProjToIndex{T}{\vVec}}_q<S^{\frac{1}{q}-1}\norm{\ProjToIndex{T^c}{\vVec}}_q+\tau\normRHS{\AMat\vVec}
		=S^{\frac{1}{q}-1}\norm{\ProjToIndex{T^c}{\vVec}}_q.
	\end{align*}
	Statement (2):
	Let $T\subset\SetOf{N}$ be an arbitrary set with $\SetSize{T}\leq S$. We set
	\begin{align*}
		\tau_T
		:=\sup_{\vVec\in\mathbb{R}^N\setminus\Kernel{\AMat}}
			\frac{\norm{\ProjToIndex{T}{\vVec}}_q-S^{\frac{1}{q}-1}\norm{\ProjToIndex{T^c}{\vVec}}_1}{\normRHS{\AMat\vVec}}
		=\sup_{\underset{\norm{\vVec}_2=1}{\vVec\in\mathbb{R}^N\setminus\Kernel{\AMat}:}}
			\frac{\norm{\ProjToIndex{T}{\vVec}}_q-S^{\frac{1}{q}-1}\norm{\ProjToIndex{T^c}{\vVec}}_1}{\normRHS{\AMat\vVec}}.
	\end{align*}
	If $\Kernel{\AMat}=\mathbb{R}^N$, then $\AMat$ is the zero matrix
	which does not have $\ell_q$-NSP of order $S\geq 1$.
	Thus, we have $\mathbb{R}^N\setminus\Kernel{\AMat}\neq\ZeroSet$
	and hence $\tau_T>-\infty$.
	For now assume $\tau_T<\infty$ for all $\SetSize{T}\leq S$ and note that
	$
		\tau_q^0=\sup_{\SetSize{T}\leq S}\tau_T
	$.
	Since this supremum is being taken over finitely many elements, we have $\tau_q^0<\infty$.
	Since $\AMat$ has the $\ell_q$-NSP with $S\geq 1$,
	any standard unit vector $\eVec^{n'}$ can not be an element of $\Kernel{\AMat}$.
	For $n'\in T$, it follows that $\tau_T\geq\frac{1}{\normRHS{\AMat\eVec^{n'}}}>0$.
	Thus, we have $\tau_q^0\in\left(0,\infty\right)$.
	Now let $\tau'>\tau_q^0$ be arbitrary. We get for all $\vVec\notin\Kernel{\AMat}$
	\begin{align*}
		\norm{\ProjToIndex{T}{\vVec}}_q
		=&\frac{\norm{\ProjToIndex{T}{\vVec}}_q-S^{\frac{1}{q}-1}\norm{\ProjToIndex{T^c}{\vVec}}_1}{\normRHS{\AMat\vVec}}\normRHS{\AMat\vVec}
			+S^{\frac{1}{q}-1}\norm{\ProjToIndex{T^c}{\vVec}}_1
		\leq\tau_T\normRHS{\AMat\vVec}
			+S^{\frac{1}{q}-1}\norm{\ProjToIndex{T^c}{\vVec}}_1
		\\\leq&\tau_q^0\normRHS{\AMat\vVec}
			+S^{\frac{1}{q}-1}\norm{\ProjToIndex{T^c}{\vVec}}_1
		<\tau'\normRHS{\AMat\vVec}
			+S^{\frac{1}{q}-1}\norm{\ProjToIndex{T^c}{\vVec}}_1.
	\end{align*}
	For all $\vVec\in\Kernel{\AMat}\setminus\ZeroSet$ we get by the NSP anyway 
	\begin{align*}
		\norm{\ProjToIndex{T}{\vVec}}_q
		<S^{\frac{1}{q}-1}\norm{\ProjToIndex{T^c}{\vVec}}_1
		=\tau'\normRHS{\AMat\vVec}
			+S^{\frac{1}{q}-1}\norm{\ProjToIndex{T^c}{\vVec}}_1.
	\end{align*}
	So $\AMat$ has $\ell_q$-ORNSP of order $S$ wrt $\normRHS{\cdot}$
	with the claimed constant.
	It remains to prove that $\tau_T<\infty$.
	Recall that $\mathbb{R}^N\setminus\Kernel{\AMat}\neq\ZeroSet$.
	Suppose there exists a sequence of vectors
	$\left(\vVec_k\right)_{k\in\mathbb{N}}\subset\mathbb{R}^N\setminus\Kernel{\AMat}$ such that
	$\norm{\vVec_k}_2=1$ and
	\begin{align*}
			\frac{\norm{\ProjToIndex{T}{\vVec_k}}_q
			-S^{\frac{1}{q}-1}\norm{\ProjToIndex{T^c}{\vVec_k}}_1}{
			\normRHS{\AMat\vVec_k}}\rightarrow\infty.
	\end{align*}
	Since the sequence $\left(\vVec_k\right)_{k\in\mathbb{N}}$ is bounded, it contains a
	subsequence $\left(\vVec'_k\right)_{k\in\mathbb{N}}$
	that converges to some $\vVec$.
	Since the sequence $\vVec'_k$ is bounded, we have
	$
		\norm{\ProjToIndex{T}{\vVec'_k}}_q
			-S^{\frac{1}{q}-1}\norm{\ProjToIndex{T^c}{\vVec'_k}}_1
			\leq \norm{\vVec'_k}_q\leq R
	$
	for some $R>0$. Thus, we get 
	\begin{align*}
		\frac{\norm{\ProjToIndex{T}{\vVec'_k}}_q
			-S^{\frac{1}{q}-1}\norm{\ProjToIndex{T^c}{\vVec'_k}}_1}{
			\normRHS{\AMat\vVec'_k}}
		\leq\frac{R}{\normRHS{\AMat\vVec'_k}}.
	\end{align*}
	The left hand side goes to infinity for $k\rightarrow\infty$, thus the
	denominator on the right hand side needs to go to zero.
	Hence, we obtain that $\vVec\in\Kernel{\AMat}$. Since also $\norm{\vVec}_2=\lim_{k\rightarrow\infty}\norm{\vVec'_k}_2=1\neq 0$,
	we have $\vVec\in\Kernel{\AMat}\setminus\ZeroSet$.
	We can use the NSP to obtain
	$
		\norm{\ProjToIndex{T}{\vVec}}_q
			-S^{\frac{1}{q}-1}\norm{\ProjToIndex{T^c}{\vVec}}_1< 0
	$.
	By continuity there exists a $k_0$ such that for all $k\geq k_0$
	we also have the strict inequality
	\begin{align*}
		\norm{\ProjToIndex{T}{\vVec'_k}}_q
			-S^{\frac{1}{q}-1}\norm{\ProjToIndex{T^c}{\vVec'_k}}_1
		<0,
	\end{align*}
	but this is a contradiction to
	\begin{align*}
		\lim_{k\rightarrow\infty}\frac{\norm{\ProjToIndex{T}{\vVec'_k}}_q
			- S^{\frac{1}{q}-1}\norm{\ProjToIndex{T^c}{\vVec'_k}}_1}{
			\normRHS{\AMat\vVec'_k}}
		=
		\lim_{k\rightarrow\infty}\frac{\norm{\ProjToIndex{T}{\vVec_k}}_q
			- S^{\frac{1}{q}-1}\norm{\ProjToIndex{T^c}{\vVec_k}}_1}{
			\normRHS{\AMat\vVec_k}}=\infty.
	\end{align*}
	So it follows that $\tau_T<\infty$.
\end{proof}
We will see in \thref{Corollary:ShapeSNSPAndRNSP} that
we can not improve this result to any $\tau'\leq\tau_q^0$.
\begin{Lemma}[Equivalence of ORNSP and RNSP]\label{Lemma:RNSP<=>SRNSP}
	Let $\AMat\in\mathbb{R}^{M\times N}$, $S\in\SetOf{N}$, $q\in\left[1,\infty\right]$ and $\normRHS{\cdot}$ be a norm on $\mathbb{R}^M$.
	Then we have the equivalence:
	\begin{itemize}
		\item[(1)]
			If $\AMat$ has 
			$\ell_q$-RNSP of order $S$ wrt $\normRHS{\cdot}$
			with constants $\rho$ and $\tau$,
			then for every $\tau'>\tau$ $\AMat$ has $\ell_q$-ORNSP of order $S$ wrt $\normRHS{\cdot}$
			with constant $\tau'$.
		\item[(2)]
			If $\AMat$ has $\ell_q$-ORNSP of order $S$ wrt $\normRHS{\cdot}$
			with constant $\tau$,
			then $\rho_q\left(\tau\right)\in\left[0,1\right)$ and
			$\AMat$ has $\ell_q$-RNSP of order $S$ wrt $\normRHS{\cdot}$
			with constants $\rho_q\left(\tau\right)$ and $\tau$.
	\end{itemize}
\end{Lemma}
\begin{proof}
	Statement (1):
	Let $\tau'>\tau$.
	Let $\vVec\notin\Kernel{\AMat}$ and $\SetSize{T}\leq S$. By the RNSP we have
	\begin{align*}
		\norm{\ProjToIndex{T}{\vVec}}_q
		\leq& \rho S^{\frac{1}{q}-1}\norm{\ProjToIndex{T^c}{\vVec}}_1+\tau\normRHS{\AMat\vVec}
		\leq S^{\frac{1}{q}-1}\norm{\ProjToIndex{T^c}{\vVec}}_1+\tau\normRHS{\AMat\vVec}
		<S^{\frac{1}{q}-1}\norm{\ProjToIndex{T^c}{\vVec}}_1
			+\tau'\normRHS{\AMat\vVec}.
	\end{align*}
	Now let $\vVec\in\Kernel{\AMat}\setminus\ZeroSet$. Then either $\ProjToIndex{T}{\vVec}\neq 0$
	or $\ProjToIndex{T^c}{\vVec}\neq 0$. Suppose that $\ProjToIndex{T}{\vVec}\neq 0$.
	Then the RNSP yields that
	$
		0<\norm{\ProjToIndex{T}{\vVec}}_q
		\leq \rho S^{\frac{1}{q}-1}\norm{\ProjToIndex{T^c}{\vVec}}_1+\tau\normRHS{\AMat\vVec}
		=\rho S^{\frac{1}{q}-1}\norm{\ProjToIndex{T^C}{\vVec}}_1
	$.
	Thus, in both cases $\ProjToIndex{T^c}{\vVec}\neq 0$. Using this, $\rho<1$ and the RNSP once more yields
	\begin{align*}
		\norm{\ProjToIndex{T}{\vVec}}_q
		\leq& \rho S^{\frac{1}{q}-1}\norm{\ProjToIndex{T^c}{\vVec}}_1 +\tau\normRHS{\AMat\vVec}
		< S^{\frac{1}{q}-1}\norm{\ProjToIndex{T^c}{\vVec}}_1+\tau\normRHS{\AMat\vVec}
		= S^{\frac{1}{q}-1}\norm{\ProjToIndex{T^c}{\vVec}}_1+\tau'\normRHS{\AMat\vVec}.
	\end{align*}
	It follows that $\AMat$ has $\ell_q$-ORNSP of order $S$ wrt $\normRHS{\cdot}$
	with constant $\tau'$.\\
	Statement (2):
	Let $T\subset \SetOf{N}$ be an arbitray set with $\SetSize{T}\leq S$.
	We set
	\begin{align*}
		\rho_T
		:=\sup_{\underset{\ProjToIndex{T^c}{\vVec}\neq 0}{\vVec\in\mathbb{R}^N:}}
			\frac{\norm{\ProjToIndex{T}{\vVec}}_q-\tau\normRHS{\AMat\vVec}}{
			S^{\frac{1}{q}-1}\norm{\ProjToIndex{T^c}{\vVec}}_1}
		=\sup_{\underset{\ProjToIndex{T^c}{\vVec}\neq 0,\norm{\vVec}_2=1}{\vVec\in\mathbb{R}^N:}}
			\frac{\norm{\ProjToIndex{T}{\vVec}}_q-\tau\normRHS{\AMat\vVec}}{
			S^{\frac{1}{q}-1}\norm{\ProjToIndex{T^c}{\vVec}}_1}.
	\end{align*}
	By the ORNSP we have $\rho_T\leq 1$.
	For now assume that	$\rho_T<1$ for all $\SetSize{T}\leq S$ and note that
	$
		\rho_q\left(\tau\right)=\max\left\{0,\sup_{\SetSize{T}\leq S}\rho_T\right\}
	$.
	Since this supremum is being taken over finitely many elements, we have $\rho_q\left(\tau\right)\in\left[0,1\right)$.
	For all $\vVec$ such that $\ProjToIndex{T^c}{\vVec}\neq 0$ we get
	\begin{align*}
		\norm{\ProjToIndex{T}{\vVec}}_q
		=&\frac{\norm{\ProjToIndex{T}{\vVec}}_q-\tau\normRHS{\AMat\vVec}}{
				S^{\frac{1}{q}-1}\norm{\ProjToIndex{T^c}{\vVec}}_1
			}
			S^{\frac{1}{q}-1}\norm{\ProjToIndex{T^c}{\vVec}}_1 +\tau\normRHS{\AMat\vVec}
		\leq\rho_T S^{\frac{1}{q}-1}\norm{\ProjToIndex{T^c}{\vVec}}_1+\tau\normRHS{\AMat\vVec}
		\\\leq&\rho_q\left(\tau\right) S^{\frac{1}{q}-1}\norm{\ProjToIndex{T^c}{\vVec}}_1+\tau\normRHS{\AMat\vVec}.
	\end{align*}
	For all $\vVec\neq 0$ such that $\ProjToIndex{T^c}{\vVec}= 0$ we get by using the ORNSP
	\begin{align*}
		\norm{\ProjToIndex{T}{\vVec}}_q
		<S^{\frac{1}{q}-1}\norm{\ProjToIndex{T^c}{\vVec}}_1+\tau\normRHS{\AMat\vVec}
		=\rho_q\left(\tau\right) S^{\frac{1}{q}-1}\norm{\ProjToIndex{T^c}{\vVec}}_1+\tau\normRHS{\AMat\vVec}.
	\end{align*}
	So $\AMat$ has $\ell_q$-RNSP of order $S$ wrt $\normRHS{\cdot}$ with the claimed stableness constant.\\
	It remains to prove $\rho_T<1$.
	If $\left\{\vVec\TextSuchThat \ProjToIndex{T^c}{\vVec}\neq 0\right\}=\emptyset$, then $\rho_T=-\infty<1$.
	On the other hand assume $\left\{\vVec\TextSuchThat \ProjToIndex{T^c}{\vVec}\neq 0\right\}\neq \emptyset$.
	Suppose there exists a sequence of vectors
	$\left(\vVec_k\right)_{k\in\mathbb{N}}$ such that
	$\norm{\vVec_k}_2=1$, $\ProjToIndex{T^c}{\vVec_k}\neq 0$ and
	\begin{align*}
		\frac{\norm{\ProjToIndex{T}{\vVec_k}}_q-\tau\norm{\AMat\vVec_k}}{
				S^{\frac{1}{q}-1}\norm{\ProjToIndex{T^c}{\vVec_k}}_1
			}
			\rightarrow 1.
	\end{align*}
	Since $\left(\vVec_k\right)_{k\in\mathbb{N}}$ is bounded, it contains a convergent subsequence
	$\left(\vVec'_k\right)_{k\in\mathbb{N}}$. Let $\vVec:=\lim_{k\rightarrow\infty}\vVec'_k$.
	There are now two cases that both result in a contradiction. The first one is
	$\ProjToIndex{T^c}{\vVec}\neq 0$. Then we have by the ORNSP
	\begin{align*}
		1
		=\lim_{k\rightarrow\infty}\frac{\norm{\ProjToIndex{T}{\vVec'_k}}_q-\tau\normRHS{\AMat\vVec'_k}}{
				S^{\frac{1}{q}-1}\norm{\ProjToIndex{T^c}{\vVec'_k}}_1
			}
		=\frac{\norm{\ProjToIndex{T^c}{\vVec}}_q-\tau\normRHS{\AMat\vVec}}{
				S^{\frac{1}{q}-1}\norm{\ProjToIndex{T^c}{\vVec}}_1
			}
		<1
	\end{align*}
	which is a contradiction. The second case is $\ProjToIndex{T^c}{\vVec}=0$.
	Since $\norm{\vVec}_2=1\neq 0$, the ORNSP yields that
	$
		\norm{\ProjToIndex{T}{\vVec}}_q
		<S^{\frac{1}{q}-1}\norm{\ProjToIndex{T^c}{\vVec}}_1+\tau\normRHS{\AMat\vVec}
		=\tau\normRHS{\AMat\vVec}
	$.
	By continuity there exists a $k_0\in\mathbb{N}$ such that for all $k\geq k_0$ we also
	have the strict inequality
	\begin{align*}
		\norm{\ProjToIndex{T^c}{\vVec'_k}}_q-\tau\normRHS{\AMat\vVec'_k}<0,
	\end{align*}
	but this is a contradiction to 
	\begin{align*}
		\lim_{k\rightarrow\infty}
			\frac{\norm{\ProjToIndex{T}{\vVec'_k}}_q-\tau\norm{\AMat\vVec'_k}}{
				S^{\frac{1}{q}-1}\norm{\ProjToIndex{T^c}{\vVec'_k}}_1
			}
		\lim_{k\rightarrow\infty}
			\frac{\norm{\ProjToIndex{T}{\vVec_k}}_q-\tau\norm{\AMat\vVec_k}}{
				S^{\frac{1}{q}-1}\norm{\ProjToIndex{T^c}{\vVec_k}}_1
			}
			=1.
	\end{align*}
	It follows that $\rho_T<1$.
\end{proof}
Recall that we want to prove \thref{Theorem:HRL1SRdecEquivalence}.
In order to do that, we need to prove that
$\AMat$ has $\ell_q$-ORNSP of order $S$ with some constant $\tau<\hrloneParam$.
This topological property is also a consequence of \thref{Lemma:NSP<=>RNSP} as
we will prove next.
\begin{Corollary}\label{Corollary:ShapeSNSPAndRNSP}
	Let $S\in\SetOf{N}$, $q\in\left[1,\infty\right]$ and $\normRHS{\cdot}$ be a norm on $\mathbb{R}^M$.
	Let $\AMat\in\mathbb{R}^{M\times N}$ have $\ell_q$-NSP of order $S$.
	Then we have
	\begin{align*}
		\left\{\text{$\tau\in\left[0,\infty\right):$
			$\AMat$ has $\ell_q$-ORNSP of order $S$ wrt $\normRHS{\cdot}$ with constant $\tau$}\right\}
		=&\left(\tau_q^0,\infty\right),
	\end{align*}
	which is an open set.
\end{Corollary}
\begin{proof}
	Note that by \thref{Lemma:NSP<=>RNSP} we obtain the inclusion $\supset$.
	By the definiton of $\tau_q^0$ there exist $T$, $\left(\vVec_k\right)_{k\in\mathbb{N}}$ such that $\norm{\vVec_k}_2=1$ and
	\begin{align*}
		\tau_q^0=\lim_{k\rightarrow\infty}
		\frac{\norm{\ProjToIndex{T}{\vVec_k}}_q-S^{\frac{1}{q}-1}\norm{\ProjToIndex{T^c}{\vVec_k}}_1}{\normRHS{\AMat\vVec_k}}.
	\end{align*}
	Since $\norm{\vVec_k}_2=1$, there exists a subsequence $\left(\vVec'_k\right)_{k\in\mathbb{N}}$ that converges to some
	$\vVec$ with $\norm{\vVec}_2=1$.
	If $\vVec\in\Kernel{\AMat}$, we set $\epsilon:=S^{\frac{1}{q}-1}\norm{\ProjToIndex{T^c}{\vVec}}_1-\norm{\ProjToIndex{T}{\vVec}}_q$
	which is strictly positive by the $\ell_q$-NSP.
	Since $\vVec\in\Kernel{\AMat}$, $\normRHS{\AMat\vVec'_k}$ converges to zero. Hence,
	$\norm{\ProjToIndex{T}{\vVec'_k}}_q-S^{\frac{1}{q}-1}\norm{\ProjToIndex{T^c}{\vVec'_k}}_1$ needs to converge to zero too.
	By continuity there exists a $k_0$ such that for all $k\geq k_0$ we have
	\begin{align*}
		\norm{\ProjToIndex{T}{\vVec'_k}}_q-S^{\frac{1}{q}-1}\norm{\ProjToIndex{T^c}{\vVec'_k}}_1\leq \frac{\epsilon}{2}.
	\end{align*}
	Hence,
	\begin{align*}
		\epsilon
		=S^{\frac{1}{q}-1}\norm{\ProjToIndex{T^c}{\vVec}}_1-\norm{\ProjToIndex{T}{\vVec}}_q
		=\lim_{k\rightarrow\infty}\norm{\ProjToIndex{T}{\vVec'_k}}_q-S^{\frac{1}{q}-1}\norm{\ProjToIndex{T^c}{\vVec'_k}}_1\leq \frac{\epsilon}{2}
	\end{align*}
	which is a contradiction. Thus, we assume $\vVec\notin\Kernel{\AMat}$. Then we have
	\begin{align*}
		\norm{\ProjToIndex{T}{\vVec}}_q
		=\frac{\norm{\ProjToIndex{T}{\vVec}}_q-S^{\frac{1}{q}-1}\norm{\ProjToIndex{T^c}{\vVec}}_1}{\normRHS{\AMat\vVec}}\normRHS{\AMat\vVec}
			+S^{\frac{1}{q}-1}\norm{\ProjToIndex{T^c}{\vVec}}_1
		=\tau_q^0\normRHS{\AMat\vVec}+S^{\frac{1}{q}-1}\norm{\ProjToIndex{T^c}{\vVec}}_1.
	\end{align*}
	In this case, $\tau_q^0$ is not an ORNSP constant since we lack the strict inequality.
	Since we can increase ORNSP constants arbitrarily, it follows that no element from $\left[0,\tau_q^0\right]$ is an ORNSP constant.
	This is the inclusion $\subset$
	and finishes the proof.
\end{proof}
We can finally proof \thref{Theorem:HRL1SRdecEquivalence} and \thref{Corollary:ShapeHRL1Param}.
\begin{proof}[Proof of \thref{Theorem:HRL1SRdecEquivalence}]
	(1)$\Rightarrow$(2): By \thref{Corollary:ShapeSNSPAndRNSP} we have
	\begin{align*}
		\hrloneParam\in\left\{\text{$\tau\in\left[0,\infty\right):$
			$\AMat$ has $\ell_1$-ORNSP of order $S$ wrt $\normRHS{\cdot}$ with constant $\tau$}\right\}
		=\left(\tau_1^0,\infty\right)
	\end{align*}
	and thus
	$\AMat$ has $\ell_1$-ORNSP of order $S$ wrt $\normRHS{\cdot}$ with constant $\tau':=\frac{\hrloneParam+\tau_1^0}{2}$.
	By \thref{Lemma:RNSP<=>SRNSP} $\AMat$ has 
	$\ell_1$-RNSP of order $S$ wrt $\normRHS{\cdot}$
	with some stableness constant $\rho$ and robustness constant $\tau'$.
	By \thref{Theorem:HRL1Recovery:Lambda>TauSq} the SRD property follows since $\hrloneParam>\tau'$
	and $q=1$.\\
	(2)$\Rightarrow$(3):
	It is helpful to consider the set
	\begin{align*}
		\textnormal{Dec}\left(\AMat\right)
		:=\left\{\Decoder{}:\mathbb{R}^M\rightarrow\mathbb{R}^N\TextSuchThat
			\Decoder{\yVec}\in\argmin{\zVec\in\mathbb{R}^N}\norm{\zVec}_1
					+\hrloneParam\normRHS{\AMat\zVec-\yVec} \TextForAll \yVec\in\mathbb{R}^M\right\},
	\end{align*}
	which by assumption only constains $\ell_1$-SRD of order $S$ wrt $\normRHS{\cdot}$ for $\AMat$.
	Now let $\xVec$ be $S$-sparse and set $\yVec:=\AMat\xVec$. For any minimizer $\xVec^\#$ of
	\refP{Problem:HRL1} with input $\yVec$ choose one decoder $\Decoder{}_{\xVec^\#}\in\textnormal{Dec}\left(\AMat\right)$
	that maps $\yVec$ to $\xVec^\#$.
	Since it is an $\ell_1$-SRD of order $S$ wrt $\normRHS{\cdot}$ for $\AMat$,
	there exists $C_{\xVec^\#}$, $D_{\xVec^\#}$ such that
	\begin{align*}
		\norm{\Decoder{}_{\xVec^\#}\left(\yVec\right)-\xVec}_1
		\leq C_{\xVec^\#}\compr{\xVec}{S}+D_{\xVec^\#}\normRHS{\yVec-\AMat\xVec}
		=0+D_{\xVec^\#}\normRHS{\AMat\xVec-\AMat\xVec}=0
	\end{align*}
	holds true.
	It follows that $\xVec^\#=\xVec$ and the minimizer of \refP{Problem:HRL1} with input $\yVec=\AMat\xVec$ is unique and $\xVec$.\\
	(3)$\Rightarrow$(1):
	Let $\vVec\in\mathbb{R}^M\setminus\ZeroSet$ and $\SetSize{T}\leq S$.
	Set $\yVec:=\AMat\ProjToIndex{T}{\vVec}$.
	By assumption we obtain that $\ProjToIndex{T}{\vVec}$ is the minimizer of
	\refP{Problem:HRL1} with input $\yVec$.
	Since $-\ProjToIndex{T^c}{\vVec}$ is feasible, we have
	\begin{align}\label{Equation:EQ1:Theorem:Wurschli<=>SRNSPq=1}
		\norm{\ProjToIndex{T}{\vVec}}_1
		+\hrloneParam\normRHS{\yVec-\AMat\ProjToIndex{T}{\vVec}}
		\leq\norm{-\ProjToIndex{T^c}{\vVec}}_1
		+\hrloneParam\normRHS{\yVec-\AMat\left(-\ProjToIndex{T^c}{\vVec}\right)}.
	\end{align}
	Since $\vVec\neq 0$, we have
	\begin{align*}
		\ProjToIndex{T}{\vVec}
		=\vVec+\ProjToIndex{T}{\vVec}-\vVec
		=\vVec-\ProjToIndex{T^c}{\vVec}
		\neq -\ProjToIndex{T^c}{\vVec}.
	\end{align*}
	By the assumption we also obtain that $\ProjToIndex{T}{\vVec}$ is the unique minimizer of
	\refP{Problem:HRL1} with input $\yVec$. Thus, the inequality in
	\refP{Equation:EQ1:Theorem:Wurschli<=>SRNSPq=1} is a strict inequality and we get
	\begin{align*}
		\norm{\ProjToIndex{T}{\vVec}}_1
		+\hrloneParam\normRHS{\yVec-\AMat\ProjToIndex{T}{\vVec}}
		<\norm{-\ProjToIndex{T^c}{\vVec}}_1
		+\hrloneParam\normRHS{\yVec-\AMat\left(-\ProjToIndex{T^c}{\vVec}\right)}.
	\end{align*}
	Since we have set $\yVec=\AMat\ProjToIndex{T}{\vVec}$, it follows that
	\begin{align*}
		\norm{\ProjToIndex{T}{\vVec}}_1
		<&\norm{-\ProjToIndex{T^c}{\vVec}}_1
		+\hrloneParam
			\normRHS{\AMat\left(\ProjToIndex{T}{\vVec}+\ProjToIndex{T^c}{\vVec}\right)}
		=\norm{\ProjToIndex{T^c}{\vVec}}_1
		+\hrloneParam\normRHS{\AMat\vVec}
	\end{align*}
	holds true.	Doing this for all $T$ with $\SetSize{T}\leq S$ and all
	$\vVec\in\mathbb{R}^N\setminus\ZeroSet$ yields that
	$\AMat$ has the $\ell_1$-ORNSP of order $S$ wrt $\normRHS{\cdot}$
	with constant $\tau=\hrloneParam$.
\end{proof}
\begin{proof}[Proof of \thref{Corollary:ShapeHRL1Param}]
	By basic norm inequalities $\Decoder{}:\mathbb{R}^M\rightarrow\mathbb{R}^N$
	is an $\ell_q$-SRD of order $S$ wrt $\normRHS{\cdot}$ for $\AMat$ if and only if
	it is an $\ell_1$-SRD of order $S$ wrt $\normRHS{\cdot}$ for $\AMat$.
	By this and by \thref{Theorem:HRL1SRdecEquivalence} we obtain
	the equality in \refP{Equation:Sets:Corollary:ShapeHRL1Param}.
	Now suppose that $\AMat$ has $\ell_q$-NSP of order $S$.
	Note that by a general norm inequality we have
	\begin{align}\label{Equation:EQ1:Corollary:ShapeHRL1Paramq}
		\norm{\ProjToIndex{T}{\vVec}}_1\leq S^{1-\frac{1}{q}}\norm{\ProjToIndex{T}{\vVec}}_q
		\TextForAll \SetSize{T}\leq S \TextAnd \vVec\in\mathbb{R}^N.
	\end{align}
	It immediately follows that $\AMat$ has $\ell_1$-NSP of order $S$.
	In particular, using \refP{Equation:EQ1:Corollary:ShapeHRL1Paramq} on
	the definition of $\tau_1^0$ yields that $\tau_1^0\leq S^{1-\frac{1}{q}}\tau_q^0$.
	By \thref{Corollary:ShapeSNSPAndRNSP} and by \thref{Theorem:HRL1SRdecEquivalence}
	the interval $\left(\tau_1^0,\infty\right)$ equals one and thus both sets of
	\refP{Equation:Sets:Corollary:ShapeHRL1Param}.
\end{proof}
\section{Proofs of Section \ref{Section:NotNPHard} \hrlone{} is a Practical Usable Recovery Algorithm}\label{Section:NotNPHard:Proofs}
\subsection{Proofs of Subsection \ref{Subsection:NotNPHard_GaussianMatrices} Gaussian Measurements}\label{Subsection:NotNPHard_GaussianMatrices:Proofs}
In order to prove \thref{Theorem:Gaussian=>NSP}, we follow the proof of \cite[Theorem~11]{NSPCone}
and adapt to account for a better robustness constant $\tau$.
\begin{Definition}\label{Definition:RobustnessCone}
	For $S\in\SetOf{N}$ and $\rho\in\left[0,1\right)$ the set
	\begin{align*}
		\NSPCone_{\rho,S}^q
		:=\left\{\vVec\in\mathbb{R}^N\TextSuchThat \exists \SetSize{T}\leq S\TextWith
			\norm{\ProjToIndex{T}{\vVec}}_q
			\geq \rho S^{\frac{1}{q}-1}\norm{\ProjToIndex{T^c}{\vVec}}_1
		\right\}
	\end{align*}
	is called robustness cone for $S$ and $\rho$.
\end{Definition}
The robustness cone can generally be interpreted as the set of vectors where the robustness
summand of the RNSP inequality is required.
We can use the robustness cone to get an estimate for the robustness constant.
\begin{Lemma}\label{Lemma:HasNSP}
	Let $q\in\left[1,\infty\right]$ and $\normRHS{\cdot}$ be a norm on $\mathbb{R}^M$.
	Let $\AMat\in\mathbb{R}^{M\times N}$, $S\in\SetOf{N}$, $\rho\in\left[0,1\right)$
	and $\tau\in\left(0,\infty\right)$.
	If $\inf_{\vVec\in\NSPCone_{\rho,S}^q\cap \UnitSphere{N}{\ell_q}}
			\normRHS{\AMat\vVec}\geq\tau^{-1}>0$,
	then $\AMat$ has $\ell_q$-RNSP of order $S$ wrt $\normRHS{\cdot}$
	with constants $\rho$ and $\tau$.
\end{Lemma}
\begin{proof}
	If on the one hand $\vVec\in\mathbb{R}^N\setminus\NSPCone_{\rho,S}^q$, then
	we have for all $\SetSize{T}\leq S$
	\begin{align*}
		\norm{\ProjToIndex{T}{\vVec}}_q
		<\rho S^{\frac{1}{q}-1}\norm{\ProjToIndex{T^c}{\vVec}}_1
		\leq \rho S^{\frac{1}{q}-1}\norm{\ProjToIndex{T^c}{\vVec}}_1 + \tau\normRHS{\AMat\vVec}.
	\end{align*}
	If on the other hand $\vVec\in\NSPCone_{\rho,S}^q\setminus\ZeroSet$, then
	$\frac{\vVec}{\norm{\vVec}_q}\in\NSPCone_{\rho,S}^q\cap \UnitSphere{N}{\ell_q}$ and we have
	\begin{align*}
			\tau\normRHS{\AMat\vVec}
			=\norm{\vVec}_q\tau\normRHS{\AMat\frac{\vVec}{\norm{\vVec}_q}}
			\geq \norm{\vVec}_q\tau\inf_{\vVec'\in\NSPCone_{\rho,S}^q\cap B_{\ell_q}}
				\normRHS{\AMat\vVec'}
			\geq \norm{\vVec}_q.
	\end{align*}
	For any $\SetSize{T}\leq S$ it follows that
	\begin{align*}
		\norm{\ProjToIndex{T}{\vVec}}_q
		\leq \norm{\vVec}_q
		\leq\rho S^{\frac{1}{q}-1}\norm{\ProjToIndex{T^c}{\vVec}}_1+\tau\normRHS{\AMat\vVec}.
	\end{align*}
	Consequently $\AMat$ has $\ell_q$-RNSP of order $S$ wrt $\normRHS{\cdot}$
	with constants $\rho$ and $\tau$.
\end{proof}
Interestingly, all normalized (and rescaled) vectors of the robustness cone are a convex combination of sparse normalized vectors.
The proof is given in \cite[Lemma~3(b)]{NSPCone}.
\begin{Lemma}[ {\cite[Lemma~3(b)]{NSPCone}} ]\label{Lemma:NSPCone}
	Let $S\in\SetOf{N}$ and $\rho\in\left(0,1\right)$. Then
	\begin{align*}
		\NSPCone_{\rho,S}^2\cap \UnitSphere{N}{\ell_2}
		\subset \sqrt{1+\left(1+\rho^{-1}\right)^2}\conv{\Sigma_S\cap\UnitSphere{N}{\ell_2}}.
	\end{align*}
\end{Lemma}
We introduce the Gaussian width.
\begin{Definition}\label{Definition:GaussianWidth}
	Let $T\subset\mathbb{R}^N$ and let the entries of $\gVec\in\mathbb{R}^N$ be independent
	$\GaussianRV{0}{1}$ random variables.
	Then
	\begin{align*}
		\GW{T}:=\Expect{\sup_{\vVec\in T}\scprod{\gVec}{\vVec}}
	\end{align*}
	is called Gaussian width of $T$.
\end{Definition}
Further, we need an estimate for the Gaussian width of $\conv{\Sigma_S\cap\UnitSphere{N}{\ell_2}}$. A proof can be found in \cite[Lemma~4]{NSPCone}.
\begin{Lemma}[ {\cite[Lemma~4]{NSPCone}} ]\label{Lemma:GaussianWidth}
	Let $S\in\SetOf{N}$. Then
	\begin{align*}
		\GW{\conv{\Sigma_S\cap\UnitSphere{N}{\ell_q}}}
		\leq \sqrt{2S\Ln{\ExpE\frac{N}{S}}}+\sqrt{S}.
	\end{align*}
\end{Lemma}
Lastly we introduce Gordon's escape through the mesh theorem. It was originally proven in
\cite{gordon}. A different proof can be found in \cite[Theorem~9.21]{IntroductionCS}.
\begin{Theorem}[ {\cite[Theorem~9.21]{IntroductionCS}} ]\label{Theorem:Gordons_Mesh}
	Let the entries of $\AMat\in\mathbb{R}^{M\times N}$
	be independent $\GaussianRV{0}{1}$ random variables
	and $T\subset\UnitSphere{N}{\ell_2}$.
	Then for any $t\in\left(0,\infty\right)$ we have
	\begin{align*}
		\Prob{\inf_{\vVec\in T}\norm{\AMat\vVec}_2\leq \GMC{M}\sqrt{M}-\GW{T}-t}\leq\Exp{-\frac{t^2}{2}}.
	\end{align*}
\end{Theorem}
With all these statements we can prove \thref{Theorem:Gaussian=>NSP}.
\begin{proof}[Proof of \thref{Theorem:Gaussian=>NSP}]
	Note that the phase transition inequality
	\refP{Equation:PhaseTransition:Theorem:Gaussian=>NSP} is equivalent to
	\begin{align}\label{Equation:PhaseTransition_1:Theorem:Gaussian=>NSP}
		\GMC{M}\sqrt{M}-\sqrt{1+\left(1+\rho^{-1}\right)^2}
			\left(\sqrt{2S\Ln{\ExpE\frac{N}{S}}}+\sqrt{S}\right)
		-\sqrt{2\Ln{\eta^{-1}}}
		\geq\tau^{-1}\sqrt{M}.
	\end{align}	
	We set $T:=\NSPCone_{\rho,S}^2\cap\UnitSphere{N}{\ell_2}$. To estimate the Gaussian width of $T$,
	let the entries of $\gVec\in\mathbb{R}^N$ be independent $\GaussianRV{0}{1}$ random variables.
	By \thref{Lemma:NSPCone} and \thref{Lemma:GaussianWidth} we can estimate
	\begin{align}
		\nonumber
			\GW{T}
		=&
			\GW{\NSPCone_{\rho,S}^2\cap\UnitSphere{N}{\ell_2}}
		=
			\Expect{\sup_{\vVec\in \NSPCone_{\rho,S}^2\cap\UnitSphere{N}{\ell_2}}\scprod{\gVec}{\vVec}}
		\leq
			\Expect{\sup_{\vVec\in \sqrt{1+\left(1+\rho^{-1}\right)^2}\conv{\Sigma_S\cap\UnitSphere{N}{\ell_2}}}\scprod{\gVec}{\vVec}}
		\\=&
			\sqrt{1+\left(1+\rho^{-1}\right)^2}\Expect{\sup_{\vVec\in\conv{\Sigma_S\cap\UnitSphere{N}{\ell_2}}}\scprod{\gVec}{\vVec}}
		=\label{Equation:GW:Theorem:Gaussian=>NSP}
			\sqrt{1+\left(1+\rho^{-1}\right)^2}
			\left(\sqrt{2S\Ln{\ExpE\frac{N}{S}}}+\sqrt{S}\right).
	\end{align}
	Setting $t:=\sqrt{2\Ln{\eta^{-1}}}\in\left(0,\infty\right)$ and using
	\refP{Equation:GW:Theorem:Gaussian=>NSP} and \refP{Equation:PhaseTransition_1:Theorem:Gaussian=>NSP}
	yields that
	\begin{align*}
			\GMC{M}\sqrt{M}-\GW{T}-t=
		&
			\GMC{M}\sqrt{M}-\GW{T}-\sqrt{2\Ln{\eta^{-1}}}
		\\\geq&
			\GMC{M}\sqrt{M}-\sqrt{1+\left(1+\rho^{-1}\right)^2}
			\left(\sqrt{2S\Ln{\ExpE\frac{N}{S}}}+\sqrt{S}\right)-\sqrt{2\Ln{\eta^{-1}}}
		\geq
			\tau^{-1}\sqrt{M}.
	\end{align*}
	Hence, we have the logical statement
	\begin{align}\label{Equation:logical:Theorem:Gaussian=>NSP}
		\inf_{\vVec\in T}\norm{M^\frac{1}{2}\AMat\vVec}_2>\GMC{M}\sqrt{M}-\GW{T}-t
		\Rightarrow
		\inf_{\vVec\in T}\norm{\AMat\vVec}_2>\tau^{-1}.
	\end{align}
	Since the entries of $M^\frac{1}{2}\AMat$ are independent $\GaussianRV{0}{1}$ random variables,
	\thref{Theorem:Gordons_Mesh} together with \refP{Equation:logical:Theorem:Gaussian=>NSP}
	yields that
	\begin{align*}
			\Prob{\inf_{\vVec\in T}\norm{\AMat\vVec}_2>\tau^{-1}}
		\geq
			\Prob{\inf_{\vVec\in T}\norm{M^\frac{1}{2}\AMat\vVec}_2>\sqrt{M}\GMC{M}-\GW{T}-t}
		\geq
			1-\Exp{-\frac{t^2}{2}}
		=
			1-\eta.
	\end{align*}
	Hence, by \thref{Lemma:HasNSP} $\AMat$ has $\ell_2$-RNSP of order $S$ wrt $\norm{\cdot}_2$
	with constants $\rho$ and $\tau$ with probability of at least $1-\eta$.
\end{proof}
In order to estimate the NSP shape constant $\tau_2^0$, we want to optimize this to account
for the smallest possible $\tau$. At first we choose a particular $\eta$ and draw a temporary result
to remove $\eta$.
\begin{Corollary}\label{Corollary:Gaussian=>NSP:no_eta}
	Let the entries of $\AMat\in\mathbb{R}^{M\times N}$
	be independent $\GaussianRV{0}{M^{-1}}$ random variables.
	If
	\begin{align}\label{Equation:PhaseTransition:Corollary:Gaussian=>NSP:no_eta}
		\tau>\left(\GMC{M}-\sqrt{5}\left(\sqrt{2\frac{S}{M}\Ln{\ExpE\frac{N}{S}}}+\sqrt{\frac{S}{M}}\right)\right)^{-1}>0,
	\end{align}
	then
	$\rho\left(\tau\right):=\left(\sqrt{\left(\frac{\GMC{M}-\tau^{-1}}{\sqrt{2\frac{S}{M}\Ln{\ExpE\frac{N}{S}}}+\sqrt{\frac{S}{M}}}\right)^2-1}-1\right)^{-1}\in\left(0,1\right)$, and for any
	$\rho\in\left(\rho\left(\tau\right),1\right)$
	with probability of at least
	\begin{align*}
		1-\Exp{-\frac{1}{2}\left(\sqrt{1+\left(1+\rho\left(\tau\right)^{-1}\right)^2}-\sqrt{1+\left(1+\rho^{-1}\right)^2}\right)^2\left(\sqrt{2\frac{S}{M}\Ln{\ExpE\frac{N}{S}}}+\sqrt{\frac{S}{M}}\right)^2M}\in\left(0,1\right)
	\end{align*}
	$\AMat$ has $\ell_2$-RNSP of order $S$ wrt $\norm{\cdot}_2$
	with constants $\rho$ and $\tau$.
\end{Corollary}
\begin{proof}
	To prove that $\rho\left(\tau\right)$ is well defined note that we have the logical statements
	\begin{align*}
		&
			\tau>\left(\GMC{M}-\sqrt{5}\left(\sqrt{2\frac{S}{M}\Ln{\ExpE\frac{N}{S}}}+\sqrt{\frac{S}{M}}\right)\right)^{-1}>0
		\\\Rightarrow\hspace{5pt}&
			\frac{\GMC{M}-\tau^{-1}}{\sqrt{2\frac{S}{M}\Ln{\ExpE\frac{N}{S}}}+\sqrt{\frac{S}{M}}}\in\left(\sqrt{5},\infty\right)
		\\\Leftrightarrow\hspace{5pt}&
			\rho\left(\tau\right)=\left(\sqrt{\left(\frac{\GMC{M}-\tau^{-1}}{\sqrt{2\frac{S}{M}\Ln{\ExpE\frac{N}{S}}}+\sqrt{\frac{S}{M}}}\right)^2-1}-1\right)^{-1}\in\left(0,1\right).
	\end{align*}
	We set
	\begin{align*}
		\eta:=\Exp{-\frac{1}{2}\left(\sqrt{1+\left(1+\rho\left(\tau\right)^{-1}\right)^2}-\sqrt{1+\left(1+\rho^{-1}\right)^2}\right)^2\left(\sqrt{2\frac{S}{M}\Ln{\ExpE\frac{N}{S}}}+\sqrt{\frac{S}{M}}\right)^2M},
	\end{align*}
	which obeys $\eta\in\left(0,1\right)$ since $\rho>\rho\left(\tau\right)$.
	It follows that
	\begin{align*}
		&
			\sqrt{1+\left(1+\rho^{-1}\right)^2}\left(\sqrt{2S\Ln{\ExpE\frac{N}{S}}}+\sqrt{S}\right)
			+\sqrt{2\Ln{\eta^{-1}}}
		\\=&
			\sqrt{1+\left(1+\rho^{-1}\right)^2}\left(\sqrt{2S\Ln{\ExpE\frac{N}{S}}}+\sqrt{S}\right)
		\\&+
			\left(\sqrt{1+\left(1+\rho\left(\tau\right)^{-1}\right)^2}-\sqrt{1+\left(1+\rho^{-1}\right)^2}\right)\left(\sqrt{2\frac{S}{M}\Ln{\ExpE\frac{N}{S}}}+\sqrt{\frac{S}{M}}\right)\sqrt{M}
		\\=&
			\sqrt{1+\left(1+\rho\left(\tau\right)^{-1}\right)^2}\left(\sqrt{2S\Ln{\ExpE\frac{N}{S}}}+\sqrt{S}\right).
	\end{align*}
	If we plug in the definition of $\rho\left(\tau\right)$ into this, we obtain
	\begin{align*}
			\sqrt{1+\left(1+\rho^{-1}\right)^2}\left(\sqrt{2S\Ln{\ExpE\frac{N}{S}}}+\sqrt{S}\right)
			+\sqrt{2\Ln{\eta^{-1}}}
		=&
			\frac{\GMC{M}-\tau^{-1}}{\sqrt{2\frac{S}{M}\Ln{\ExpE\frac{N}{S}}}+\sqrt{\frac{S}{M}}}\left(\sqrt{2S\Ln{\ExpE\frac{N}{S}}}+\sqrt{S}\right)
		\\=&
			\left(\GMC{M}-\tau^{-1}\right)\sqrt{M}.
	\end{align*}
	Together with $\tau>0$ from \refP{Equation:PhaseTransition:Corollary:Gaussian=>NSP:no_eta}
	this yields that \refP{Equation:PhaseTransition:Theorem:Gaussian=>NSP} holds true.
	The proof now follows from \thref{Theorem:Gaussian=>NSP}.
\end{proof}
By a certain choice of $\rho$ in this result and a limit argument we can deduce \thref{Proposition:Gaussian=>HRL1} from this.
\begin{proof}[Proof of \thref{Proposition:Gaussian=>HRL1}]
	Given $\alpha\in\left(0,1\right)$ and $\tau\in\left(\left(\GMC{M}-\sqrt{5}\left(\sqrt{2\frac{S}{M}\Ln{\ExpE\frac{N}{S}}}+\sqrt{\frac{S}{M}}\right)\right)^{-1},\hrloneParam S^{-\frac{1}{2}}\right)$
	we set
	\begin{align*}
		\tilde{\rho}\left(\alpha,\tau\right):=\left(\sqrt{\left(\left(1-\alpha\right)\sqrt{1+\left(1+\rho\left(\tau\right)^{-1}\right)^2}+\alpha\sqrt{5}\right)^2-1}-1\right)^{-1}
	\end{align*}
	with $\rho\left(\tau\right)$ from \thref{Corollary:Gaussian=>NSP:no_eta}.
	By \thref{Corollary:Gaussian=>NSP:no_eta} we have $\rho\left(\tau\right)<1$ which we can use
	once for each bound to obtain that
	\begin{align*}
			\tilde{\rho}\left(\alpha,\tau\right)
		<&
			\left(\sqrt{\left(\left(1-\alpha\right)\sqrt{1+\left(1+1^{-1}\right)^2}+\alpha\sqrt{5}\right)^2-1}-1\right)^{-1}
			=1\TextAnd
		\\
			\tilde{\rho}\left(\alpha,\tau\right)
		>&
			\left(\sqrt{\left(\left(1-\alpha\right)\sqrt{1+\left(1+\rho\left(\tau\right)^{-1}\right)^2}+\alpha\sqrt{1+\left(1+\rho\left(\tau\right)^{-1}\right)^2}\right)^2-1}-1\right)^{-1}
		=\rho\left(\tau\right),
	\end{align*}
	i.e., the necessary requirement $\tilde{\rho}\left(\alpha,\tau\right)\in\left(\rho\left(\tau\right),1\right)$.
	By the definition of $\tilde{\rho}\left(\alpha,\tau\right)$ we further get that
	\begin{align*}
		\sqrt{1+\left(1+\tilde{\rho}\left(\alpha,\tau\right)^{-1}\right)^2}=\left(1-\alpha\right)\sqrt{1+\left(1+\rho\left(\tau\right)^{-1}\right)^2}+\alpha\sqrt{5}.
	\end{align*}
	Using this and $\sqrt{1+\left(1+\rho\left(\tau\right)^{-1}\right)^2}=\frac{\GMC{M}-\tau^{-1}}{\sqrt{2\frac{S}{M}\Ln{\ExpE\frac{N}{S}}}+\sqrt{\frac{S}{M}}}$
	yields
	\begin{align*}
		&
			\left(\sqrt{1+\left(1+\rho\left(\tau\right)^{-1}\right)^2}-\sqrt{1+\left(1+\tilde{\rho}\left(\alpha,\tau\right)^{-1}\right)^2}\right)^2\left(\sqrt{2\frac{S}{M}\Ln{\ExpE\frac{N}{S}}}+\sqrt{\frac{S}{M}}\right)^2
		\\=&
			\alpha^2\left(\sqrt{1+\left(1+\rho\left(\tau\right)^{-1}\right)^2}-\sqrt{5}\right)^2\left(\sqrt{2\frac{S}{M}\Ln{\ExpE\frac{N}{S}}}+\sqrt{\frac{S}{M}}\right)^2
		\\=&
			\alpha^2\left(\GMC{M}-\sqrt{5}\left(\sqrt{2\frac{S}{M}\Ln{\ExpE\frac{N}{S}}}+\sqrt{\frac{S}{M}}\right)-\tau^{-1}\right)^2.
	\end{align*}
	Using this and \thref{Corollary:Gaussian=>NSP:no_eta} we obtain
	for every $\alpha\in\left(0,1\right)$ and every
	\begin{align*}
		\tau\in\left(\left(\GMC{M}-\sqrt{5}\left(\sqrt{2\frac{S}{M}\Ln{\ExpE\frac{N}{S}}}+\frac{S}{M}\right)\right)^{-1},\hrloneParam S^{-\frac{1}{2}}\right)
	\end{align*}
	that
	\begin{align}
		\nonumber&
			\Prob{\text{$\AMat$ has $\ell_2$-RNSP of order $S$ wrt $\norm{\cdot}_2$ with constants
				$\tilde{\rho}\left(\alpha,\tau\right)$ and $\tau$}}
		\\\geq&\label{Equation:prob_calc:Theorem:Gaussian=>NSP:HRL1}
			1-\Exp{-\frac{1}{2}\alpha^2\left(\GMC{M}-\sqrt{5}\left(\sqrt{2\frac{S}{M}\Ln{\ExpE\frac{N}{S}}}+\sqrt{\frac{S}{M}}\right)-\tau^{-1}\right)^2M}
	\end{align}
	holds true.
	It follows that
	\begin{align}
		\nonumber&
			\Prob{\text{$\AMat$ has $\ell_2$-NSP of order $S$ and $\hrloneParam>\tau_2^0 S^\frac{1}{2}$}}
		\\\nonumber=&
			\Prob{\text{$\AMat$ has $\ell_2$-NSP of order $S$ and $\tau_2^0<\hrloneParam S^{-\frac{1}{2}}$}}
		\\\geq &\label{Equation:EQ1:Theorem:Gaussian=>NSP:HRL1}
			\Prob{\text{$\AMat$ has $\ell_2$-ORNSP of order $S$ wrt $\norm{\cdot}_2$ with constant $\hrloneParam S^{-\frac{1}{2}}$}}
		\\\geq&\label{Equation:EQ2:Theorem:Gaussian=>NSP:HRL1}
			\Prob{\text{$\AMat$ has $\ell_2$-RNSP of order $S$ wrt $\norm{\cdot}_2$ with constants
	$\tilde{\rho}\left(\alpha,\tau\right)$ and $\tau$}}
		\\\geq&\label{Equation:EQ3:Theorem:Gaussian=>NSP:HRL1}
			1-\Exp{-\frac{1}{2}\alpha^2\left(\GMC{M}-\sqrt{5}\left(\sqrt{2\frac{S}{M}\Ln{\ExpE\frac{N}{S}}}+\sqrt{\frac{S}{M}}\right)-\tau^{-1}\right)^2M},
	\end{align}
	where
	\begin{itemize}
		\item[]
			\refP{Equation:EQ1:Theorem:Gaussian=>NSP:HRL1} follows from \thref{Lemma:NSP<=>RNSP} and
			\thref{Corollary:ShapeSNSPAndRNSP},
		\item[]
			\refP{Equation:EQ2:Theorem:Gaussian=>NSP:HRL1} follows from \thref{Lemma:RNSP<=>SRNSP}
			and $\hrloneParam S^{-\frac{1}{2}}>\tau$ and
		\item[]
			\refP{Equation:EQ3:Theorem:Gaussian=>NSP:HRL1} follows from
			\refP{Equation:prob_calc:Theorem:Gaussian=>NSP:HRL1}.
	\end{itemize}
	The function in \refP{Equation:EQ3:Theorem:Gaussian=>NSP:HRL1} is continuous in $\left(\alpha,\tau\right)$
	and	we can send $\tau\rightarrow \hrloneParam S^{-\frac{1}{2}}$ and $\alpha\rightarrow 1$ to obtain
	\begin{align*}
		&
			\Prob{\text{$\AMat$ has $\ell_2$-NSP of order $S$ and $\hrloneParam>\tau_2^0 S^\frac{1}{2}$}}
		\\\geq&
			1-\Exp{-\frac{1}{2}\left(\GMC{M}-\sqrt{5}\left(\sqrt{2\frac{S}{M}\Ln{\ExpE\frac{N}{S}}}+\sqrt{\frac{S}{M}}\right)-\sqrt{S}\hrloneParam^{-1}\right)^2M}.
	\end{align*}
	In this case, \thref{Corollary:ShapeHRL1Param} yields
	$\hrloneParam>\tau_2^0 \sqrt{S}\geq \tau_1^0$ and $\hrloneParam\in\left(\tau_1^0,\infty\right)$
	and thus also that $\Decoder{}$ is an $\ell_2$-SRD of order $S$ wrt $\norm{\cdot}_2$ for $\AMat$.
	Alternatively, one could also use \thref{Corollary:ShapeSNSPAndRNSP}
	to obtain that $\AMat$ has $\ell_2$-ORNSP of order $S$ wrt $\norm{\cdot}_2$ with constant
	$\frac{1}{2}\left(\hrloneParam S^{-\frac{1}{2}}+\tau_2^0\right)$,
	\thref{Lemma:RNSP<=>SRNSP} to obtain that $\AMat$ has $\ell_2$-RNSP of order $S$ wrt $\norm{\cdot}_2$
	with some stableness constant and robustness constant $\frac{1}{2}\left(\hrloneParam S^{-\frac{1}{2}}+\tau_2^0\right)$
	and \thref{Theorem:HRL1Recovery:Lambda>TauSq} to get the claim.
\end{proof}
\subsection{Proofs of Subsection \ref{Subsection:NotNPHard_LRBGMatrices} Random Walk Matrices of Uniformly Distributed $D$-Left Regular Bipartite Graphs}\label{Subsection:NotNPHard_LRBGMatrices:Proofs}
\begin{Definition}\label{Definition:LosslessExpander}
	Let $S\in\SetOf{N}$ and $D^{-1}\AMat\in\left\{0,D^{-1}\right\}^{M\times N}$ be a $D$-left regular bipartite graph.
	If additionally there exists a $\theta\in\left[0,1\right)$ such that
	\begin{align*}
		\SetSize{\RightVertices{T}}\geq \left(1-\theta\right)D\SetSize{T} \TextForAll \SetSize{T}\leq S
	\end{align*}
	holds true, then $D^{-1}\AMat$ is called a random walk matrix of an
	$\left(S,D,\theta\right)$-lossless expander.
\end{Definition}
If we draw a $D$-left regular bipartite graph uniformly at random, it will be a random walk matrix
of an $\left(S,D,\theta\right)$-lossless expander with high probability.
\begin{Proposition}[ {\cite[Corollary~13.7]{IntroductionCS}} ]\label{Proposition:LeftRegularBipartiteGraphsAreLosslessExpanders}
	Let $S\in\SetOf{N}$, $\theta\in\left(0,1\right)$,
	$D:=\RoundUp{\frac{2}{\theta}\Ln{\frac{\ExpE N}{S}}}$
	and $\AMat\in\left\{0,D^{-1}\right\}^{M\times N}$ be a uniformly at random chosen $D$-left regular bipartite graph.
	If
	\begin{align*}
		M\geq \frac{2}{\theta}\Exp{\frac{2}{\theta}}S\Ln{\frac{\ExpE N}{S}},
	\end{align*}
	then $\AMat$ is the random walk matrix of an
	$\left(S,D,\theta\right)$-lossless expander with probability of at least
	$
		1-\frac{S}{\ExpE N}
	$.
\end{Proposition}
This statement is proven in \cite[Corollary~13.7]{IntroductionCS}.
Further, the random walk matrix of a $\left(2S,D,\theta\right)$-lossless expander has $\ell_1$-RNSP of order $S$.
\begin{Lemma}[ {\cite[Theorem~13.1]{IntroductionCS}} ]\label{Lemma:LosslessExpandersHaveSRNSP}
	Let $2S\in\SetOf{N}$, $D\in\SetOf{M}$ and $\theta\in\left[0,\frac{1}{6}\right)$.
	Let $\AMat\in\left\{0,D^{-1}\right\}^{M\times N}$ be the random walk matrix of a
	$\left(2S,D,\theta\right)$-lossless expander graph.
	Then $\AMat$ has
	$\ell_1$-RNSP of order $S$ wrt $\norm{\cdot}_1$ with constants
	$\rho=\frac{2\theta}{1-4\theta}$ and $\tau=\frac{1}{\left(1-4\theta\right)}$.
\end{Lemma}
\begin{proof}
	Note that \cite[Theorem~13.1]{IntroductionCS} yields that the matrix
	$D\AMat\in\left\{0,1\right\}^{M\times N}$
	has $\ell_1$-RNSP of order $S$ wrt $\norm{\cdot}_1$ with constants
	$\rho=\frac{2\theta}{1-4\theta}$ and $\tau'=\frac{1}{\left(1-4\theta\right)D}$.
	Rescalling yields the statement.
\end{proof}
We are now able to recover signals from measurements matrices chosen uniformly at random
among all $D$-left regular bipartite graph.
\begin{proof}[Proof of \thref{Theorem:RecoveryViaExpander}]
	By \thref{Proposition:LeftRegularBipartiteGraphsAreLosslessExpanders} $\AMat$ is the
	random walk matrix of a $\left(2S,D,\theta\right)$-lossless expander
	with the given probability.
	In this case, by \thref{Lemma:LosslessExpandersHaveSRNSP} $\AMat$ has
	$\ell_1$-RNSP of order $S$ wrt $\norm{\cdot}_1$ with constants
	$\rho=\frac{2\theta}{1-4\theta}$ and $\tau=\frac{1}{1-4\theta}$.
	Since $\theta\in\left(0,\frac{1}{6}\right)$, we have that
	$\hrloneParam=\frac{2}{1+\rho}\tau=\frac{2}{1-2\theta}\in\left(2,3\right)$.
	The statement now follows from \thref{Theorem:HRL1Recovery:Lambda>TauSq}.
\end{proof}

\section*{Acknowledgements}
The work was partially supported by DAAD grant 57417688.
PJ has been supported by DFG grant JU 2795/3.
We are thankful for the comments and input of Felix Krahmer.
	\bibliography{Bibliography/Bibliography}
	\bibliographystyle{IEEEtran}
\end{document}